\documentclass[11pt]{article}
\usepackage[utf8]{inputenc}
\usepackage{amsmath, amssymb, amsthm}
\usepackage{graphicx}
\usepackage{float}
\usepackage{algorithm}
\usepackage{algorithmic}
\usepackage{geometry}
\usepackage{booktabs}
\usepackage{multirow}
\usepackage{makecell}
\usepackage{tikz}
\usepackage{forest}
\usepackage{lipsum}
\usepackage{subcaption}
\usepackage{enumitem}
\usepackage{xcolor}
\usepackage{listings}
\usepackage{caption}
\usepackage{afterpage}
\usepackage{parskip}
\usepackage{threeparttable}
\usepackage{hyperref}
\usepackage{mathtools}
\usepackage[T1]{fontenc}
\usepackage{lmodern}
\usepackage[numbers]{natbib}
\hypersetup{
	breaklinks=true,
	colorlinks=true,
	linkcolor=blue,
	citecolor=blue,
	filecolor=black,
	urlcolor=blue,
	pdfborder={0 0 0}
}
\let\oldref\ref
\renewcommand{\ref}[1]{\textcolor{blue}{\oldref{#1}}}

\usepackage{titlesec}
\usepackage{tikz-cd}

\titlespacing*{\section}{0pt}{1.5\baselineskip}{1\baselineskip}
\titlespacing*{\subsection}{0pt}{1.2\baselineskip}{0.8\baselineskip}

\DeclareMathOperator{\tdw}{tdw}
\theoremstyle{plain}

\newtheorem{theorem}{Theorem}[section]
\newtheorem{lemma}[theorem]{Lemma}
\newtheorem{corollary}[theorem]{Corollary}

\newtheorem{hypothesis}[theorem]{Hypothesis}

\theoremstyle{definition}
\newtheorem{example}[theorem]{Example}
\newtheorem{problem}[theorem]{Problem}
\newtheorem{definition2}[theorem]{Definition}
\newtheorem{remark}[theorem]{Remark}

\titleformat{\subsubsection}[runin]
{\normalfont\bfseries}
{\thesubsubsection.}
{0.5em}
{}
[.]
\title{On the Structure of $(\min,+)$ Convolution}
\author{Huanyi Zhou\thanks{Tsinghua University. \url{zhouhuan24@mails.tsinghua.edu.cn.}}}
\date{}

\begin{document}
	\maketitle
	
	\begin{abstract}
		The $(\min,+)$ convolution is a central problem in fine-grained complexity, and it remains open whether it can be computed in truly subquadratic time. We study this problem through the algebraic structure of tropical polynomials, where $(\min,+)$ convolution is exactly tropical polynomial multiplication.
		
		We introduce the \textit{tropical decomposition width}, denoted by $\operatorname{tdw}(A)$, which measures how finely a tropical polynomial can be decomposed into factors of small degree. We prove two modular convexity theorems showing that bounded tropical decomposition width forces a convexity structure on arithmetic progression subpolynomials. As algorithmic consequences, we obtain deterministic algorithms for computing $a\otimes b$ in time
		$$O\left(n\max(\operatorname{tdw}(a),\operatorname{tdw}(b))^2\right)$$
		when the value of $\max(\operatorname{tdw}(a),\operatorname{tdw}(b))$ is given, and in
		$$O\left(ne^{\min(\operatorname{tdw}(a),\operatorname{tdw}(b))(1+o(1))}\right)$$
		time without prior knowledge about tropical decomposition width in advance. Neither algorithm requires a decomposition of the input sequences.
		
		The same structural ideas yield a randomized algorithm for Multiple-Sequence $(\min,+)$ Convolution. Given $k$ sequences of length at most $n$, their convolution can be computed in
		$$O\left(kn^2\sqrt{\min(k,n)}\log^{1.5}(kn)\right)$$
		time, improving over the natural $O(k^2n^2)$ bound. We complement this result with conditional lower bounds and a faster algorithm for the single-entry version of the problem. Through reductions to Multiple-Sequence $(\min,+)$ Convolution, these algorithms also give new upper bounds for Multiple-Choice Knapsack in the $(n,w_{\max})$ parameterization, improving the previous bound $\widetilde O(n+w_{\max}^4)$ (Jakub Pawlewicz 2026).
		
		Finally, we investigate two classical algebraic approaches to polynomial multiplication. We introduce interpolation algebras for tropical polynomials and show that classes of tropical polynomials with bounded tropical decomposition width admit interpolation algebras of finite generating rank. In contrast, distinguishing all tropical polynomials of degree at most $n$ requires generating rank $\lfloor n/2\rfloor+1$. We also prove that tropical decomposition width cannot decrease after passing to any flat $\mathbb T$-algebra extension. Together, these results relate the algorithmic tractability of $(\min,+)$ convolution to structural rigidity in tropical polynomial multiplication.
	\end{abstract}
	\newpage
	\tableofcontents
	\newpage
	\section{Introduction}
	The $(\min,+)$ convolution $c=a\otimes b$ of two sequences $a=(a_0,\ldots,a_n)$ and $b=(b_0,\ldots,b_n)$ is the sequence $c=(c_0,\ldots,c_{2n})$ defined by
	$$
	c_i=\min_{j+k=i}(a_j+b_k).
	$$
	
	Over the past decade, $(\min,+)$ convolution has become a central problem in fine-grained complexity theory. The core open question is whether it admits a truly subquadratic-time algorithm, that is, an algorithm running in $O(n^{2-\varepsilon})$ time for some constant $\varepsilon>0$. Currently, the fastest known algorithm runs in $n^2/2^{\Omega(\sqrt{\log n})}$ time \cite{wil14,cw16}, which falls short of truly subquadratic.
	
	Research on $(\min,+)$ convolution has followed three main directions. First, many works \cite{bcd06,lrc14,kps17,cmww19} treat it as a hardness assumption, deriving conditional lower bounds for various combinatorial problems via the $(\min,+)$ convolution hypothesis. In particular, they show that this hypothesis implies the APSP hypothesis and the 3SUM hypothesis, both of which are central to fine-grained complexity.
	
	Second, in light of the difficulty of obtaining a truly subquadratic algorithm, several approximation algorithms have been developed. Notably, \cite{mww19} gave a $(1+\epsilon)$-approximation algorithm running in $\widetilde O(n/\epsilon\log W)$ time for integer inputs bounded by $W$, while \cite{bkw19} obtained a strongly polynomial $(1+\epsilon)$-approximation algorithm running in $\widetilde O(n^{1.5}/\epsilon^{0.5})$ time.
	
	Third, faster exact algorithms exist for structured instances. If one sequence is convex, $(\min,+)$ convolution can be computed in $O(n)$ time using the SMAWK algorithm \cite{akmsw87}; recent work \cite{bc23,bl23} extended this to the nearly-convex setting. Other works \cite{cl15,cdx22,cdxz22,jpsx26} studied the bounded-difference case and the monotone case, where the convolution can be computed in $\widetilde O(n^{1.5})$ time.
	
	Despite this extensive body of work, the existence of a truly subquadratic algorithm for general $(\min,+)$ convolution remains open. Nor do we have an unconditional truly superlinear lower bound. This stands in sharp contrast to classical $(+,\times)$ convolution, which can be computed in $O(n\log n)$ time via the Fast Fourier Transform \cite{ct65}. This contrast raises a natural question:
	\begin{center}
		\textit{Why is $(\min,+)$ convolution so difficult, while classical $(+,\times)$ convolution is so easy, even though the two operations are defined analogously?}
	\end{center}
	In this paper, we approach this question indirectly, drawing inspiration from the field of tropical geometry \cite{ms15}. We begin with a structural transformation: view sequences as polynomials in the $(\min,+)$ semiring, called tropical polynomials. Under this transformation, the $(\min,+)$ convolution of two sequences becomes the multiplication of two tropical polynomials.
	
	In classical polynomial algebra, several properties, such as the polynomial identity property and the unique factorization property, play fundamental roles over fields and commutative rings. Related factorization phenomena have appeared in tropical and idempotent algebra in several different settings and under different formulations \cite{cm80,bcoq92,agt25,pon26}.
	
	Motivated by $(\min,+)$ convolution, we revisit a collection of classical polynomial properties and examine their behavior on formal tropical polynomials. We find that none of these properties holds for general polynomials in $\mathbb T[x]$. Interestingly, however, when we restrict to classes of $\mathbb T[x]$ that partially recover these properties, those classes support efficient computation. This offers high-level indirect evidence for why a truly subquadratic algorithm for $(\min,+)$ convolution remains elusive: $(\min,+)$ convolution lacks the nice properties that facilitate efficient computation.
	
	Inspired by this perspective, we introduce and study these properties for tropical polynomials, and define several structural parameters that characterize their behavior. In particular, we give faster algorithms for $(\min,+)$ convolution on classes of polynomials that partially satisfy these properties. Surprisingly, the same technique yields a faster algorithm for the Multiple-Sequence $(\min,+)$ Convolution, demonstrating the counterintuitive fact that although two-sequence $(\min,+)$ convolution appears resistant to improvement, the multi-sequence version admits a faster algorithm.
	\subsection{Fragmented Properties for Tropical Polynomials}
	We begin by examining three basic properties in polynomial algebra:
	\begin{enumerate}
		\item[(1)] Polynomial Identity Property: every polynomial is uniquely determined by its evaluation values on the base algebra.
		\item[(2)] Linear Factorization Property: every polynomial factors completely into linear factors.
		\item[(3)] Unique Factorization Property: every polynomial admits an essentially unique decomposition into irreducible factors, up to order and non‑zero constant multiples.
	\end{enumerate}
	Polynomials over the classical $(\mathbb C,+,\cdot)$ algebra satisfy all three properties. In the tropical $(\min,+)$ setting, however, these properties become fragmented and no longer hold for arbitrary tropical polynomials. In fact, among the classes considered here, we prove that the convex polynomial set, for which $(\min,+)$ convolution can be computed efficiently, satisfies all three properties (Lemma \ref{lem:poly_con_lf}, Lemma \ref{lem:poly_con_ip}, and Lemma \ref{lem:poly_con_uf}).
	
	Since these properties are too restrictive in the tropical setting, it is natural to consider relaxed versions that still retain enough structure to support efficient computation. We therefore consider the following two natural relaxations:
	\begin{enumerate}
		\item[(4)] Extended Identity Property: every polynomial is uniquely determined by its evaluation values on a finitely generated extension algebra over the base algebra.
		\item[(5)] Bounded Factorization Property: every polynomial admits a decomposition into factors of bounded degree.
	\end{enumerate}
	These properties are defined on multiplicatively closed sets of tropical polynomials rather than on individual polynomials, which ensures that the definitions are well-posed and that the sets are closed under multiplication. The formal definitions are given at the beginning of Section \ref{sec:char}.
	
	We examine the relations among these five basic properties together with the convexity and the convex polynomial set. They are summarized in the diagram below. An arrow $A\to B$ indicates that every polynomial set $\mathcal S\subseteq\mathbb T[x]$ satisfying property $A$ also satisfies property $B$, and that the implication is strict: there exists a set satisfying $B$ but not $A$. A double arrow $A\leftrightarrow B$ indicates that the two properties are equivalent for every polynomial set $\mathcal S\subseteq\mathbb T[x]$.
	$$
	\begin{tikzcd}[column sep=small]
		& \text{Unique Factorization Property} & & \\
		\text{Convex Polynomial Set} \arrow[ru,"\text{Lemma \ref{lem:poly_con_uf}}"]
		\arrow[r,"\text{trivial}"]
		& 
		\text{Convex Property} \arrow[r,leftrightarrow,"\text{Lemma \ref{lem:poly_con_lf}}"] 
		\arrow[d,"\text{Lemma \ref{lem:poly_con_ip}}"] 
		& 
		\text{Linear Factorization Property} \arrow[d,"\text{trivial}"] & \\
		& \text{Polynomial Identity Property} \arrow[rd,"\text{trivial}"] & 
		\text{Bounded Factorization Property} \arrow[d,"\text{Corollary \ref{cor:poly_bf_ei}}"] & \\
		& & \text{Extended Identity Property} &
	\end{tikzcd}
	$$
	In particular, we show that $\mathbb T[x]$ satisfies neither the extended identity property nor the unique factorization property (Corollary \ref{cor:poly_ei_uf}).
	
	Among these properties, we focus mainly on the bounded factorization property and the extended identity property. Motivated by the bounded factorization property, we introduce the parameter of tropical decomposition width, which reflects the rigidity of a tropical polynomial under decomposition. For a tropical polynomial $A$, its tropical decomposition width is the minimum possible value of
	$$
	\max_{i=1}^{N}\left(\deg B^{(i)}\right)
	$$
	over all decompositions
	$$
	A=\bigotimes_{i=1}^{N}B^{(i)}.
	$$
	Here $\bigotimes$ denotes multiplication in the semiring setting.
	
	We study the structural properties of this parameter in Section \ref{sec:char}, and in Section \ref{sec:alg} we design faster algorithms for $(\min,+)$ convolution-type problems based on this parameter.
	
	Motivated by the extended identity property, we introduce the notion of an interpolation algebra, which extends the classical evaluation-interpolation techniques for polynomials over fields to polynomials over
	commutative semirings.
	
	Necessary lemmas and theorems for interpolation algebras are established in Section \ref{sec:char}, and the power of interpolation algebras is investigated in Section \ref{sec:ext}.
	\subsection{Structural Properties of Bounded Tropical Decomposition Width Polynomials}
	We first prove two core structural theorems for tropical polynomials with bounded tropical decomposition width. These theorems show that bounded factorization imposes convexity properties on arithmetic progression subpolynomials--a surprising bridge between convexity and modularity. The proofs proceed by reducing the problem to an additive combinatorics statement of our own formulation, for which we then prove a quantitative bound.
	\begin{theorem}[the first modular convexity theorem]
		\label{thm:fir_con}
		For any tropical polynomial $A$, for every positive integer $d$ and
		every subpolynomial $B\in \operatorname{MS}_d(A)$, we have
		$
		\operatorname{cgap}(B)\leq \tdw(A).
		$
	\end{theorem}
	Here $\operatorname{MS}_d(A)$ denotes the mod-$d$ arithmetic progression
	subpolynomial multiset of $A$, and $\operatorname{cgap}(B)$ denotes the
	convex gap of $B$. Additionally, we provide the structural properties of the convex gap in Section \ref{sec:cgap}. Although this concept is defined combinatorially, it has an intuitive geometric interpretation.
	\begin{theorem}[the second modular convexity theorem]
		\label{thm:sec_con}
		Let $k,L$ be positive integers. If
		$
		\operatorname{lcm}_{i=1}^{k}i\mid L,
		$
		then for every tropical polynomial $A$ with $\tdw(A)\leq k$ and
		every subpolynomial $B\in \operatorname{MS}_L(A)$, the
		polynomial $B$ is convex. Conversely, if
		$
		\operatorname{lcm}_{i=1}^{k}i\nmid L,
		$
		then there exists a tropical polynomial $A$ with
		$\operatorname{tdw}(A)\leq k$ such that some
		$B\in \operatorname{MS}_L(A)$ is not convex.
	\end{theorem}
	The second modular convexity theorem will also serve as a key ingredient in our study of interpolation algebras (Theorem \ref{thm:int_alg_intro}). This connection gives an algebraic interpretation of the modular convexity phenomenon.
	
	Although computing $\operatorname{tdw}(A)$ is NP-hard (Lemma \ref{lem:tdw_np_hard}), these structural theorems allow us to obtain two algorithms for $(\min,+)$ convolution. Both algorithms are significantly faster than the quadratic brute-force approach when the tropical decomposition width is small. Notably, they do not require the decomposition itself to be provided as input. It suffices to have a promise that the tropical decomposition width is small.
	
	The first algorithm uses an adjustment lemma and has complexity quadratic in the maximum tropical decomposition width of the two sequences.
	\begin{theorem}
		\label{thm:fir_alg}
		Let $a,b$ be two tropical sequences with $|a|,|b|\leq n$. If
		$
		\max(\tdw(a),\tdw(b))
		$
		is given as part of the input, then there is a deterministic algorithm
		that computes
		$
		a\otimes b
		$
		in time
		$$
		O\left(n\max(\tdw(a),\tdw(b))^2\right).
		$$
	\end{theorem}
	
	The second algorithm does not require the tropical decomposition width to be given; it only assumes that one of the two sequences has small tropical decomposition width. However, its running time is exponential in that width.
	\begin{theorem}
		\label{thm:sec_alg}
		Let $a,b$ be two tropical sequences with $|a|,|b|\leq n$. There is a
		deterministic algorithm that computes
		$
		a\otimes b
		$
		in time
		$$
		O\left(
		n e^{\min(\tdw(a),\tdw(b))(1+o(1))}
		\right).
		$$
	\end{theorem}
	\subsection{Multiple-Sequence $(\min,+)$ Convolution Problem}
	A natural problem arising from tropical decomposition width is the $(\min,+)$ convolution of multiple sequences. Given $k$ tropical sequences $a^{(1)},\ldots,a^{(k)}$ with $|a^{(i)}|\le n$, we wish to compute
	$$
	\bigotimes_{i=1}^{k} a^{(i)}.
	$$
	If we compute the convolution sequentially or by divide-and-conquer, the intermediate sequence length grows after each step: the convolution of $t$ sequences may have length $\Omega(tn)$. Consequently, such algorithms have running time $O(k^2n^2)$.
	
	Although we cannot improve the brute-force algorithm for two-sequence $(\min,+)$ convolution in the exponent, we can improve on the brute-force exponent for the multiple-sequence version using a combinatorial algorithm. A key observation is the following: every intermediate sequence expressible as a convolution of a subset of the input sequences $a^{(1)},\dots,a^{(k)}$ has tropical decomposition width at most $n$. This suggests designing an algorithm based on the structural properties of bounded tropical decomposition width polynomials.
	
	By applying the approach of Theorem \ref{thm:fir_alg}, together with random-shuffle and isolation techniques, we obtain a randomized algorithm with running time
	$$
	O\left(kn^2\sqrt{\min(k,n)}\log^{1.5}(kn)\right),
	$$
	which has total exponent $3.5$, improving upon the brute-force exponent when $k=\Theta(n)$.
	\begin{theorem}
		The all-entry version of Multiple-Sequence $(\min,+)$ Convolution can be solved by a randomized algorithm in time
		$$
		O\left(kn^2\sqrt{\min(k,n)}\log^{1.5} (kn)\right)
		$$
		with high probability.
	\end{theorem}
	
	We also study the single-entry version of the Multiple-Sequence $(\min,+)$ Convolution problem. This version admits a faster algorithm than the all-entry version when $k\gg n$.
	\begin{theorem}
		\label{thm:mult_sing}
		The single-entry version of Multiple-Sequence $(\min,+)$ Convolution can be solved by a randomized algorithm in time
		$$
		O\left(\left(kn+n^3\sqrt{\min(k,n)}\right)\log(kn)\right)
		$$
		with high probability.
	\end{theorem}
	
	Moreover, we give a conditional lower bound for the all-entry version of the Multiple-Sequence $(\min,+)$ Convolution problem, under the $(\min,+)$ convolution hypothesis. For the single-entry version, since the $0$-$1$ Knapsack problem reduces to this problem (as we will show in Section \ref{sec:con}), the conditional lower bound $(k+n)^{2-o(1)}$ for $0$-$1$ Knapsack \cite{cmww19} transfers to the single-entry version.
	\begin{theorem}
		Assuming the $(\min,+)$ convolution hypothesis, for any fixed $\epsilon>0$, the all-entry version of Multiple-Sequence $(\min,+)$ Convolution
		cannot be solved by a randomized algorithm in time $$O(kn^{2-\epsilon}+k^{1-\epsilon}n^2),$$
		even when the sequence entries are restricted to $\{0,1,\ldots,n^{O(1)}\}$.
	\end{theorem}
	
	Finally, we observe that several Knapsack-type problems reduce to the Multiple-Sequence $(\min,+)$ Convolution problem (Section \ref{sec:con}), including Subset Sum, $0$-$1$ Knapsack and Multiple-Choice Knapsack \cite{nau78,sz79,ks26}. These problems form a natural hardness hierarchy:
	$$
	\text{Subset Sum} \longrightarrow \text{$0$-$1$ Knapsack} \longrightarrow \text{Multiple-Choice Knapsack}.
	$$
	
	Recent improvements for Subset Sum \cite{prw21,clmz24a}, $0$-$1$ Knapsack \cite{jin24,bri24}, and Multiple-Choice Knapsack \cite{paw26} have been obtained in the $(n,w_{\max})$ parameterization, where $n$ is the number of items and $w_{\max}$ is the maximum item weight. Using our Multiple-Sequence $(\min,+)$ Convolution algorithm, we improve the previous upper bound $\widetilde O(n+w_{\max}^4)$ for Multiple-Choice Knapsack from \cite{paw26}. To facilitate comparison of the best known algorithms for Subset Sum, $0$-$1$ Knapsack, and Multiple-Choice Knapsack in the $(n,w_{\max})$ parameter regime, we summarize the currently best known algorithms together with our new results in Table \ref{tab:prob}.
	\begin{table}[H]
		\centering
		\caption{Time complexities of selected Knapsack-type problems (logarithmic factors omitted)}
		\label{tab:prob}
		\scalebox{0.73}{
			\begin{tabular}{|l|c|c|}
				\hline
				Problem & Single-entry & All-entry \\
				\hline
				Subset Sum
				&
				$\tilde{O}(n+w_{\max}^{1.5})$ \cite{clmz24a},
				$\tilde{O}(n+n^{0.5}w_{\max})$ \cite{clmz24b}
				&
				$\tilde{O}(nw_{\max})^{\ast}$ \\
				\hline
				$0$-$1$ Knapsack
				&
				$\tilde{O}(n+w_{\max}^2)$ \cite{jin24,bri24},
				$\tilde{O}(n^{1.5}w_{\max})$ \cite{hx24}
				&
				$\tilde{O}(nw_{\max}^2)$ \cite{at19},
				$\tilde{O}(n^2w_{\max})$ \cite{bel52} \\
				\hline
				Multiple-Choice Knapsack
				&
				$\tilde{O}(n+w_{\max}^{3.5})^{\dagger}$,
				$\tilde{O}(nw_{\max}^{1.5})^{\dagger}$,
				$\tilde{O}(n^{1.5}w_{\max})^{\dagger}$
				&
				$\tilde{O}(nw_{\max}^{2.5})^{\dagger}$,
				$\tilde{O}(n^{1.5}w_{\max}^2)^{\dagger}$,
				$\tilde{O}(n^2w_{\max})$ \cite{dud84} \\
				\hline
			\end{tabular}
		}
		\begin{tablenotes}[flushleft]
			\footnotesize
			\item $\ast$ This bound follows from a naive divide-and-conquer approach combined with FFT.
			\item $\dagger$ Results obtained in this paper.
		\end{tablenotes}
	\end{table}
	
	The reduction from Knapsack-type problems to Multiple-Sequence $(\min,+)$ Convolution reveals that the answer sequences exhibit the same structural property as bounded tropical decomposition width polynomials. This provides a structural explanation for why such Knapsack-type problems often admit faster algorithms when the maximum weight $w_{\max}$ is small.
	\subsection{Extension Algebras for the Tropical Semiring}
	In classical algebra over $(\mathbb C,+,\cdot)$, the Fast Fourier Transform \cite{ct65} provides the following evaluation scheme. For $n=2^k$, every polynomial $A\in\mathbb C[x]$ with $\deg A<n$ is uniquely determined by its evaluations at the $n$-th roots of unity:
	$$
	A(\omega_n^0),A(\omega_n^1),\ldots,A(\omega_n^{n-1}),
	$$
	where $\omega_n$ is a primitive $n$-th root of unity. Since pointwise multiplication of values is far easier than polynomial convolution, this transform enables polynomial multiplication in nearly linear time.
	
	A natural question is whether such an evaluation-interpolation technique can be adapted to tropical polynomials. An analogue is provided by the Discrete Legendre-Fenchel transform \cite{fen49,lue97}, which can be interpreted as evaluating a tropical polynomial on the base semiring $\mathbb T$. In its usual max-plus form,
	$$
	f(x)=\max_{i=0}^{n}(ix-a_i),\quad x\in\mathbb R,
	$$
	and after converting to the $(\min,+)$ convention via $g(x)=-f(-x)$, it becomes
	$$
	g(x)=\min_{i=0}^{n}(a_i+ix),\quad x\in\mathbb R.
	$$
	
	However, this transform is not injective: it may map distinct tropical polynomials to the same evaluation function. For example, let $A=\min(0,2x)$ and $B=\min(0,x,2x)$. Then $A(x_0)=B(x_0)$ for every $x_0\in\mathbb R$, although $A\neq B$. Consequently, when multiplying two tropical polynomials $A$ and $B$, we cannot recover the exact coefficients of $A\otimes B$ from the evaluation function of $A\otimes B$.
	
	This motivates us to evaluate tropical polynomials not only in the base semiring $\mathbb T$, but also in an extension algebra $\mathcal R$ over $\mathbb T$. Ideally, we require that $\mathcal R$ be generated by finitely many elements over $\mathbb T$, so that every element admits a finite expression.
	
	A positive result in this direction is that, using the second modular convexity theorem (Theorem \ref{thm:sec_con}), we can show that for the class of tropical polynomials with bounded tropical decomposition width, a suitable tropical cyclic polynomial semiring serves as an interpolation algebra of finite generating rank over $\mathbb T$. This reveals a connection between the bounded factorization property and the extended identity property. The polynomials that support the efficient computation discussed above also admit an evaluation-interpolation approach.
	\begin{theorem}
		\label{thm:int_alg_intro}
		For every positive integer $k$,
		$$
		\mathbb T[y]/\left\langle\left(\operatorname{lcm}_{i=1}^{k}i\right)y\sim 0\right\rangle
		$$
		is an interpolation algebra of $\mathcal T_k$.
	\end{theorem}
	However, for the full semiring $\mathbb T[x]$, we obtain a negative result: the set of polynomials of degree at most $n$ requires any interpolation algebra to have generating rank at least $\lfloor n/2\rfloor+1$. Unlike the $(+,\times)$ convolution, for which $\mathbb C$ provides an interpolation algebra for $\mathbb C[x]$ with generating rank $1$, the $(\min,+)$ convolution lacks such a structure (more details of the difference between $\mathbb C[x]$ and $\mathbb T[x]$ will be discussed in Remark \ref{rmk:int}).
	\begin{theorem}
		For a positive integer $n$, let
		$\mathcal S=\{A\in\mathbb T[x]\mid \deg A\leq n\}$.
		Then every interpolation algebra $\mathcal R$ of $\mathcal S$ satisfies
		$\mu_{\mathbb T}(\mathcal R)\geq \lfloor n/2\rfloor+1$.
		Moreover, there exists an interpolation algebra $\mathcal R$ of
		$\mathcal S$ satisfying
		$\mu_{\mathbb T}(\mathcal R)=\lfloor n/2\rfloor+1$.
	\end{theorem}
	We have seen that for evaluation and interpolation, extension algebras provide extra support for polynomials of bounded tropical decomposition width, but not for the entire semiring $\mathbb T[x]$. Another natural question is what extension algebras can offer regarding decomposition.
	
	In classical algebra, a recent work \cite{gri26} shows that polynomials over any commutative ring can be factored into linear factors over an extension commutative ring. If this technique could be applied to $\mathbb T[x]$ to make every polynomial factorable into linear factors and units, then tropical polynomial multiplication could be reduced to the linear case.
	
	We give a negative result concerning this approach: no tropical polynomial admits a refined decomposition over any commutative semiring that is a flat $\mathbb T$-algebra. This reveals a rigidity in the decomposition structure of tropical polynomials, in contrast to classical polynomials. On the other hand, it also shows that the tropical decomposition width is a natural parameter, as it remains invariant under extension of the algebra.
	\begin{theorem}
		\label{thm:tdw_inv_intro}
		Let $\mathcal R$ be a commutative semiring that is a flat $\mathbb T$-algebra. Then for every tropical polynomial $A$,
		$$
		\operatorname{tdw}_{\mathcal R}(A)=\operatorname{tdw}(A).
		$$
	\end{theorem}
	\section{Preliminaries}
	\label{sec:pre}
	\subsection*{Notations}
	We use $\widetilde O(f)$ to denote $O(f\operatorname{polylog} f)$. For a
	logical proposition $P$, let $[P]$ be $1$ if $P$ is true and $0$ otherwise. We use $\log x$ to denote the logarithmic function with base $2$.
	
	For two positive integers $a,b$, we use $\gcd(a,b)$ and
	$\operatorname{lcm}(a,b)$ to denote their greatest common divisor and least
	common multiple, respectively.
	
	For sequences over a commutative semiring (such as tropical
	sequences), indices start at $0$, and for such a sequence
	$a=(a_0,a_1,\dots,a_n)$ we write $|a|=n$ for its largest index.
	For all other sequences (such as an index sequence), indices start at $1$, and for a sequence
	$a=(a_1,\dots,a_n)$ we write $|a|=n$ for its length.
	
	For a randomized algorithm, we say that it solves a problem with high probability if its success probability is at least $1-n^{-c}$ for every constant $c>0$, where $n$ denotes the input size.
	
	Throughout this paper, whenever we denote a time complexity by $T(n_1,\dots,n_k)$, we assume that it is monotone and closed under scaling of any argument by a constant factor: for any $c>0$ and any $1\le i\le k$,
	$T(n_1,\dots,cn_i,\dots,n_k)=O(T(n_1,\dots,n_k)).$
	\subsection*{Definitions}
	\begin{definition2}[commutative semiring]
		A \textit{commutative semiring} is a nonempty set $\mathcal R$ with two binary operations $\oplus$ and $\otimes$ such that
		\begin{itemize}
			\item $(\mathcal R,\oplus)$ is a commutative monoid with identity $0_{\mathcal R}$.
			\item $(\mathcal R,\otimes)$ is a commutative monoid with identity $1_{\mathcal R}\neq 0_{\mathcal R}$.
			\item $\otimes$ distributes over $\oplus$: $(a\oplus b)\otimes c=(a\otimes c)\oplus(b\otimes c)$ for all $a,b,c\in\mathcal R$.
			\item $0_{\mathcal R}\otimes a=0_{\mathcal R}$ for all $a\in\mathcal R$.
		\end{itemize}
		If $\oplus$ is idempotent, then $\mathcal R$ is called a \textit{commutative idempotent semiring}.
		
		For $k\in \mathbb N$ and $a\in \mathcal R$, we define
		$$
		a^{\otimes k}=
		\begin{cases}
			1_{\mathcal R} & \text{if } k=0,\\
			\underbrace{a\otimes a\cdots\otimes a}_{k\text{ times}} & \text{if } k\neq 0\\
		\end{cases}.
		$$
		
		Let $\mathcal R$ and $\mathcal S$ be commutative semirings. We say that $\mathcal S$ is an $\mathcal R$-algebra if there exists a unital semiring homomorphism $\varphi:\mathcal R\to\mathcal S$, called the structure homomorphism. By definition, such a homomorphism preserves both additive and multiplicative identities, so $\varphi(0_{\mathcal R})=0_{\mathcal S}$ and $\varphi(1_{\mathcal R})=1_{\mathcal S}$. If there are several natural choices for $\varphi$, we explicitly specify which one is intended; otherwise, we assume the canonical one whenever it is unambiguous.
		
		A subset $\mathcal{S}\subseteq\mathcal{R}$ is called
		\textit{multiplicatively closed} if for every $x,y\in\mathcal{S}$, we have
		$x\otimes y\in\mathcal{S}$.
	\end{definition2}
	\begin{lemma}[\cite{gol99}, Proposition 20.19]
		Let $\mathcal R$ be a commutative idempotent semiring. Define a binary
		relation $\leq$ on $\mathcal R$ by
		$$
		x\leq y
		\iff
		x\oplus y=x.
		$$
		Then $\leq$ is a partial order.
	\end{lemma}
	For a commutative idempotent semiring $\mathcal R$, the preceding lemma shows that $\oplus$ induces a natural partial order, with respect to which it acts as the minimum. To make subsequent formulas more readable and intuitive, we adopt the following notation for any commutative idempotent semiring:
	\begin{itemize}
		\item we write $\min$ for $\oplus$, $+$ for $\otimes$, $\infty$ for $0_{\mathcal R}$, and $0$ for $1_{\mathcal R}$;
		\item we denote the induced order by $\le$;
		\item for $k\in\mathbb N$ and $a\in\mathcal R$, we write $ka$ for $a^{\otimes k}$.
	\end{itemize}
	Thus, in the commutative idempotent case, we denote the semiring by \((\mathcal R,\min,+,\infty,0)\) instead of \((\mathcal R,\oplus,\otimes,0_{\mathcal R},1_{\mathcal R})\).
	\begin{definition2}[tropical semiring]
		The \textit{tropical semiring} is
		$
		\mathbb T=(\mathbb R\cup\{\infty\},\min,+,\infty,0),
		$
		where $\min$ and $+$ restrict to the usual minimum and addition on
		$\mathbb R$. In particular, for $x\in\mathbb R$, we set
		$
		\min(x,\infty)=\min(\infty,x)=x$ and $x+\infty=\infty+x=\infty$.
		For $k\in\mathbb R_{\geq 0}$ and $a\in\mathbb T$, we define
		$$
		ka=
		\begin{cases}
		0 & \text{if } k=0,\\
		k\cdot a & \text{if } k\neq 0 \text{ and } a\in\mathbb R,\\
		\infty & \text{if } k\neq 0 \text{ and } a=\infty,
		\end{cases}
		$$
		where $\cdot$ denotes the usual multiplication on $\mathbb R$.
		
		This algebraic structure first appeared in the context of automata
		theory in the work of \cite{sim78}, where the semiring
		$(\mathbb N\cup\{\infty\},\min,+)$ was used to study the limitedness
		problem for regular expressions. It has since been developed and
		generalized in subsequent works, evolving into a fundamental subject in
		theoretical computer science and mathematics \cite{sim94,pin98}, now
		often called \textit{tropical geometry} \cite{ms15} or \textit{tropical algebra}.
	\end{definition2}
	\begin{definition2}[polynomial semirings over a commutative semiring]
		Let $\mathcal R$ be a commutative semiring. We denote by $\mathcal R[x]$ the polynomial semiring in one variable $x$ over $\mathcal R$, and by $\mathcal R[x_1,\ldots,x_n]$ the polynomial semiring in $n$ variables $x_1,\ldots,x_n$ over $\mathcal R$.
		
		For two polynomials $A,B\in \mathcal R[x_1,\ldots,x_n]$, we say that $A$ and $B$ are \textit{equal}, denoted $A=B$, if all their corresponding coefficients are equal.
		
		Let $\mathcal R_2$ be a commutative semiring that is an $\mathcal R$-algebra. We say that $A$ and $B$ are \textit{functionally equivalent over $\mathcal R_2$}, denoted $A\sim_{\mathcal R_2} B$, if $A(\tau)=B(\tau)$ for all $\tau\in\mathcal R_2^n$.
		
		For a univariate polynomial $A\in\mathcal R[x]$ of a nonnegative degree $n$, we write
		$$
		A=\bigoplus_{i=0}^{n}\left(a_i\otimes x^{\otimes i}\right).
		$$
		Separately, if $\mathcal R$ is a commutative idempotent semiring, this becomes
		$$
		A=\min_{i=0}^{n}(a_i+ix).
		$$
		We denote by $L(A)$ the smallest index $i$ such that $a_i\neq 0_{\mathcal R}$.
		
		Throughout the remainder of this paper, the term \textit{tropical polynomial} refers to a polynomial in $\mathbb T[x]$.
	\end{definition2}
	\begin{definition2}[convex parameters of tropical polynomials]
		\label{def:conv_par}
		Let
		$
		A=\min_{i=0}^{n}(a_i+ix)
		$
		be a non-$\infty$ tropical polynomial. For indices $0\leq x<y<z\leq n$, the triplet
		$(x,y,z)$ is called a \textit{convex triplet} of $A$ if
		$$
		(z-x)a_y\leq (z-y)a_x+(y-x)a_z.
		$$
		
		The \textit{convex gap} of $A$, denoted
		$\operatorname{cgap}(A)$, is the smallest nonnegative integer $k$ such
		that for all integers $0\leq x<y\leq n$ with $y-x>k$, there exists an
		integer $z$ satisfying
		$	
		x<z<y
		$
		such that $(x,z,y)$ is a convex triplet of $A$. By convention, we set $\operatorname{cgap}(\infty)=0$. We call $A$ a \textit{convex polynomial} if
		$
		\operatorname{cgap}(A)\leq 1.
		$
		
		The \textit{convex support sequence} of $A$ is the
		increasing index sequence
		$
		p=(p_1,p_2,\ldots,p_t)
		$
		such that, for every $0\leq x\leq n$,
		$$
		x\in p
		\iff
		a_x\neq\infty
		\text{ and }
		\hat f_A(x)=a_x.
		$$
		We denote this sequence by $\operatorname{csupp}(A)$.
		
		An increasing sequence
		$
		p'=(p'_1,p'_2,\ldots,p'_r)
		$
		is called a \textit{weak convex support sequence} of $A$ if
		\begin{itemize}
			\item $p'\subseteq \operatorname{csupp}(A)$,
			\item for every integer $x$ with $0\leq x\leq n$, if
			$
			a_x\neq\infty
			$ and
			$
			|\partial \hat f_A(x)|\geq 2,
			$
			then $x\in p'$.
		\end{itemize}
		
		For every nonnegative integer $k$, we define $\mathcal C_k=\{A\in \mathbb T[x]\mid \operatorname{cgap}(A)\leq k\}$.
	\end{definition2}
	\begin{lemma}
		Let $A$ be a non-$\infty$ tropical polynomial. Then
		$\operatorname{csupp}(A)$ is a weak convex support sequence of
		$A$, and it has maximum length among all weak convex support
		sequences of $A$.
		
		Moreover, if $q'=(q'_1,\ldots,q'_m)$ is a weak convex support sequence of
		$A$, then
		$q'_1=L(A)$, $q'_m=\deg A$.
	\end{lemma}
	As in the computation of the convex hull of a planar point set, we can use
	Andrew's monotone chain algorithm \cite{and79} to compute the lower convex
	hull of the points
	$(i,a_i)$ with $a_i\neq\infty$.
	This gives the convex support sequence of a tropical polynomial in linear
	time, since the points are already sorted by their first coordinate.
	\begin{definition2}[functions with values in the tropical semiring]
		A function $f:\mathbb{R}\to\mathbb{T}$ is called a \textit{convex
			function} if for all $x,y\in\mathbb{R}$ and all $0\leq\lambda\leq 1$,
		$$
		f(\lambda x+(1-\lambda)y)
		\leq
		\lambda f(x)+(1-\lambda)f(y).
		$$
		
		For a tropical polynomial
		$
		A=\min_{i=0}^{n}(a_i+ix),
		$
		we define the \textit{convex envelope} of $A$ to be the function
		$\hat f_A:\mathbb{R}\to\mathbb{T}$ given by
		$$
		\hat f_A(x)
		=
		\min\left(
		\sum_{i=0}^{n}\lambda_i a_i
		\;\middle|\;
		\lambda=(\lambda_0,\ldots,\lambda_n)\in\mathbb{R}_{\geq 0}^{n+1},
		\sum_{i=0}^{n}\lambda_i i=x,
		\sum_{i=0}^{n}\lambda_i=1
		\right).
		$$
		If the feasible set is empty, the minimum is understood to be
		$\infty$.
		
		Let $f:\mathbb{R}\to\mathbb{T}$ be a function. For $x_0\in\mathbb{R}$,
		the \textit{subdifferential} of $f$ at $x_0$ is
		$$
		\partial f(x_0)
		=
		\{k\in\mathbb{R}\mid
		f(y)\geq f(x_0)+k(y-x_0)\text{ for all }y\in\mathbb{R}\}.
		$$
		
		For two convex functions $f_1,f_2:\mathbb{R}\to\mathbb{T}$, define
		their \textit{infimal
			convolution}, to be the function
		$$
		(f_1\square f_2)(x)
		=
		\inf_{y\in\mathbb{R}}\left(f_1(y)+f_2(x-y)\right).
		$$
		
		For a non-$\infty$ tropical polynomial $A$, a sequence
		$
		\left((s,h),(k_1,b_1),\ldots,(k_t,b_t)\right)
		$
		with $t\geq 0$ is called an \textit{expression} of $\hat f_A$ if
		$s+\sum_{i=1}^{t}b_i=\deg A$, $b_i>0$ for all $1\leq i\leq t$, 
		and $\hat f_A$ is given by
		{
		\small
		$$
			\hat f_A(x)=
			\begin{cases}
				h & \text{if } t=0 \text{ and } x=s,\\[2mm]
				\displaystyle
				h+\sum_{i=1}^{u-1}k_i b_i
				+k_u\left(x-s-\sum_{i=1}^{u-1}b_i\right)
				&
				\text{if } t>0 \text{ and }
				x\in
				\left[s+\sum_{i=1}^{u-1}b_i,\;
				s+\sum_{i=1}^{u}b_i\right]
				\text{ for some }1\leq u\leq t,\\[2mm]
				\infty & \text{otherwise}.
			\end{cases}
		$$
		}
		Here $(s,h)$ denotes the starting position and value of $\hat f_A$,
		while $(k_i,b_i)$ denotes the $i$-th segment: $k_i$ is its slope and
		$b_i$ is its horizontal length.
	\end{definition2}
	We present two important lemmas on the convex support sequence and convex envelope of tropical polynomials. Since these are standard results in discrete convex analysis \cite{mur03} (tropical polynomial multiplication can be seen as infimal convolution of functions $\mathbb Z\to \mathbb T$), we omit their proofs.
	\begin{lemma}
		\label{lem:env}
		Let $A=\min_{i=0}^{n}(a_i+ix)$ be a non-$\infty$ tropical polynomial, and let $\operatorname{csupp}(A)=(p_1,\ldots,p_m)$.
		
		The function $\hat f_A$ is given as follows. When $m=1$,
		$$
		\hat f_A(x)=
		\begin{cases}
			a_{p_1} & \text{if } x=p_1,\\
			\infty & \text{otherwise}.
		\end{cases}
		$$
		When $m\ge 2$,
		$$
		\hat f_A(x)=
		\begin{cases}
			\displaystyle
			\frac{p_{d+1}-x}{p_{d+1}-p_d}a_{p_d}
			+
			\frac{x-p_d}{p_{d+1}-p_d}a_{p_{d+1}}
			&
			\text{if } x\in[p_d,p_{d+1}]
			\text{ for some }1\le d<m,\\[2mm]
			\infty & \text{if } x\notin[p_1,p_m].
		\end{cases}
		$$
		
		Consequently,
		$$
		\left(
		(p_1,a_{p_1}),
		\left(
		\frac{a_{p_2}-a_{p_1}}{p_2-p_1},
		p_2-p_1
		\right),
		\ldots,
		\left(
		\frac{a_{p_m}-a_{p_{m-1}}}{p_m-p_{m-1}},
		p_m-p_{m-1}
		\right)
		\right)
		$$
		is an expression of $\hat f_A$.
		
		Furthermore, $\hat f_A$ is a convex function. For every $x_0\in [p_1,n]$, $\partial \hat f_A(x_0)$ is nonempty. Moreover, for any expression
		$\left((s,h),(k_1,b_1),\ldots,(k_t,b_t)\right)$ of $\hat f_A$, we have
		$k_1\le k_2\le\cdots\le k_t$.
	\end{lemma}
	\begin{lemma}
		\label{lem:expr}
		Let $A,B$ be two non-$\infty$ tropical polynomials, and let $C=A+B$.
		Suppose $\left((s,h),(k_1,b_1),\ldots,(k_t,b_t)\right)$ is an expression of $\hat{f}_A$,
		and $\left((s',h'),(k'_1,b'_1),\ldots,(k'_r,b'_r)\right)$ is an expression of $\hat{f}_B$.
		
		Let $\big((K_1,B_1),\ldots,(K_{t+r},B_{t+r})\big)$ be a stable merge of the two lists
		$\left((k_1,b_1),\dots,(k_t,b_t)\right)$ and $\left((k'_1,b'_1),\dots,(k'_r,b'_r)\right)$,
		where the merged list is ordered by nondecreasing slope,
		segments of equal slope may be placed in any order,
		and the relative order of segments coming from the same original list is preserved.
		Then
		$$
		\hat{f}_C = \hat{f}_A \mathbin{\square} \hat{f}_B,
		$$
		and
		$$
		\big((s+s',h+h'),(K_1,B_1),\dots,(K_{t+r},B_{t+r})\big)
		$$
		is an expression of $\hat{f}_C$.
	\end{lemma}
	\begin{definition2}[mod-$k$ arithmetic progression subpolynomial multiset]
		Let
		$
		A=\min_{i=0}^{n}(a_i+ix)
		$
		be a non-$\infty$ tropical polynomial. A tropical polynomial $B$ is called a
		\textit{subpolynomial} of $A$ if $B=\infty$ or there exists a strictly increasing
		index sequence
		$
		p=(p_0,p_1,\ldots,p_t)
		$
		with $0\leq p_0<\cdots<p_t\leq n$ such that
		$
		B=\min_{j=0}^{t}(a_{p_j}+jx).
		$
		
		For a positive integer $k$ and an integer $r$ with $0\leq r<k$, we define
		$$
		\operatorname{ms}_{k,r}(A)=
		\begin{cases}
			\displaystyle
			\min_{i=0}^{\left\lfloor(n-r)/k\right\rfloor}
			(a_{r+ki}+ix) & \text{if } r\leq n,\\
			\infty & \text{if } r>n.
		\end{cases}
		$$
		The \textit{mod-$k$ arithmetic progression subpolynomial multiset} of $A$ is the multiset
		$
		\operatorname{MS}_k(A)
		=
		\{\operatorname{ms}_{k,r}(A)\mid 0\leq r<k\}.
		$
		
		In the special case where $A=\infty$, the subpolynomial of $A$ is just $\infty$. For every positive integer $k$ and every integer $r$ with $0\leq r<k$, we define $\operatorname{ms}_{k,r}(A)=\infty$. The mod-$k$ arithmetic progression subpolynomial multiset is then defined in the same way as $\operatorname{MS}_k(A)=\{\operatorname{ms}_{k,r}(A)\mid 0\le r<k\}$.
	\end{definition2}
	\begin{definition2}[tropical decomposition width]
		Let $A$ be a non-$\infty$ tropical polynomial. A \textit{tropical polynomial
			decomposition} of $A$ is an equation
		$$
		A=\sum_{i=1}^{N}B^{(i)},
		$$
		where $B^{(1)},\ldots,B^{(N)}$ are tropical polynomials.
		
		The \textit{tropical decomposition width} of $A$, denoted by $\operatorname{tdw}(A)$, is the smallest nonnegative integer $k$ such that $A$ admits such a decomposition with $\deg B^{(i)}\le k$ for all $1\leq i\leq N$. Separately, we set $\operatorname{tdw}(\infty)=0$.
		
		For every nonnegative integer $k$, we define $\mathcal T_k=\{A\in \mathbb T[x]\mid \operatorname{tdw}(A)\leq k\}$.
	\end{definition2}
	\begin{definition2}[unique factorization polynomial]
		Let $A$ be a tropical polynomial. We call $A$ \textit{linear} if $\deg A=1$, and a \textit{unit} if $A\in \mathbb R$.
		
		For a non-$\infty$ and non-unit tropical polynomial $A$, we say that $A$ is \textit{irreducible} over $\mathcal S$ if every decomposition $A=\sum_{i=1}^N B^{(i)}$ with $B^{(i)}\in\mathcal S$ satisfies $\max_{i=1}^{N}\deg B^{(i)}=\deg A$; otherwise $A$ is \textit{reducible} over $\mathcal S$.
		
		Two tropical polynomials $A$ and $B$ are \textit{associated} if $A=B+u$ for some unit $u$.
		
		A tropical polynomial decomposition $A=\sum_{i=1}^N B^{(i)}$ is \textit{normal} over $\mathcal S$ if each $B^{(i)}$ lies in $\mathcal S$ and is irreducible over $\mathcal S$.
		
		Let $\mathcal S\subseteq \mathbb T[x]$. A polynomial $A\in\mathcal S$ is a \textit{unique factorization polynomial over $\mathcal S$} if either $A$ is $\infty$ or a unit, or any two normal decompositions of $A$ into factors in $\mathcal S$ agree up to permutation and association: precisely, if
		$
		A=\sum_{i=1}^N B^{(i)}
		$
		and
		$
		A=\sum_{j=1}^M C^{(j)}
		$
		are normal, then $N=M$, and after a permutation of the $C^{(j)}$, each $B^{(i)}$ is associated to $C^{(i)}$. When $\mathcal S=\mathbb T[x]$, we simply call $A$ a \textit{unique factorization polynomial}.
	\end{definition2}
	\begin{lemma}
		\label{lem:tdw_np_hard}
		Let $A$ be a tropical polynomial whose coefficients are integers in $\{0,1,2,3\}$. Then computing $\tdw(A)$ is NP-hard.
	\end{lemma}
	\begin{proof}
		By \cite{kr05} (Theorem 7), deciding reducibility of tropical polynomials whose coefficients are integers in $\{0,1,2,3\}$ is NP-hard. Since
		$$A\text{ is reducible} \iff \tdw(A) < \deg A\text{ and }A\not \in \mathbb T,$$ computing $\operatorname{tdw}(A)$ would solve reducibility, hence is NP-hard.
	\end{proof}
	\begin{definition2}[generating rank]
		Let $\mathcal R$ be a commutative semiring, and let $\mathcal M$ be a $\mathcal R$-semimodule. A finite set $\mathcal G=\{g_1,\ldots,g_m\}\subseteq\mathcal M$ is called a \textit{generating set} of $\mathcal M$ over $\mathcal R$ if 
		$$
		\mathcal R\mathcal G=\left\{\bigoplus_{i=1}^{m}\left(r_i\otimes g_i\right)\;\middle|\; r_i\in \mathcal R\right\}=\mathcal M.
		$$
		The \textit{generating rank} of $\mathcal M$ over $\mathcal R$ is defined to be the minimum cardinality of a generating set of $\mathcal M$ over $\mathcal R$, denoted by $\mu_{\mathcal R}(\mathcal M)$. If no such generating set exists, then $\mu_{\mathcal R}(\mathcal M)=\infty$.
	\end{definition2}
	\begin{definition2}[interpolation algebra]
		Let $\mathcal R$ be a commutative semiring, and let
		$\mathcal S\subseteq\mathcal R[x]$ be a set of polynomials. A commutative semiring $\mathcal R_2$ is called an
		\textit{interpolation algebra} of $\mathcal S$ if $\mathcal R_2$ is an $\mathcal R$-algebra, and for all
		$A,B\in\mathcal S$,
		$$
		A=B
		\iff
		A\sim_{\mathcal R_2}B.
		$$
	\end{definition2}
	\begin{remark}
		\label{rmk:int}
		The semiring $\mathbb T[x]$ itself provides a trivial interpolation algebra for $\mathbb T[x]$. However, since $\mu_{\mathbb T}(\mathbb T[x])=\infty$, for any fixed $m$, we cannot represent every element of $\mathbb T[x]$ using only $m$ registers over $\mathbb T$. To obtain an interpolation algebra $\mathcal R$ that is computationally manageable, we impose the condition $\mu_{\mathbb T}(\mathcal R)<\infty$, which aligns with the finitely generated setting in the extended identity property.
		
		In the classical $(+,\times)$ setting, $\mathbb C$ serves as an interpolation algebra for $\mathbb C[x]$: two complex polynomials are equal if and only if they agree on all complex inputs. Moreover, $\mu_{\mathbb C}(\mathbb C)=1$. Thus $\mathbb C$ provides an interpolation algebra for the structural information of $\mathbb C[x]$, with a single complex number register sufficient to represent any element of $\mathbb C$.
		
		It is important to note that the generating rank of an interpolation algebra is fundamentally different from the minimum number of evaluation points required. Although any polynomial over $\mathbb C[x]$ of degree at most $n$ can be uniquely determined by $n+1$ fixed evaluation points, the same does not hold for tropical polynomials: even for degree at most $1$, no fixed finite set of evaluation points can uniquely determine all such tropical polynomials. This is precisely why we use generating rank, rather than the number of evaluation points, to capture the complexity of interpolation in the tropical setting.
	\end{remark}
	Since tropical polynomials may be less intuitive in the context of combinatorial optimization, we adopt tropical sequence notation in Sections \ref{sec:cgap} and \ref{sec:alg}. Section \ref{sec:cgap} focuses on the geometric structure rather than the algebraic one, while Section \ref{sec:alg} is concerned with algorithm design for $(\min,+)$ convolution. For convenience, we assume that every non-$\infty$ tropical sequence $(a_0,\ldots,a_n)$ satisfies $a_n\neq\infty$, which avoids unnecessary complications and is consistent with the operations on tropical polynomials.
	\begin{definition2}[tropical sequence]
		We identify a tropical sequence $a=(a_0,\ldots,a_n)$ (with $a_n\neq\infty$) with the tropical polynomial
		$$
		A=\min_{i=0}^{n}(a_i+ix).
		$$
		The polynomial $\infty$ is identified with a distinguished tropical sequence, also denoted by $\infty$; equivalently, trailing $\infty$-entries are omitted, and the all-$\infty$ sequence is represented by $\infty$.
		
		Under this identification, we write
		$$
		c=\min(a,b) \iff C=\min(A,B),
		$$
		and
		$$
		c=a\otimes b \iff C=A+B.
		$$
	\end{definition2}
	\begin{remark}
		The above definitions for tropical polynomials also apply to the
		corresponding tropical sequences. If
		$
		a=(a_0,\ldots,a_n)
		$
		corresponds to
		$
		A=\min_{i=0}^{n}(a_i+ix),
		$
		then we write
		$$
		\operatorname{tdw}(a)=\operatorname{tdw}(A),
		\quad
		\operatorname{cgap}(a)=\operatorname{cgap}(A),
		\quad
		L(a)=L(A),
		\quad
		\hat f_a=\hat f_A
		$$
		and similarly for the other notions.
	\end{remark}
	\subsection*{Computational Model}
	We use a restricted version of the Real RAM, following the framework of \cite{cwx22}. Let $N$ denote the total input size. In addition to the standard unit-cost operations of a Word RAM with $O(\log N)$-bit words, the model supports unit-cost comparisons and arithmetic operations (addition, subtraction, multiplication, and division) on real numbers, as well as unit-cost conversion of word integers into real numbers. Conversion from real numbers to integers is not allowed, and randomization only uses random $O(\log N)$-bit words.
	
	To avoid unrealistic uses of unlimited-precision real arithmetic, we impose a low-degree computation restriction similar to Model (B') in Appendix A of \cite{cwx22}. At every point of the computation, every finite real register is required to equal the evaluation of a constant-degree polynomial in the real input values, with rational coefficients whose numerators and denominators have $O(\log N)$ bits. Arithmetic operations are permitted only when the resulting value satisfies this restriction. In particular, divisions by polynomially bounded word integers, after casting them to real numbers, are allowed.
	
	We augment the model with a distinguished symbol $\infty$. A tropical value is represented either by a finite real register or by $\infty$, and operations involving $\infty$ follow the conventions
	$$
	x+\infty=\infty+x=\infty,
	\qquad
	\min(x,\infty)=\min(\infty,x)=x.$$
	Such operations take constant time.
	
	For polynomially bounded integer inputs, all algorithms in this paper can be implemented on a standard Word RAM with essentially the same running times; the real quantities generated by our algorithms have $O(\log N)$-bit exact representations.
	\section{Structural Characterization for Tropical Polynomials}
	\label{sec:char}
	\subsection{Fragmented Properties for Tropical Polynomials}
	Let $\mathcal{S}\subseteq \mathbb{T}[x]$ be a multiplicatively closed set. We now give formal definitions
	of the five properties together with the convexity and the convex polynomial set mentioned in the introduction.
	\begin{definition2}[polynomial identity property]
		We say that $\mathcal{S}$ satisfies the \textit{polynomial identity property} if $\mathbb{T}$ is an interpolation algebra of $\mathcal{S}$.
	\end{definition2}
	\begin{definition2}[unique factorization property]
		We say that $\mathcal{S}$ satisfies the \textit{unique factorization property} if every polynomial $f\in \mathcal{S}$ is a unique
		factorization polynomial over $\mathcal{S}$.
	\end{definition2}
	\begin{definition2}[linear factorization property]
		We say that $\mathcal{S}$ satisfies the \textit{linear factorization property} if for every $f\in \mathcal{S}$,
		$
		\operatorname{tdw}(f)\leq 1.
		$
	\end{definition2}
	\begin{definition2}[bounded factorization property]
		We say that $\mathcal{S}$ satisfies the \textit{bounded factorization property} if there exists a constant $k>0$ such that for every
		$f\in \mathcal{S}$,
		$
		\operatorname{tdw}(f)\leq k.
		$
	\end{definition2}
	\begin{definition2}[extended identity property]
		\label{def:eip}
		We say that $\mathcal{S}$ satisfies the \textit{extended identity property} if there exists a commutative semiring
		$\mathcal{R}$ that is a $\mathbb T$-algebra such that
		\begin{itemize}
			\item $\mu_{\mathbb{T}}(\mathcal{R})<\infty$,
			\item $\mathcal{R}$ is an interpolation algebra of $\mathcal{S}$.
		\end{itemize}
	\end{definition2}
	\begin{definition2}[convexity]
		We say that $\mathcal S$ satisfies \textit{convexity} if every polynomial $f\in \mathcal S$ is a convex polynomial.
	\end{definition2}
	\begin{definition2}[convex polynomial set]
		Define the convex polynomial set $\mathcal S$ as the set of all convex polynomials.
	\end{definition2}
	We will prove some of the relations between the properties introduced in the Introduction. We begin with an elementary lemma that will be used repeatedly throughout this subsection; its proof is straightforward and hence omitted.
	\begin{lemma}
		\label{lem:conv_equiv}
		Let $A=\min_{i=0}^{n}(a_i+ix)$ be a non-$\infty$
		tropical polynomial, and let $L(A)=k$. Then the following statements
		are equivalent:
		\begin{itemize}
			\item[(1)] $A$ is a convex polynomial.
			\item[(2)] $a_i=\infty$ for all $0\leq i<k$, $a_i\neq\infty$ for all
			$k\leq i\leq n$, and
			$$
			a_{i+1}-a_i\leq a_{i+2}-a_{i+1}
			$$
			for all $k\leq i\leq n-2$.
			\item[(3)] $\operatorname{csupp}(A)=(k,k+1,\ldots,n)$.
		\end{itemize}
	\end{lemma}
	We now prove some of the implications in the relation graph introduced in the Introduction.
	\begin{lemma}
		\label{lem:poly_con_lf}
		A multiplicatively closed set $\mathcal{S}\subseteq \mathbb{T}[x]$ satisfies
		convexity if and only if it satisfies the linear factorization property.
	\end{lemma}
	\begin{proof}
		Assume first that $\mathcal{S}$ satisfies convexity. Let $A\in \mathcal S$. If $A=\infty$, then $\operatorname{tdw}(A)=0\leq 1$; otherwise write $A=\min_{i=0}^{n}(a_i+ix)$ and put $k=L(A)$.
		Since $A$ is convex, by Lemma \ref{lem:conv_equiv}, the finite coefficients of $A$ are exactly
		$a_k,\ldots,a_n$, and the differences
		$
		a_{k+1}-a_k,a_{k+2}-a_{k+1},\ldots,a_n-a_{n-1}
		$
		are nondecreasing. Hence
		$$
		A
		=
		kx+a_k+\sum_{i=k}^{n-1}\min(0,a_{i+1}-a_i+x),
		$$
		so $\operatorname{tdw}(A)\leq 1$.
		
		Conversely, assume that $\mathcal S$ satisfies the linear factorization property. For all $A\in \mathcal S$, consider a tropical polynomial decomposition $A=\sum_{i=1}^{N}B^{(i)}$ with $\deg B^{(i)}\leq 1$ for all $1\leq i\leq N$. We permute all factors $B^{(1)},\ldots,B^{(N)}$, write them as $B^{(i)}=\min(b_i,c_i+x)$ for $1\leq i\leq p$ with $d_i=c_i-b_i$ ordered such that $d_1\leq\cdots\leq d_p$, and $B^{(i)}=c_i+x$ for $p<i\leq n$, and $B^{(i)}=c_i$ for $n+1\leq i\leq N$.
		Then
		$$
		\begin{aligned}
			A&=\sum_{i=1}^{p}\min(b_i,c_i+x)+\sum_{i=p+1}^{n}(c_i+x)+\sum_{i=n+1}^{N}c_i\\
			&=\left(\sum_{i=1}^{p}b_i+\sum_{i=p+1}^{N}c_i\right)+(n-p)x+\sum_{i=1}^{p}\min(0,d_i+x)\\
			&=\left(\sum_{i=1}^{p}b_i+\sum_{i=p+1}^{N}c_i\right)+\min_{i=n-p}^{n}\left(\sum_{j=1}^{i-(n-p)}d_j+ix\right).
		\end{aligned}
		$$
		As the successive coefficient differences
		$d_1,\ldots,d_p$ are nondecreasing, Lemma
		\ref{lem:conv_equiv} implies that $A$ is a convex polynomial.
	\end{proof}
	\begin{lemma}
		\label{lem:poly_con_ip}
		If a multiplicatively closed set $\mathcal{S}\subseteq \mathbb{T}[x]$ satisfies
		convexity, then it satisfies the polynomial identity property.
		Furthermore, there exists a multiplicatively closed set
		$\mathcal{S}\subseteq \mathbb{T}[x]$ satisfying the polynomial identity
		property but not satisfying convexity.
	\end{lemma}
	\begin{proof}
		Assume that $\mathcal{S}$ satisfies convexity, and let $A,B$ be two distinct tropical polynomials in $\mathcal{S}$. If $A=\infty$ or $B=\infty$, then clearly $A\not\sim_{\mathbb{T}}B$. Thus it suffices to consider the case $A,B\neq\infty$. Write $A=\min_{i=0}^{m}(a_i+ix)$ and $B=\min_{i=0}^{t}(b_i+ix)$. If $m\neq t$, then evaluating $A$ and $B$ at a real number $x_0$ and taking $x_0\to -\infty$ separates them.
		
		Now suppose $m=t$, and let $p$ be the first index with $a_p\neq b_p$.
		Since we may swap $A$ and $B$, we assume that $a_p<b_p$. If $b_p=\infty$, then
		$L(A)<L(B)$, and evaluating $A$ and $B$ at a real number $x_0$ and taking $x_0\to \infty$ separates them. Otherwise $b_p\ne \infty$. Since $B$ is convex,
		$\hat f_B(p)=b_p$. By Lemma \ref{lem:env}, choose $\lambda\in\partial\hat f_B(p)$. Then for every
		$0\leq y\leq t$,
		$$
		b_y\geq \hat{f}_B(y)\geq \hat{f}_B(p)+\lambda(y-p)=b_p+\lambda(y-p),
		$$
		and hence
		$
		B(-\lambda)=\min_{i=0}^{t}(b_i-i\lambda)=b_p-p\lambda.
		$
		But
		$$
		A(-\lambda)\leq a_p-p\lambda<b_p-p\lambda=B(-\lambda).
		$$
		Thus any two distinct polynomials in $\mathcal S$ are separated by
		evaluation over $\mathbb T$, so $\mathbb T$ is an interpolation algebra of
		$\mathcal S$.
		
		For the second statement, let $\mathcal S=\{k\min(0,2x)\mid k\in \mathbb Z_{>0}\}$.
		For any two elements $a\min(0,2x)$ and $b\min(0,2x)$ in $\mathcal S$, we have $$a\min(0,2x)+b\min(0,2x)=(a+b)\min(0,2x)\in \mathcal S,$$ 
		so $\mathcal S$ is multiplicatively closed. Moreover, evaluating at $-1$ separates the elements of $\mathcal S$, since $a\min(0,-2)=-2a$. Thus $\mathcal S$ satisfies the polynomial identity property. However, the function $\min(0,2x)\in \mathcal S$ is not convex, hence $\mathcal S$ does not satisfy convexity.
	\end{proof}
	\begin{lemma}
		\label{lem:poly_con_uf}
		The convex polynomial set satisfies the unique factorization property. However, convexity of a multiplicatively closed set $\mathcal S\subseteq \mathbb T[x]$ does not imply the unique factorization property. Moreover, there exists a multiplicatively closed set $\mathcal S\subseteq \mathbb T[x]$ satisfying the unique factorization property but not satisfying the bounded factorization property.
	\end{lemma}
	\begin{proof}
		For the convex polynomial set $\mathcal S$, let
		$A=\min_{i=0}^{n}(a_i+ix)\in\mathcal S$ be a non-$\infty$ and non-unit polynomial, and consider any
		normal tropical polynomial decomposition
		$
		A=\sum_{i=1}^{N} B^{(i)}
		$
		with $B^{(i)}\in \mathcal S$. By Lemma \ref{lem:poly_con_lf}, $\mathcal S$ satisfies the linear factorization property. Then for all $1\leq i\leq N$, $\operatorname{tdw}(B^{(i)})\leq 1$. Since all linear polynomials are convex, they belong to $\mathcal S$. Hence any linear tropical decomposition of $B^{(i)}$ is a decomposition over $\mathcal S$. Since $B^{(i)}$ is irreducible over $\mathcal S$, it must itself be linear. Consequently, we have $n=N$.
		
		Applying the coefficient comparison from the proof of the converse direction of Lemma \ref{lem:poly_con_lf} (using the notation $p,d_i$ introduced there), we obtain
		$$p=n-L(A),\qquad d_i=a_{i+n-p}-a_{i-1+n-p}(1\le i\le p).$$
		Thus the monomial part and all binomial linear factors are determined up to association and permutation.
		
		For the second statement, let $\mathcal S=\{\min_{i=0}^{k}(ix)\mid k\geq 10, k\in\mathbb Z\}$
		which satisfies convexity. For any two elements $\min_{i=0}^{a}(ix)$ and $\min_{i=0}^{b}(ix)$ in $\mathcal S$, we have
		$$\min_{i=0}^{a}(ix)+\min_{i=0}^{b}(ix)=\min_{i=0}^{a+b}(ix)\in \mathcal S.$$
		So $\mathcal S$ is multiplicatively closed. Moreover,
		$$\min_{i=0}^{23}(ix)=\min_{i=0}^{10}(ix)+\min_{i=0}^{13}(ix)\quad \text{and}\quad
		\min_{i=0}^{23}(ix)=\min_{i=0}^{11}(ix)+\min_{i=0}^{12}(ix)$$
		provides two distinct normal tropical polynomial decompositions. Thus, $\min_{i=0}^{23}(ix)$ is not a unique factorization polynomial over $\mathcal S$. Hence $\mathcal S$ does not satisfy the unique factorization property.
		
		For the third statement, define
		$$\mathcal S=\left\{\sum_{i=1}^{m}\min(0,c_i+1+c_i x)\;\middle|\; m\ge 1, (c_1,\ldots,c_m)\in \mathbb Z_{>0}^{m}\right\}.$$
		For any two elements $\sum_{i=1}^{m}\min(0,c_i+1+c_i x)$ and $\sum_{i=1}^{t}\min(0,d_i+1+d_i x)$ in $\mathcal S$, their product $\sum_{i=1}^{m}\min(0,c_i+1+c_ix)+\sum_{i=1}^{t}\min(0,d_i+1+d_ix)$ is again of the same form and hence belongs to $\mathcal S$. Thus $\mathcal S$ is multiplicatively closed.
		
		Moreover, every irreducible polynomial in $\mathcal S$ must be of the form $\min(0,c+1+cx)$. Consequently, for any $A\in\mathcal S$, any normal tropical polynomial decomposition can only be written as
		$A=\sum_{i=1}^{m}\min(0,c_i+1+c_i x)$.
		
		Now suppose that $A$ admits two distinct normal decompositions
		$$A=\sum_{i=1}^{m}\min(0,c_i+1+c_i x)\quad\text{and}\quad A=\sum_{i=1}^{t}\min(0,d_i+1+d_i x).$$
		Permute the factors so that $c_1\ge\cdots\ge c_m$ and $d_1\ge\cdots\ge d_t$. For all $1\leq i\leq m$, let $B^{(i)}=\min(0,c_i+1+c_i x)$. Then $\bigl((0,0),((c_i+1)/c_i,c_i)\bigr)$ is an expression of $\hat f_{B^{(i)}}$. By Lemma \ref{lem:expr},
		$$\left((0,0),((c_1+1)/c_1,c_1),\ldots,((c_m+1)/c_m,c_m)\right)$$
		is an expression of $\hat f_A$. Similarly,
		$$\left((0,0),((d_1+1)/d_1,d_1),\ldots,((d_t+1)/d_t,d_t)\right)$$
		is also an expression of $\hat f_A$.
		
		Since the two decompositions are distinct, and since $\sum_{i=1}^{m} c_i=\deg A=\sum_{i=1}^{t}d_i$, we can choose the smallest index $p$ such that $c_p\ne d_p$. Then $c_i=d_i$ for all $1\le i<p$. Set $x=\sum_{i=1}^{p-1}c_i+1/2=\sum_{i=1}^{p-1}d_i+1/2$.
		Using the two expressions for $\hat f_A$, we obtain
		$$
		\hat f_A(x)=
		\sum_{i=1}^{p-1}(c_i+1)+\frac12\cdot\frac{c_p+1}{c_p}
		=
		\sum_{i=1}^{p-1}(d_i+1)+\frac12\cdot\frac{d_p+1}{d_p}.
		$$
		Since $\sum_{i=1}^{p-1}(c_i+1)=\sum_{i=1}^{p-1}(d_i+1)$, it follows that $c_p=d_p$, contradicting $c_p\neq d_p$. Therefore, $\mathcal S$ satisfies the unique factorization property.
		
		Finally, for every $k\in\mathbb Z_{>0}$, we have $\operatorname{tdw}(\min(0,k+2+(k+1)x))=k+1>k$, which shows that $\mathcal S$ does not satisfy the bounded factorization property.
	\end{proof}
	\subsection{Tropical Polynomials of Bounded Tropical Decomposition Width}
	To prove two structural theorems (Theorems \ref{thm:fir_con} and \ref{thm:sec_con}), we first need to solve an additive combinatorics problem.
	\begin{definition2}
		\label{def:add_com}
		For positive integers $m$ and $k$, define $\mathcal{F}(m,k)$ to be the smallest integer $d$ such that every positive integer multiset $S$ satisfying
		\begin{itemize}
		\item $1\leq x\leq k$ for all $x\in S$
		\item $\sum_{x\in S}x\geq d$
		\end{itemize}
		contains a nonempty proper sub-multiset $T$ with sum divisible by $m$.
	\end{definition2}
	Since obtaining the exact value appears difficult, we instead establish upper and lower bounds using different approaches.
	\begin{lemma}
		\label{lem:add_com_upp}
		For any positive integers $m,k$, we have $\mathcal{F}(m,k)\le mk+1.$
	\end{lemma}
	\begin{proof}
		Let $S=\{s_1,\ldots,s_n\}$ be a positive integer multiset with
		$1\le s_i\le k$ for all $1\leq i\leq n$, and suppose that
		$\sum_{x\in S}x\ge mk+1$. Since $s_i\le k$, we have $n\ge m+1$.
		
		Consider the prefix sums
		$t_i=\left(\sum_{j=1}^i s_j\right)\bmod m$.
		By the pigeonhole principle, there exist $0\le l<r\le m$ such that
		$t_l=t_r$. Hence
		$$
		\sum_{j=l+1}^r s_j\equiv t_r-t_l\equiv 0\pmod m.
		$$
		Set $T=\{s_i\mid l+1\leq i\leq r\}$. Since $0<|T|\leq m<n$, $T$ gives a nonempty proper sub-multiset of $S$ whose sum is divisible by $m$. Therefore
		$\mathcal{F}(m,k)\le mk+1$.
	\end{proof}
	\begin{lemma}
		\label{lem:add_com_upp2}
		For any positive integers $m,k$, we have
		$$
		\mathcal{F}(m,k)\le
		\max\left(
		\max_{i=1}^{k}
		\left(
		\frac{(i-1)k(k+1)}{2}+1+\operatorname{lcm}(i,m)
		\right),
		k^3
		\right).
		$$
	\end{lemma}
	\begin{proof}
		Let
		$$
		B=
		\max\left(
		\max_{i=1}^{k}
		\left(
		\frac{(i-1)k(k+1)}{2}+1+\operatorname{lcm}(i,m)
		\right),
		k^3
		\right).
		$$
		Let $S$ be a positive integer multiset with $1\le x\le k$ for all
		$x\in S$, and suppose that
		$
		\sum_{x\in S}x\ge B.
		$
		For $1\le i\le k$, let
		$
		a_i=\sum_{x\in S}[x=i].
		$
		Choose $p\in\{1,\ldots,k\}$ such that $pa_p$ is maximal. Then
		$$
		pa_p\ge \frac{1}{k}\sum_{i=1}^k ia_i
		=\frac{1}{k}\sum_{x\in S}x
		\ge k^2.
		$$
		In particular,
		$
		a_p\geq k^2/p\geq k.
		$
		Put
		$
		L=\operatorname{lcm}(p,m).
		$
		
		For every $1\le q\le k$ with $q\neq p$, call a multiset of $p$ copies of $q$ \textit{a block of type} $q$. We build a multiset $U$ by processing the types in the order
		$1,2,\ldots,p-1,p+1,\ldots,k$. Initially $U=\varnothing$. For each type $q$ in this order, we repeatedly take a block of type $q$ from $S$ and add it to $U$, as long as the sum of $U$ does not exceed $L$ and $S$ contains enough copies of $q$ to form such a block.
		
		Equivalently, at step $q$, we add $b$ blocks of type $q$, where
		$b$ is the largest integer satisfying
		$$
		0\le b\le \left\lfloor \frac{a_q}{p}\right\rfloor
		\quad\text{and}\quad
		\sum_{x\in U}x+bpq\le L.
		$$
		
		There are two cases:
		
		$\textbf{Case 1}$. If for each step $q$, we add exactly $\lfloor a_q/p\rfloor$ blocks of type $q$ to $U$. Then $U$ contains exactly
		$\left\lfloor a_q/p\right\rfloor $ blocks of type $q$ for each $q\ne p$. Hence
		$$
		\begin{aligned}
		\sum_{x\in U}x+a_pp
		&=
		\sum_{q\ne p}\left\lfloor \frac{a_q}{p}\right\rfloor pq+a_pp\\
		&=
		\sum_{q\ne p}(a_q-(a_q\bmod p))q+a_pp\\
		&=
		\sum_{q=1}^k a_qq-\sum_{q\ne p}(a_q\bmod p)q\\
		&\ge
		\sum_{x\in S}x-(p-1)\sum_{q=1}^k q\\
		&\ge
		\left(\frac{(p-1)k(k+1)}{2}+1+L\right)
		-\frac{(p-1)k(k+1)}{2}\\
		&>L.
		\end{aligned}
		$$
		Since both $L$ and $\sum_{x\in U}x$ are divisible by $p$,
		$
		w=(L-\sum_{x\in U}x)/p
		$
		is a nonnegative integer. The inequality above gives $w<a_p$. Therefore we may add
		$w$ copies of $p$ to $U$, obtaining a sub-multiset $T\subseteq S$ with
		$
		\sum_{x\in T}x=L.
		$
		
		$\textbf{Case 2}.$ If for some step $q$, we add less than $\lfloor a_q/p\rfloor$ blocks of $q$ to $U$. At the first such step, say for
		$q$, let $U_0$ be the multiset before adding blocks of type $q$. Then
		$$
		\sum_{x\in U_0}x+
		\left\lfloor \frac{a_q}{p}\right\rfloor pq>L.
		$$
		Since we added the maximum possible number of blocks, the new
		multiset $U_1$ satisfies
		$$
		\begin{aligned}
		\sum_{x\in U_1}x
		&=
		\sum_{x\in U_0}x+
		\left\lfloor
		\frac{L-\sum_{x\in U_0}x}{pq}
		\right\rfloor pq\\
		&=
		L-
		\left(
		\left(L-\sum_{x\in U_0}x\right)\bmod pq
		\right)\\
		&>L-pq\\
		&\ge L-pk.
		\end{aligned}
		$$
		Since the total sum only increases afterwards, the final $U$ satisfies
		$
		\sum_{x\in U}x>L-pk.
		$
		Again $p\mid L-\sum_{x\in U}x$, so
		$
		w=(L-\sum_{x\in U}x)/p
		$
		is a nonnegative integer. From the previous inequality,
		$
		w<k\leq a_p.
		$
		Thus we may add $w$ copies of $p$ to $U$, obtaining a sub-multiset
		$T\subseteq S$ with
		$
		\sum_{x\in T}x=L.
		$
		
		In both cases, we obtain a sub-multiset $T$ such that
		$
		\sum_{x\in T}x=L=\operatorname{lcm}(p,m),
		$
		so
		$
		m\mid \sum_{x\in T}x.
		$
		Moreover, $T$ is nonempty. It is proper because
		$$
		\sum_{x\in T}x=L
		<
		\frac{(p-1)k(k+1)}{2}+1+L
		\le B\le \sum_{x\in S}x.
		$$
		Therefore every such $S$ contains a nonempty proper sub-multiset whose
		sum is divisible by $m$. Hence $\mathcal{F}(m,k)\le B$.
	\end{proof}
	\begin{lemma}
		For positive integers $m,k$, we have
		$
		\mathcal{F}(m,k)\geq
		\max_{i=1}^{k}\left(\operatorname{lcm}(i,m)\right)+1.
		$
	\end{lemma}
	\begin{proof}
		Let $p$ be an integer with $1\leq p\leq k$ such that
		$
		\operatorname{lcm}(p,m)=\max_{i=1}^{k}\operatorname{lcm}(i,m).
		$
		Put
		$
		L=\operatorname{lcm}(p,m),
		$
		and let $S$ be the multiset consisting of $L/p$ copies of $p$. Then
		$
		\sum_{x\in S}x=L.
		$
		
		Every nonempty proper sub-multiset $T$ of $S$ consists of $r$ copies of
		$p$ for some integer
		$
		0<r<L/p.
		$
		If $m\mid \sum_{x\in T}x$, then
		$
		m\mid rp.
		$
		Equivalently,
		$
		\frac{m}{\gcd(m,p)}\mid r.
		$
		But
		$$
		\frac{m}{\gcd(m,p)}=\frac{\operatorname{lcm}(p,m)}{p}=\frac{L}{p},
		$$
		which contradicts $0<r<L/p$. Hence no nonempty proper sub-multiset of
		$S$ has sum divisible by $m$.
		
		Thus there exists a multiset $S$ with
		$
		\sum_{x\in S}x=L
		$
		which does not satisfy the property in Definition \ref{def:add_com}. Therefore
		$$
		\mathcal{F}(m,k)\geq L+1
		=
		\max_{i=1}^{k}\left(\operatorname{lcm}(i,m)\right)+1.
		$$
	\end{proof}
	For the second structural theorem (Theorem \ref{thm:sec_con}), we need to prove that for all positive integers $m,k$ with $\operatorname{lcm}_{i=1}^{k}i\big|m$, we have $\mathcal{F}(m,k)\leq 2m$. Lemma \ref{lem:add_com_upp2} gives $\mathcal{F}(m,k)\leq m+k^3$, which proves the desired bound whenever $k^3\leq \operatorname{lcm}_{i=1}^{k}i$. This inequality holds for all sufficiently large $k$, but fails for some small values of $k$, so we give a more precise bound for $\mathcal{F}$ in the small-$k$ regime.
	
	To obtain an absolute lower bound for $\operatorname{lcm}_{i=1}^{k}i$ for large $k$, rather than one in the asymptotic sense, we will present some crude but tight lower bounds.
	\begin{lemma}
		\label{lem:numb}
		For positive integers $k,d$ with $k\geq d$, $$\operatorname{lcm}_{i=1}^{k} i\geq \prod_{i=d}^{k}\frac{i}{(k-i)!}.$$
	\end{lemma}
	\begin{proof}
		Let $L_i=\operatorname{lcm}_{j=i}^{k} j$ for all $1\leq i\leq k$, and set $L_{k+1}=1$. For each $1\leq i\leq k$, we have
		$$\prod_{j=i+1}^k j \equiv \prod_{j=i+1}^k (j-i) \equiv (k-i)! \pmod i,$$
		so
		$$\gcd(i,L_{i+1})\leq \gcd(i,\prod_{j=i+1}^{k}j)=\gcd(i,(k-i)!)\leq (k-i)!.$$
		Therefore
		$$\frac{L_i}{L_{i+1}}=\frac{i}{\operatorname{gcd}(i,L_{i+1})}\geq \frac{i}{(k-i)!},$$
		and multiplying over $i=d,\ldots,k$ gives
		$$\operatorname{lcm}_{i=1}^{k}i=L_1\geq L_d\geq \prod_{i=d}^k \frac{i}{(k-i)!}.$$
	\end{proof}
	For the case $1\leq k\leq 6$, the greedy construction argument of Lemma \ref{lem:add_com_upp2} does not apply. We instead give a different greedy construction argument to handle these cases.
	\begin{lemma}
		\label{lem:dif_est}
		Let $m,k$ be positive integers such that
		$
		\operatorname{lcm}_{i=1}^{k}i\mid m.
		$
		Let $S$ be a positive integer multiset with all elements in $[1,k]$ and
		$
		\sum_{x\in S}x>m.
		$
		Suppose that no nonempty proper sub-multiset of $S$ has sum divisible
		by $m$. Define
		$
		a_i=\sum_{x\in S}[x=i]
		$.
		Let $b=(b_1,\ldots,b_t)$ be a sequence of distinct elements from
		$[1,k]$, and set
		$$
		c_i=\operatorname{lcm}_{j=1}^{i}b_j,
		\quad
		1\leq i\leq t.
		$$
		Then
		$$
		\sum_{i=1}^{t}
		c_i\left\lfloor \frac{b_i a_{b_i}}{c_i}\right\rfloor
		<m.
		$$
	\end{lemma}
	\begin{proof}
		Suppose, to the contrary, that
		$$
		\sum_{i=1}^{t}
		c_i\left\lfloor\frac{b_i a_{b_i}}{c_i}\right\rfloor
		\geq m.
		$$
		
		For each $1\leq q\leq t$, since $b_q\mid c_q$, $c_q/b_q$ is a positive integer. Define \textit{a block of type} $q$ to be a multiset consisting of $c_q/b_q$ copies of
		$b_q$, its total sum is $c_q$. Set
		$$
		F_q = \left\lfloor \frac{b_q a_{b_q}}{c_q}\right\rfloor,
		$$
		the maximum number of blocks of type $q$ that can be chosen from $S$.
		
		We now build a multiset $U$ by processing the types in the order $t,t-1,\ldots,1$, similarly to the greedy construction in Lemma \ref{lem:add_com_upp2}. Initially $U=\varnothing$. For each type $q$ in this order, we repeatedly take a block of type $q$ from $S$ and add it to $U$, as long as the sum of $U$ does not exceed $m$ and $S$ contains enough copies of $q$ to form such a block.
		
		As in the proof of Lemma \ref{lem:add_com_upp2}, we distinguish two cases.
		
		\textbf{Case 1.} If there exists a step $q$ where adding all $F_q$ blocks would exceed $m$.
		Let $q$ be the first such step, and let $U_0$ be the multiset before this step.
		Then $\sum_{x\in U_0}x$ is a multiple of $c_q$. This is because all previously chosen blocks are of types $j>q$, and for each such type, the sum of its blocks is a multiple of $c_j$, with $c_q\mid c_j$.
		Moreover, $c_q\mid m$ because $c_q\mid\operatorname{lcm}_{i=1}^{k}i$ and the latter divides $m$.
		Hence $m-\sum_{x\in U_0}x$ is a nonnegative multiple of $c_q$.
		
		Since adding all $F_q$ blocks would exceed $m$, we have
		$\sum_{x\in U_0}x+c_qF_q>m$. This implies that
		$$
		0\le\frac{m-\sum_{x\in U_0}x}{c_q}<F_q.
		$$
		Thus we may add exactly $(m-\sum_{x\in U_0}x)/c_q$ blocks of type $q$ to $U_0$, 
		obtaining a sub-multiset $T\subseteq S$ with $\sum_{x\in T}x=m$.
		
		\textbf{Case 2.} If no such step occurs. Then we add all \(F_q\) blocks of every type. The total sum of $U$ is at most $m$ by construction, and at least $m$ since
		$$\sum_{x\in U}x=\sum_{i=1}^{t} c_i F_i = \sum_{i=1}^{t} c_i \left\lfloor \frac{b_i a_{b_i}}{c_i} \right\rfloor \ge m,$$
		so it equals $m$. Hence $U$ gives a sub-multiset of $S$ with total sum $m$.
		
		In both cases we obtain $T\subseteq S$ with $\sum_{x\in T}x=m$.
		Since $m>0$, $T$ is nonempty, and since $\sum_{x\in S}x>m$, $T$ is proper.
		This contradicts the hypothesis that no nonempty proper sub-multiset has sum divisible by $m$.
		Therefore $$\sum_{i=1}^{t} c_i\left\lfloor \frac{b_i a_{b_i}}{c_i}\right\rfloor < m.$$
	\end{proof}
	Combining Lemmas \ref{lem:add_com_upp2} and \ref{lem:dif_est} yields our final upper bound for $\mathcal F$ in the case $\operatorname{lcm}_{i=1}^{k} i\mid m$.
	\begin{lemma}
		\label{lem:add_com_upp3_tig}
		For positive integers $m,k$ with
		$
		\operatorname{lcm}_{i=1}^{k} i\mid m,
		$
		we have
		$
		\mathcal{F}(m,k)\leq 2m.
		$
	\end{lemma}
	\begin{proof}
		We consider two cases for $k$:
		
		\textbf{Case 1}. If $k\geq 7$, let $L_k=\operatorname{lcm}_{i=1}^{k}i$. We first note that $L_k\geq k^3$ for all $k\geq 7$. Indeed, by Lemma
		\ref{lem:numb},
		$$
		L_k\geq
		\prod_{i=k-3}^{k}\frac{i}{(k-i)!}
		=
		\frac{k(k-1)(k-2)(k-3)}{12},
		$$
		which is at least $k^3$ for $k\geq 18$; the remaining cases
		$7\leq k\leq 17$ are checked directly.
		
		Since $L_k\mid m$, we have $\operatorname{lcm}(i,m)=m$ for all
		$1\leq i\leq k$. Therefore, by Lemma \ref{lem:add_com_upp2},
		$$
		\begin{aligned}
			\mathcal{F}(m,k)&\leq
			\max\left(
			\max_{i=1}^{k}\left(\frac{(i-1)k(k+1)}{2}+1+m\right),
			k^3
			\right)\\
			&=\max\left(\frac{(k-1)k(k+1)}{2}+1+m,k^3\right)\\
			&\leq m+k^3\\
			&\leq m+L_k\\
			&\leq 2m.
		\end{aligned}
		$$
		
		\textbf{Case 2}. If $k\leq 6$, let $S$ be a positive integer multiset with all elements in $[1,k]$ and
		$
		\sum_{x\in S}x\geq 2m>m.
		$
		Put
		$
		a_i=\sum_{x\in S}[x=i]
		$
		and suppose, for contradiction, that no nonempty proper sub-multiset of
		$S$ has sum divisible by $m$. We repeatedly apply Lemma \ref{lem:dif_est}.
		
		For $k=1$, taking $b=(1)$ gives
		$$
		a_1<m,
		$$
		contradicting $\sum_{x\in S}x=a_1\geq 2m$.
		
		For $k=2$, taking $b=(1,2)$ gives
		$$
		a_1+2a_2<m,
		$$
		contradicting $\sum_{x\in S}x=a_1+2a_2\geq 2m$.
		
		For $k=3$, taking $b=(1,2)$ and $b=(1,3)$ gives
		$$
		a_1+2a_2<m,
		\qquad
		a_1+3a_3<m.
		$$
		Hence
		$$
		a_1+2a_2+3a_3
		\leq
		2a_1+2a_2+3a_3
		<2m,
		$$
		contradicting $\sum_{x\in S}x=a_1+2a_2+3a_3\geq 2m$.
		
		For $k=4$, taking $b=(1,2,4)$ and $b=(1,3)$ gives
		$$
		a_1+2a_2+4a_4<m,
		\qquad
		a_1+3a_3<m.
		$$
		Thus
		$$
		a_1+2a_2+3a_3+4a_4
		\leq
		2a_1+2a_2+3a_3+4a_4
		<2m,
		$$
		again a contradiction.
		
		For $k=6$, taking
		$b=(1,2,4,5)$, $b=(1,3,6,5)$, and $b=(1,3,6,2,4)$ gives
		$$
		a_1+2a_2+4a_4+20\left\lfloor\frac{a_5}{4}\right\rfloor<m,\qquad
		a_1+3a_3+6a_6+30\left\lfloor\frac{a_5}{6}\right\rfloor<m,
		$$
		and
		$$
		a_1+3a_3+6a_6
		+6\left\lfloor\frac{a_2}{3}\right\rfloor
		+12\left\lfloor\frac{a_4}{3}\right\rfloor<m.
		$$
		Using
		$$
		5a_5\leq 20\left\lfloor\frac{a_5}{4}\right\rfloor+15,
		\qquad
		5a_5\leq 30\left\lfloor\frac{a_5}{6}\right\rfloor+25,
		$$
		and
		$$
		2a_2+4a_4
		\leq
		6\left\lfloor\frac{a_2}{3}\right\rfloor
		+
		12\left\lfloor\frac{a_4}{3}\right\rfloor
		+12,
		$$
		we obtain
		$$
		a_1+2a_2+4a_4+5a_5<m+15,
		\qquad
		a_1+3a_3+6a_6+5a_5<m+25,
		$$
		and
		$$
		a_1+3a_3+6a_6+2a_2+4a_4<m+12.
		$$
		Adding these three inequalities gives
		$$
		3a_1+4a_2+6a_3+8a_4+10a_5+12a_6<3m+52.
		$$
		Therefore
		$$
		a_1+2a_2+3a_3+4a_4+5a_5+6a_6
		<
		\frac{3}{2}m+26.
		$$
		Since $\operatorname{lcm}_{i=1}^{6}i=60\mid m$, we have $m\geq 60$, and
		hence
		$$
		\sum_{i=1}^{6}ia_i<\frac{3}{2}m+26<2m.
		$$
		This contradicts $\sum_{x\in S}x=\sum_{i=1}^{6}ia_i\geq 2m$.
		
		Finally, since
		$
		\operatorname{lcm}_{i=1}^{5}i
		=
		\operatorname{lcm}_{i=1}^{6}i,
		$
		the case $k=5$ follows from the case $k=6$ by setting $a_6=0$.
		
		Thus in every case there exists a nonempty proper sub-multiset $T$ with
		$
		m\mid \sum_{x\in T}x.
		$
		Therefore $\mathcal{F}(m,k)\leq 2m$.
	\end{proof}
	\begin{corollary}
		\label{cor:sub_find}
		For positive integers $m,k$, let $S$ be an integer multiset with
		$|x|\leq k$ for all $x\in S$ and
		$
		\sum_{x\in S}x=cm
		$
		for some integer $c>k$. Then $S$ contains a sub-multiset $T$ such that
		$$
		m\mid \sum_{x\in T}x
		\quad\text{and}\quad
		0<\sum_{x\in T}x<\sum_{x\in S}x.
		$$
	\end{corollary}
	\begin{proof}
		Write $S=\{s_1,\ldots,s_t\}$ with $s_1\geq s_2\geq \cdots \geq s_t$.
		Let $p$ be the smallest index such that $\sum_{i=1}^{p}s_i\geq cm$.
		Then $\sum_{i=1}^{p-1}s_i<cm$ and hence $s_p>0$. Therefore $s_i>0$ for every $1\leq i\leq p$.
		
		Let $S'=\{s_1,\ldots,s_p\}$.
		Then every element of $S'$ lies in $[1,k]$, and
		$$
		\sum_{x\in S'}x\geq cm\geq (k+1)m\geq mk+1.
		$$
		By Lemma \ref{lem:add_com_upp}, $S'$ contains a nonempty proper sub-multiset
		$T$ such that $m\mid \sum_{x\in T}x$.
		Since $T$ is nonempty and all elements of $S'$ are positive, $\sum_{x\in T}x>0$.
		Moreover, since $T$ is proper in $S'$,
		$$
		\sum_{x\in T}x
		\leq
		\sum_{x\in S'}x-\min_{x\in S'}x
		=
		\sum_{i=1}^{p-1}s_i
		<cm
		=
		\sum_{x\in S}x.
		$$
		Thus $T$ satisfies the desired conditions.
	\end{proof}
	\begin{corollary}
		\label{cor:sub_find2}
		For positive integers $m,k$ with $\operatorname{lcm}_{i=1}^{k}i\big|m$, let $S$ be an integer multiset with $|x|\leq k$ for all $x\in S$ and
		$\sum_{x\in S}x=cm$ for some integer $c\geq 2$. Then $S$ contains a sub-multiset $T$ such that
		$$
		m\mid \sum_{x\in T}x
		\quad\text{and}\quad
		0<\sum_{x\in T}x<\sum_{x\in S}x.
		$$
	\end{corollary}
	\begin{proof}
		The proof is similar to that of Corollary \ref{cor:sub_find}, but uses Lemma \ref{lem:add_com_upp3_tig} instead.
	\end{proof}
	Using these lemmas, theorems, and corollaries, we can now prove the two structural theorems.
	\begin{theorem}
		\label{thm:fir_con_sec}
		For any tropical polynomial $A$, for every positive integer $d$ and
		every subpolynomial $B\in \operatorname{MS}_d(A)$, we have
		$\operatorname{cgap}(B)\leq \tdw(A).$
	\end{theorem}
	\begin{proof}
		If $A\in \mathbb T$, then $\operatorname{cgap}(B)=0=\tdw(A)$. Otherwise let
		$
		q=\tdw(A)\geq 1.
		$
		Choose a tropical polynomial decomposition
		$
		A=\sum_{t=1}^{N}W^{(t)}
		$
		with
		$
		\deg W^{(t)}\leq q.
		$
		Write
		$
		A=\min_{i=0}^{n}(a_i+ix),
		$
		and
		$
		W^{(t)}=\min_{j=0}^{q}(w^{(t)}_{j}+jx).
		$
		
		We first prove the following claim. For any $0\leq x<y\leq n$ with
		$
		d\mid (y-x)$
		and
		$(y-x)/d>q$, 
		there exists an integer $w$ such that
		$
		x<w<y,d\mid (w-x),
		$
		and $(x,w,y)$ is a convex triplet of $A$.
		
		If $a_x=\infty$ or $a_y=\infty$, this is immediate: take any
		integer $w$ with $x<w<y$ and $d\mid(w-x)$.
		
		Now suppose $a_x,a_y\ne \infty$. Choose indices
		$
		X_1,\ldots,X_N
		$ and
		$
		Y_1,\ldots,Y_N
		$
		such that
		$
		\sum_{t=1}^{N}X_{t}=x,
		\sum_{t=1}^{N}Y_{t}=y
		$
		and
		$
		a_x=\sum_{t=1}^{N}w^{(t)}_{X_{t}},
		a_y=\sum_{t=1}^{N}w^{(t)}_{Y_{t}}.
		$
		Set $\Delta_{t}=Y_{t}-X_{t}$, then $|\Delta_{t}|\leq q$ and
		$$
		\sum_{t=1}^{N}\Delta_{t}=y-x=d\cdot \frac{y-x}{d},
		\quad
		\frac{y-x}{d}>q.
		$$
		By Corollary \ref{cor:sub_find}, there exists a nonempty proper subset
		$I\subsetneq \{1,\ldots,N\}$ such that
		$
		\sigma=\sum_{t\in I}\Delta_{t}
		$
		satisfies
		$
		d\mid \sigma
		$
		and
		$
		0<\sigma<y-x.
		$
		
		Put
		$z=x+\sigma$,
		$z'=y-\sigma$.
		Then $x<z,z'<y$ and both $z-x$ and $z'-x$ are divisible by $d$.
		
		Let
		$
		E=\sum_{t\in I}\left(w^{(t)}_{Y_{t}}-w^{(t)}_{X_{t}}\right),
		$
		since
		$$
		z=x+\sum_{t\in I}\Delta_t=\sum_{t\in I}Y_t+\sum_{t\not \in I}X_t
		\quad\text{and}\quad
		z'=y-\sum_{t\in I}\Delta_t=\sum_{t\in I}X_t+\sum_{t\not \in I}Y_t,
		$$
		we have
		$$
		a_z\leq \sum_{t\in I}w^{(t)}_{Y_t}+\sum_{t\not \in I}w^{(t)}_{X_t}=a_x+E
		\quad\text{and}\quad
		a_{z'}\leq \sum_{t\in I}w^{(t)}_{X_t}+\sum_{t\not \in I}w^{(t)}_{Y_t}=a_y-E.
		$$
		
		If
		$
		(y-x)E\leq \sigma(a_y-a_x),
		$
		then
		$$
		(y-x)a_z
		\leq
		(y-x)(a_x+E)
		\leq
		(y-x)a_x+\sigma(a_y-a_x)
		=
		(y-z)a_x+(z-x)a_y.
		$$
		Thus $(x,z,y)$ is a convex triplet of $A$.
		
		Otherwise,
		$
		(y-x)E>\sigma(a_y-a_x),
		$
		and hence
		$$
		(y-x)a_{z'}\leq (y-x)(a_y-E)<(y-x)a_y-\sigma(a_y-a_x)
		=(y-z')a_x+(z'-x)a_y.
		$$
		Thus $(x,z',y)$ is a convex triplet of $A$. This proves the claim.
		
		Now for each $0\leq r<d$, let $B=\operatorname{ms}_{d,r}(A)$. If $B=\infty$, then $\operatorname{cgap}(B)=0\leq \tdw(A)$; otherwise write
		$B=
		\min_{i=0}^{\left\lfloor(n-r)/d\right\rfloor}
		(a_{r+di}+ix)
		$. Take $0\leq u<v\leq \deg B$ with
		$
		v-u>q.
		$
		
		Applying the claim to
		$x=r+du$,
		$y=r+dv$, 
		we obtain an index $w$ with
		$r+du<w<r+dv$, $d\mid(w-r)$
		such that $(r+du,w,r+dv)$ is a convex triplet of $A$. Write
		$
		w=r+de.
		$
		Then $u<e<v$, and dividing the convex-triplet inequality by $d$ gives
		that $(u,e,v)$ is a convex triplet of $B$. Thus, by definition of the convex gap,
		$$\operatorname{cgap}(B)\leq q=\operatorname{tdw}(A).$$
	\end{proof}
	This theorem combines convexity and modularity in a natural way. The modular restriction on the indices is essential: an arbitrary subpolynomial of $A$ need not satisfy the same convex gap bound. The following example illustrates this limitation and shows that the theorem is tight in this sense.
	\begin{example}
		For $A=\min(0,2x,4x,6x)=3\min(0,2x)$, $\tdw(A)\leq 2$. However, for the subpolynomial $B=\min(0,4x)$, we have $\operatorname{cgap}(B)=4>\tdw(A)$, so the same bound fails for arbitrary subpolynomials.
	\end{example}
	\begin{theorem}
		\label{thm:sec_con_sec}
		Let $k,L$ be positive integers. If
		$
		\operatorname{lcm}_{i=1}^{k}i\mid L
		$,
		then for every tropical polynomial $A$ with $\tdw(A)\leq k$ and every subpolynomial $B\in \operatorname{MS}_L(A)$, the polynomial $B$ is convex. Conversely, if
		$
		\operatorname{lcm}_{i=1}^{k}i\nmid L
		$, 
		then there exists a tropical polynomial $A$ with $\tdw(A)\leq k$ such that some $B\in \operatorname{MS}_L(A)$ is not convex.
	\end{theorem}
	\begin{proof}
		The first statement follows by the same argument as in Theorem
		\ref{thm:fir_con_sec}, but with Corollary \ref{cor:sub_find2}. We only prove the converse.
		
		Assume that
		$
		\operatorname{lcm}_{i=1}^{k}i\nmid L.
		$
		Then there exists $d$ with $1\leq d\leq k$ and $d\nmid L$. Put
		$$
		g=\gcd(d,L),
		\qquad
		h=\frac{\operatorname{lcm}(L,d)}{L}=\frac{d}{g}>1.
		$$
		Let
		$
		A=(L/g)\min(0,dx).
		$
		Then $\tdw(A)\leq d\leq k$.
		
		The finite coefficients of $A$ occur exactly in degrees
		$
		0,d,2d,\ldots,\operatorname{lcm}(L,d).
		$
		Therefore
		$
		\operatorname{ms}_{L,0}(A)=\min(0,hx).
		$
		Since $h>1$, this $\min(0,hx)\in \operatorname{MS}_L(A)$ is not convex. Hence the converse follows.
	\end{proof}
	\begin{lemma}
		For any tropical polynomial $A$, we have
		$\operatorname{cgap}(A)\leq \tdw(A).$ Equality holds when $A$ is convex.
	\end{lemma}
	\begin{proof}
		For any tropical polynomial $A$, the inequality $\operatorname{cgap}(A)\le \tdw(A)$ follows directly from Theorem \ref{thm:fir_con_sec}. When $A$ is convex, if $A\in\mathbb T$, then $\operatorname{cgap}(A)=\operatorname{tdw}(A)=0$; otherwise $\deg A\geq1$, so $\operatorname{cgap}(A)\geq1$. Since $A$ is convex, the proof in Lemma \ref{lem:poly_con_lf} shows that $\tdw(A)\leq 1$, hence $\operatorname{cgap}(A)=\operatorname{tdw}(A)=1$.
	\end{proof}
	\subsection{Convex Gap}
	\label{sec:cgap}
	In the first structural theorem for tropical polynomials of bounded
	tropical decomposition width, the notion of convex gap arises naturally
	from an additive-combinatorial argument. In this subsection, we relate this
	combinatorial definition to the geometry of the convex support sequence.
	\begin{theorem}
		\label{thm:cgap}
		Let $a$
		be a tropical sequence with $a\not \in \mathbb T$ and
		$
		\operatorname{csupp}(a)=(p_1,\ldots,p_t).
		$
		Then
		$$
		\operatorname{cgap}(a)
		=
		\max\left(\max_{i=1}^{t-1}(p_{i+1}-p_i),1\right).
		$$
	\end{theorem}
	\begin{proof}
		If $t=1$, the statement is interpreted
		separately and is trivial. Otherwise, let $d=\max\left(\max_{i=1}^{t-1}(p_{i+1}-p_i),1\right)=\max_{i=1}^{t-1}(p_{i+1}-p_i)$. 
		
		We first prove that $\operatorname{cgap}(a)\leq d$. Let
		$0\leq x<y\leq |a|$ with $y-x>d$. If $a_x=\infty$ or $a_y=\infty$, then
		$(x,x+1,y)$ is a convex triplet of $a$. Hence we may assume that
		$a_x,a_y\ne \infty$.
		
		Since $a_x,a_y\ne \infty$, both $x$ and $y$ lie in $[p_1,p_t]$. Moreover,
		$x$ and $y$ cannot lie in the same interval $[p_i,p_{i+1}]$, because every
		such interval has length at most $d$. Therefore there exists some convex support
		index $p_w$ such that $x<p_w<y$.
		By Lemma \ref{lem:env}, choose $\lambda\in\partial\hat f_a(p_w)$.
		Since $p_w\in\operatorname{csupp}(a)$, we have
		$\hat f_a(p_w)=a_{p_w}$. Thus, for every $0\leq r\leq |a|$,
		$$
		a_r\geq \hat f_a(r)\geq a_{p_w}+\lambda(r-p_w).
		$$
		In particular, 
		$
		a_x\geq a_{p_w}+\lambda(x-p_w),
		a_y\geq a_{p_w}+\lambda(y-p_w).
		$
		These inequalities imply
		$$
		\frac{a_{p_w}-a_x}{p_w-x}
		\leq
		\lambda
		\leq
		\frac{a_y-a_{p_w}}{y-p_w}.
		$$
		Equivalently,
		$
		(y-x)a_{p_w}
		\leq
		(y-p_w)a_x+(p_w-x)a_y.
		$
		Hence $(x,p_w,y)$ is a convex triplet of $a$. Therefore
		$\operatorname{cgap}(a)\leq d$.
		
		It remains to prove the reverse inequality. Choose $q$ such that
		$
		p_{q+1}-p_q=d.
		$
		For every integer $z$ with
		$
		p_q<z<p_{q+1},
		$
		we have $z\notin\operatorname{csupp}(a)$. Hence
		$$
		a_z>\hat f_a(z)
		=
		\frac{p_{q+1}-z}{p_{q+1}-p_q}a_{p_q}
		+
		\frac{z-p_q}{p_{q+1}-p_q}a_{p_{q+1}}.
		$$
		Thus
		$$
		(p_{q+1}-p_q)a_z
		>
		(p_{q+1}-z)a_{p_q}
		+
		(z-p_q)a_{p_{q+1}},
		$$
		so $(p_q,z,p_{q+1})$ is not a convex triplet. Therefore, $\operatorname{cgap}(a)\geq d$.
		
		Combining the two inequalities gives
		$$
		\operatorname{cgap}(a)=d=\max\left(\max_{i=1}^{t-1}(p_{i+1}-p_i),1\right).
		$$
	\end{proof}
	\begin{remark}
		The theorem shows that although the convex gap is defined combinatorially,
		it has a simple geometric meaning: it is the maximum distance between two
		adjacent indices in the convex support sequence. Thus a small convex gap
		means that the coefficient sequence is close to being convex, while a large
		convex gap indicates long intervals where the coefficients do not touch the
		lower convex envelope.
	\end{remark}
	\begin{corollary}
		For any tropical sequence $a$, $\operatorname{cgap}(a)$ can be computed in time $O(|a|)$.
	\end{corollary}
	\begin{proof}
		If $a\in \mathbb T$, then returning $0$ gives the correct answer. Otherwise, we first compute $\operatorname{csupp}(a)$ in $O(|a|)$ time, and then apply the formula from Theorem \ref{thm:cgap} to obtain $\operatorname{cgap}(a)$. The total running time is $O(|a|)$.
	\end{proof}
	The geometric interpretation above makes convex gap a useful measure for how frequently the lower convex envelope touches the coefficient sequence. In particular, when the input sequences have small convex gap, a weak convex support sequence of their $(\min,+)$ convolution yields a relatively dense set of exact output positions. These positions serve as anchor points, which are then used as core positions in the single-entry algorithm for Multiple-Sequence $(\min,+)$ Convolution (Theorem \ref{thm:mult_sing_sec}). We formalize the required computation in the following problem.
	\begin{problem}[computing a weak convex support sequence]
		\label{prb:weak_css}
		Given $N$ non-$\infty$ tropical sequences $a^{(1)},\ldots,a^{(N)}$, let
		$
		c=\bigotimes_{i=1}^{N}a^{(i)}.
		$
		The goal is to find a weak convex support sequence $q$ of $c$ such that
		$$
		\max_{j=1}^{|q|-1}(q_{j+1}-q_j)
		\leq
		\max_{i=1}^{N}\left(\operatorname{cgap}(a^{(i)})\right),
		$$
		where the maximum over an empty set is $0$, and a sequence $w$ of length $|q|$ such that
		$
		w_j=c_{q_j}
		$ for all $1\leq j\leq |q|$.
	\end{problem}
	We now give an algorithm for Problem \ref{prb:weak_css} running in
	$$
	O\left(\left(\sum_{i=1}^{N}|a^{(i)}|\right)\log N\right)
	$$
	time.
	\begin{algorithm}[H]
		\caption{Computing a weak convex support sequence}
		\label{alg1}
		\begin{algorithmic}[1]
			\REQUIRE Tropical sequences $a^{(1)},\ldots,a^{(N)}$
			\ENSURE A weak convex support sequence $q$ and values $w_j=c_{q_j}$
			\FOR{$i=1$ \TO $N$}
			\STATE Compute $p^{(i)}=\operatorname{csupp}(a^{(i)})$
			\STATE Write $p^{(i)}=(p^{(i)}_1,\ldots,p^{(i)}_{r_i})$
			\STATE Set
			$
			s_i\leftarrow p^{(i)}_1,\quad h_i\leftarrow a^{(i)}_{p^{(i)}_1}
			$
			\STATE Form the expression
			$$
			\operatorname{Expr}_i\leftarrow 
			\left(
			(s_i,h_i),
			\left(
			\frac{a^{(i)}_{p^{(i)}_2}-a^{(i)}_{p^{(i)}_1}}
			{p^{(i)}_2-p^{(i)}_1},
			p^{(i)}_2-p^{(i)}_1
			\right),
			\ldots,
			\left(
			\frac{a^{(i)}_{p^{(i)}_{r_i}}-a^{(i)}_{p^{(i)}_{r_i-1}}}
			{p^{(i)}_{r_i}-p^{(i)}_{r_i-1}},
			p^{(i)}_{r_i}-p^{(i)}_{r_i-1}
			\right)
			\right)
			$$
			\ENDFOR
			\STATE Let
			$
			\operatorname{Expr}_i=
			((s_i,h_i),(k_{i,1},b_{i,1}),\ldots,(k_{i,m_i},b_{i,m_i}))
			$
			\STATE Merge the $N$ sorted lists
			$
			\{\left((k_{i,1},b_{i,1}),\ldots,(k_{i,m_i},b_{i,m_i})\right)\big|1\leq i\leq N\}
			$
			by nondecreasing slope, preserving the relative order inside each list.
			\STATE Let the merged list be
			$
			((K_1,B_1),\ldots,(K_t,B_t))
			$
			\STATE Set
			$
			S\leftarrow \sum_{i=1}^{N}s_i
			$,\quad 
			$
			H\leftarrow \sum_{i=1}^{N}h_i
			$
			\STATE Initialize $q\leftarrow (S),\quad w\leftarrow (H)$
			\FOR{$j=1$ \TO $t$}
			\STATE $S\leftarrow S+B_j$
			\STATE $H\leftarrow H+K_jB_j$
			\STATE Append $S$ to $q$ and $H$ to $w$
			\ENDFOR
			\RETURN $q,w$
		\end{algorithmic}
	\end{algorithm}
	\begin{theorem}
		\label{thm:alg_weak_css}
		Algorithm \ref{alg1} correctly finds a weak convex support sequence $q$ and a sequence $w$ satisfying the condition in Problem \ref{prb:weak_css}, and runs in $$O\left(\left(\sum_{i=1}^{N}|a^{(i)}|\right)\log N\right)$$ time.
	\end{theorem}
	\begin{proof}
		Let $c=\bigotimes_{i=1}^{N}a^{(i)}$. For each $1\leq i\leq N$, the algorithm computes an expression
		$$
		\operatorname{Expr}_i=
		((s_i,h_i),(k_{i,1},b_{i,1}),\ldots,(k_{i,m_i},b_{i,m_i}))
		$$
		of $\hat f_{a^{(i)}}$. By Lemma \ref{lem:expr}, merging all pairs
		in order gives an expression
		$$
		\operatorname{Expr}=
		\left(
		\left(\sum_{i=1}^{N}s_i,\sum_{i=1}^{N}h_i\right),
		(K_1,B_1),\ldots,(K_t,B_t)
		\right)
		$$
		of $\hat f_c$.
		
		We now prove the theorem in four parts.
		
		$\textbf{Part 1}$. First, we show that the sequence $q$ found by Algorithm \ref{alg1} satisfies
		$
		\max_{j=1}^{|q|-1}(q_{j+1}-q_j)
		\leq
		\max_{i=1}^{N}\left(\operatorname{cgap}(a^{(i)})\right).
		$ Every gap $q_{j+1}-q_j$ is equal to some segment length $B_u$,
		which comes from some $\operatorname{Expr}_i$. By Theorem \ref{thm:cgap},
		$
		B_u\leq \operatorname{cgap}(a^{(i)}).
		$
		Therefore
		$$
		\max_{j=1}^{|q|-1}(q_{j+1}-q_j)
		\leq
		\max_{i=1}^{N}\left(\operatorname{cgap}(a^{(i)})\right).
		$$
		
		$\textbf{Part 2}$. Next, we prove that $w_j=c_{q_j}$ for all $1\leq j\leq |q|$. Since $(q_j,w_j)$ lies on
		$\hat f_c$, we have
		$
		w_j=\hat f_c(q_j)\leq c_{q_j}.
		$
		On the other hand, by the construction of the merged expression, each
		$q_j$ is obtained as
		$
		q_j=\sum_{i=1}^{N}u_i
		$
		for some convex support indices $u_i\in\operatorname{csupp}(a^{(i)})$, and
		$
		w_j=\sum_{i=1}^{N}a^{(i)}_{u_i}.
		$
		Thus, by the definition of $(\min,+)$ convolution,
		$$
		c_{q_j}\leq \sum_{i=1}^{N}a^{(i)}_{u_i}=w_j.
		$$
		Hence
		$
		w_j=c_{q_j}.
		$
		
		$\textbf{Part 3}$. It remains to prove that $q$ is a weak convex support sequence of $c$.
		From the previous step, for every $q_j\in q$ we have
		$
		w_j=c_{q_j}.
		$
		On the other hand, by the construction of the expression of $\hat f_c$,
		we also have
		$
		\hat f_c(q_j)=w_j=c_{q_j}.
		$
		Therefore
		$
		q\subseteq \operatorname{csupp}(c).
		$
		
		Now let $x$ be an integer such that
		$
		c_x\neq\infty
		$,
		$
		|\partial \hat f_c(x)|\geq 2.
		$
		Pick
		$
		d_0,d_1\in\partial\hat f_c(x)
		$ with
		$
		d_0<d_1.
		$
		By Lemma \ref{lem:expr}, $\hat f_c=\square_{i=1}^{N}\hat f_{a^{(i)}}$. Then we can conclude that there exist $x_1,\ldots,x_N$ such that
		$$
		\hat f_c(x)=\sum_{i=1}^{N}\hat f_{a^{(i)}}(x_i) \quad\text{and}\quad
		\sum_{i=1}^{N}x_i=x.
		$$
		We claim that
		$
		d_0,d_1\in\partial \hat f_{a^{(i)}}(x_i)
		$
		for all $1\leq i\leq N$, which will imply $|\partial \hat{f}_{a^{(i)}}(x_i)|\geq 2$.
		Indeed, if for some $i$ and some $y$,
		$
		\hat f_{a^{(i)}}(y)<\hat f_{a^{(i)}}(x_i)+d_0(y-x_i),
		$
		then
		$$
		\begin{aligned}
		\hat f_c(x-x_i+y)
		&\leq
		\hat f_{a^{(i)}}(y)+\sum_{1\leq j\leq N,j\neq i}\hat f_{a^{(j)}}(x_j)\\
		&<\sum_{j=1}^{N}\hat f_{a^{(j)}}(x_j)+d_0(y-x_i)\\
		&=
		\hat f_c(x)+d_0(y-x_i),
		\end{aligned}
		$$
		contradicting $d_0\in\partial\hat f_c(x)$. Thus
		$d_0\in\partial\hat f_{a^{(i)}}(x_i)$ for all $1\leq i\leq N$. The proof for $d_1$ is the
		same.
		
		For each $1\leq i\leq N$, write
		$
		\operatorname{Expr}_i=
		((s_i,h_i),(k_{i,1},b_{i,1}),\ldots,(k_{i,m_i},b_{i,m_i})).
		$
		Since $|\partial\hat f_{a^{(i)}}(x_i)|\geq 2$, by Definition \ref{def:conv_par}, $x_i\in \operatorname{csupp}(a^{(i)})$. And by the definition of expression there
		exists $0\leq u_i\leq m_i$ such that
		$
		x_i=s_i+\sum_{v=1}^{u_i}b_{i,v}.
		$
		
		Moreover, by the definition of expression, we have
		$k_{i,v}\leq d_0$
		for all $v\leq u_i$
		and
		$k_{i,v}\geq d_1$
		for all $v>u_i$. Hence, in the merged list $((K_1,B_1),\ldots,(K_t,B_t))$, every pair $(k_{i,v},b_{i,v})$ with $1\leq i\leq N,1\leq v\leq u_i$ precedes every pair with $1\leq i\leq N$, $u_i<v\leq m_i$. 
		Therefore, for some $0\leq r\leq t$,
		$$
		\bigcup_{i=1}^{N}
		\{(k_{i,v},b_{i,v})\mid 1\leq v\leq u_i\}
		=
		\{(K_1,B_1),\ldots,(K_r,B_r)\},
		$$
		where the union is taken as a multiset union. It follows that
		$
		x
		=
		\sum_{i=1}^{N}x_i
		=
		\sum_{i=1}^{N}s_i+\sum_{j=1}^{r}B_j.
		$
		By the construction of the algorithm, every index of the form
		$
		\sum_{i=1}^{N}s_i+\sum_{j=1}^{r}B_j
		$
		is appended to $q$. Hence $x\in q$.
		
		Thus every integer $x$ with $c_x\neq\infty$ and
		$|\partial\hat f_c(x)|\geq 2$ lies in $q$. Together with
		$q\subseteq\operatorname{csupp}(c)$, this proves that $q$ is a weak
		convex support sequence of $c$.
		
		$\textbf{Part 4}$. Finally, we analyze the running time. Computing all convex support sequences takes
		$
		O\left(\sum_{i=1}^{N}|a^{(i)}|\right)
		$
		time. The total number of pairs in lists is
		$
		O\left(\sum_{i=1}^{N}|a^{(i)}|\right),
		$
		and since each list is already sorted, the merged list is obtained by an
		$N$-way merge in
		$
		O\left(\left(\sum_{i=1}^{N}|a^{(i)}|\right)\log N\right)
		$
		time. The remaining steps are linear in the output size. Therefore the
		algorithm runs in
		$$
		O\left(\left(\sum_{i=1}^{N}|a^{(i)}|\right)\log N\right)
		$$
		time.
	\end{proof}
	\begin{remark}
		\label{rmk:sol_rec}
		Algorithm \ref{alg1} can also support solution recovery with only minor bookkeeping. During the merge step, we record for each merged segment the index of the tropical sequence from which it comes. Then, for each computed pair $(q_t,w_t)$, we obtain a solution by assigning to the $i$-th term the value of the starting position $s_i$ plus the total horizontal length of the merged segments from sequence $a^{(i)}$ that appear before $q_t$ is appended.
		
		By construction, these positions sum to $q_t$, and the corresponding input values sum to $w_t=c_{q_t}$. Thus Algorithm \ref{alg1} can output both the values on the weak convex support sequence and their corresponding solutions, without changing its asymptotic running time.
	\end{remark}
	This algorithm admits a clear geometric interpretation. Figure \ref{fig:conv} illustrates an example with $N=2$, $a^{(1)}=(2,5,3,0,4)$, and $a^{(2)}=(3,0,0,4,2)$. The top two plots depict the geometric shapes of $a^{(1)}$ and $a^{(2)}$, respectively, while the bottom plot shows the shape of $a^{(1)}\otimes a^{(2)}$.
	\begin{figure}[H]
		\centering
		\includegraphics[width=0.55\linewidth]{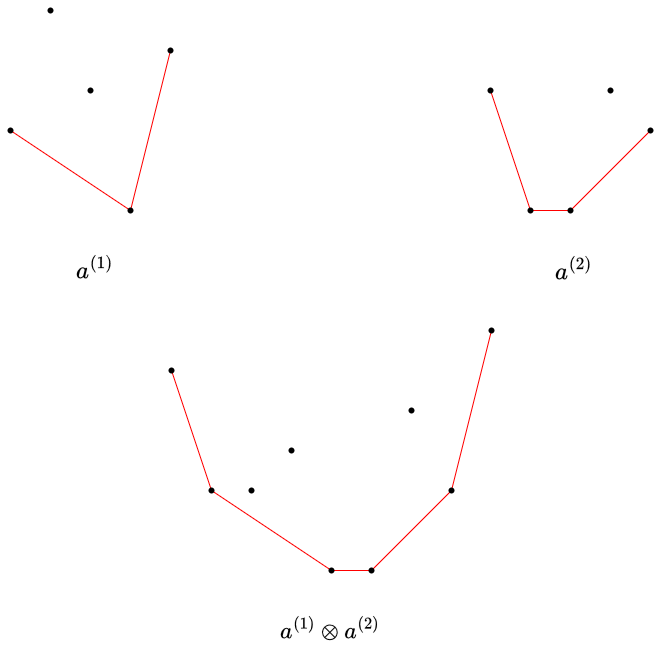}
		\caption{Geometric interpretation of Algorithm \ref{alg1}}
		\label{fig:conv}
	\end{figure}
	For a tropical sequence $b$, the black points mark the coordinates $(i,b_i)$, and the red lines indicate the segments in the expression of $\hat f_b$. A point $(i,b_i)$ lies on a red segment precisely when $i$ belongs to the convex support sequence of $b$. As the figure demonstrates, the segments in the expression of $\hat f_{a^{(1)}\otimes a^{(2)}}$ are obtained by concatenating the segments in the expressions of $\hat f_{a^{(1)}}$ and $\hat f_{a^{(2)}}$, arranged in nondecreasing order of slope.
	
	As an immediate consequence of Algorithm \ref{alg1}, we obtain the
	following property of convex gap.
	\begin{corollary}
		\label{cor:cgap}
		For $N$ tropical sequences $a^{(1)},a^{(2)},\ldots,a^{(N)}$, we have
		$$
		\operatorname{cgap}\left(\bigotimes_{i=1}^{N}a^{(i)}\right)
		\leq
		\max_{i=1}^{N}\left(\operatorname{cgap}(a^{(i)})\right).
		$$
	\end{corollary}
	\begin{proof}
		If $\max_{i=1}^{N}\left(\operatorname{cgap}(a^{(i)})\right)=0$, then $a^{(i)}\in \mathbb T$ for all $1\leq i\leq N$, hence $\bigotimes_{i=1}^{N}a^{(i)}\in \mathbb T$, which implies $\operatorname{cgap}\left(\bigotimes_{i=1}^{N}a^{(i)}\right)=0$, proving the statement; otherwise, let
		$
		b=\bigotimes_{i=1}^{N}a^{(i)}.
		$
		By Algorithm \ref{alg1} and Theorem \ref{thm:alg_weak_css}, we obtain a weak convex
		support sequence $q$ of $b$ such that
		$$
		\max_{i=1}^{|q|-1}(q_{i+1}-q_i)
		\leq
		\max_{i=1}^{N}\left(\operatorname{cgap}(a^{(i)})\right).
		$$
		Let
		$
		q'=\operatorname{csupp}(b).
		$
		Since $q$ is a weak convex support sequence, we have $q\subseteq q'$ and the startpoints and endpoints of $q$ and $q'$ coincide. Hence refining $q$ to $q'$ would not increase the maximum adjacent gap. By Theorem \ref{thm:cgap},
		$$
		\operatorname{cgap}(b)
		=
		\max\left(\max_{i=1}^{|q'|-1}(q'_{i+1}-q'_i),1\right)
		\leq
		\max\left(\max_{i=1}^{|q|-1}(q_{i+1}-q_i),1\right)
		\leq
		\max_{i=1}^{N}\left(\operatorname{cgap}(a^{(i)})\right).
		$$
	\end{proof}
	\begin{corollary}
		For every nonnegative integer $k$, both $\mathcal{T}_k$ and $\mathcal{C}_k$ are multiplicatively closed in $\mathbb{T}[x]$.
	\end{corollary}
	\begin{proof}
		The first set is multiplicatively closed by the definition of tropical
		decomposition width. The second set is multiplicatively closed by Corollary \ref{cor:cgap}.
	\end{proof}
	\section{Algorithms for $(\min,+)$ Convolution-Type Problems}
	\label{sec:alg}
	\subsection{$(\min,+)$ Convolution for Sequences of Low Tropical Decomposition Width}
	Now we have several structural properties for convex gap and tropical
	decomposition width. These properties allow us to design and analyze
	combinatorial algorithms more naturally. To obtain a faster $(\min,+)$
	convolution algorithm for sequences of low tropical decomposition width,
	we need the following additive-combinatorial lemma.
	\begin{lemma}
		\label{lem:sub_find3}
		Let $k,d$ be positive integers. Let $S$ be an integer multiset satisfying
		\begin{itemize}
			\item $|x|\leq k$ for all $x\in S$,
			\item $\left|\sum_{x\in S}x\right|\leq d$,
			\item $\sum_{x\in T}x\neq 0$ for every nonempty sub-multiset
			$T$ of $S$.
		\end{itemize}
		Then
		$
		|S|\leq \max(2k,d).
		$
	\end{lemma}
	\begin{proof}
		If $S=\varnothing$, the statement is immediate. Otherwise, define $D=\sum_{x\in S}x$.
		Since $\sum_{x\in T}x\neq 0$ for every nonempty sub-multiset $T$ of $S$, taking $T=\{x\}$ gives $x\neq 0$ for all $x\in S$, and taking $T=S$ gives $D\neq 0$. If $D<0$, replacing each element in $S$ by its negative preserves the property and changes the sign of $D$. Thus, we may assume $D>0$.
		
		We construct an ordering
		$
		c_1,c_2,\ldots,c_{|S|}
		$
		of the elements of $S$. Suppose that $c_1,\ldots,c_r$ have already been
		chosen, and write
		$
		P_r=\sum_{i=1}^{r}c_i.
		$
		
		If $P_r\leq D/2$, choose any positive remaining element. Such an element
		exists because the remaining elements have total sum $D-P_r>0$. If
		$P_r>D/2$, choose a negative remaining element if one exists; otherwise
		choose any remaining element.
		
		We claim that the inequality
		$$
		\left|P_r-\frac{D}{2}\right|
		\leq
		\max\left(k,\frac{D}{2}\right)
		$$
		holds for $0\leq r\leq |S|$ in this process. This is true for $r=0$. Suppose it holds for $r$, and let $x=c_{r+1}$.
		
		Consider three cases:
		
		$\textbf{Case 1}$. If $P_r\leq D/2$, then $0<x\leq k$, so
		$$
		\left|P_r+x-\frac{D}{2}\right|
		\leq
		\max\left(\frac{D}{2}-P_r,x\right)
		\leq
		\max\left(k,\frac{D}{2}\right).
		$$
		
		$\textbf{Case 2}.$ If $P_r>D/2$ and a negative element is chosen, then $-k\leq x<0$, and
		similarly
		$$
		\left|P_r+x-\frac{D}{2}\right|
		\leq
		\max\left(P_r-\frac{D}{2},-x\right)
		\leq
		\max\left(k,\frac{D}{2}\right).
		$$
		
		$\textbf{Case 3}.$ Finally, if $P_r>D/2$ and no negative element remains, then all remaining
		elements are positive. Hence
		$
		P_r+x\leq D,
		$
		and therefore
		$$
		\left|P_r+x-\frac{D}{2}\right|
		=
		P_r+x-\frac{D}{2}
		\leq
		\frac{D}{2}
		\leq
		\max\left(k,\frac{D}{2}\right).
		$$
		Thus the claim follows by induction.
		
		Set
		$
		M=\max\left(k,D/2\right).
		$
		Then all prefix sums
		$
		P_0,P_1,\ldots,P_{|S|}
		$
		are integers lying in the interval
		$
		\left[
		\left\lceil D/2-M\right\rceil,
		\left\lfloor D/2+M\right\rfloor
		\right].
		$
		This interval contains at most $2M+1$ integers. If
		$|S|>\max(2k,d)$, then since $D\leq d$, we have
		$
		|S|+1>\max(2k,D)+1=2M+1.
		$
		By the pigeonhole principle, there exist
		$0\leq r<s\leq |S|$ such that
		$
		P_r=P_s.
		$
		Hence
		$
		\sum_{i=r+1}^{s}c_i=0,
		$
		and then $T=\{c_i\mid r+1\leq i\leq s\}$ gives a nonempty sub-multiset of $S$ with sum zero, contradicting the assumption.
		
		Therefore
		$
		|S|\leq \max(2k,d).
		$
	\end{proof}
	\begin{definition2}[solution for $(\min,+)$ convolution]
		\label{def:sol}
		Let $a^{(1)},\ldots,a^{(N)}$ be $N$ tropical sequences, and define
		$
		b=\bigotimes_{i=1}^{N}a^{(i)}.
		$
		For $0\leq d\leq |b|$, a sequence $r=(r_1,\ldots,r_N)$ is called a
		\textit{solution} of $(a^{(1)},\ldots,a^{(N)},d)$ if
		$$
		\sum_{i=1}^{N}r_i=d
		\quad\text{and}\quad
		\sum_{i=1}^{N}a^{(i)}_{r_i}=b_d.
		$$
		For two sequences $r,r'$ of length $N$, define their \textit{difference}
		by
		$$
		D(r,r')=\sum_{i=1}^{N}[r_i\neq r'_i].
		$$
	\end{definition2}
	As a consequence of Lemma \ref{lem:sub_find3}, we obtain the following adjustment lemma for $(\min,+)$ convolution. This lemma generalizes a proximity technique that has been used in recent improvements for $0$-$1$ Knapsack \cite{prw21,clmz24a,jin24,bri24}. The connection is made precise by the reduction from $0$-$1$ Knapsack to Multiple-Sequence $(\min,+)$ Convolution, which we present in Section \ref{sec:con}.
	\begin{lemma}[adjustment lemma for $(\min,+)$ convolution]
		\label{lem:adj}
		Let $a^{(1)},\ldots,a^{(N)}$ be tropical
		sequences with $|a^{(i)}|\leq k$ for all $1\leq i\leq N$.
		Define
		$
		b=\bigotimes_{i=1}^{N}a^{(i)}.
		$
		Then for all $0\leq x,y\leq |b|$ with $|y-x|\leq 2k$ and
		$b_x\neq\infty$, if $r$ is a solution of $(a^{(1)},\ldots,a^{(N)},x)$, then
		there exists a solution $r'$ of $(a^{(1)},\ldots,a^{(N)},y)$ such that
		$$
		D(r,r')\leq 2k.
		$$
	\end{lemma}
	\begin{proof}
		Choose a solution $q$ of $(a^{(1)},\ldots,a^{(N)},y)$ minimizing $D(r,q)$.
		We prove that $D(r,q)\leq 2k$.
		
		Suppose otherwise. Let
		$
		I=\{i\mid 1\leq i\leq N,r_i\neq q_i\},
		$
		and consider the integer multiset
		$
		S=\{q_i-r_i\mid i\in I\}.
		$
		Since $|a^{(i)}|\leq k$, we have $|q_i-r_i|\leq k$ for all $i\in I$. Moreover,
		$$
		\left|\sum_{i\in I}(q_i-r_i)\right|
		=
		\left|\sum_{i=1}^{N}q_i-\sum_{i=1}^{N}r_i\right|
		=
		|y-x|
		\leq 2k.
		$$
		Since $|S|=D(r,q)>2k$, Lemma \ref{lem:sub_find3} gives a nonempty submultiset
		of $S$ with sum zero. Equivalently, there exists a nonempty subset
		$J\subseteq I$ such that
		$
		\sum_{j\in J}(q_j-r_j)=0.
		$
		
		First suppose $b_y=\infty$. Define $t$ by replacing the coordinates of
		$q$ on $J$ by those of $r$:
		$$
		t_i=
		\begin{cases}
		r_i & \text{if } i\in J,\\
		q_i & \text{if } i\notin J.
		\end{cases}
		$$
		Then $\sum_{i=1}^{N}t_i=y$ and $\sum_{i=1}^{N}a^{(i)}_{t_i}\leq \infty=b_y$. Since $b=\bigotimes_{i=1}^{N}a^{(i)}$, the inequality forces equality. Hence $\sum_{i=1}^{N}a_{t_i}^{(i)}=b_y$. This implies $t$ is also a solution of $(a^{(1)},\ldots,a^{(N)},y)$. But
		$
		D(r,t)<D(r,q),
		$
		contradicting the choice of $q$.
		
		Now assume $b_y\neq\infty$. Since $b_x\neq\infty$ and $b_y\neq\infty$,
		all entries appearing in the two solutions are finite. Put
		$$
		\Delta=\sum_{j\in J}(a^{(j)}_{q_j}-a^{(j)}_{r_j}).
		$$
		Consider two cases for $\Delta$:
		
		\textbf{Case 1}. If $\Delta<0$, define $t$ by replacing the coordinates of $r$ on $J$ by
		those of $q$:
		$$
		t_i=
		\begin{cases}
		q_i & \text{if } i\in J,\\
		r_i & \text{if } i\notin J.
		\end{cases}
		$$
		Then
		$
		\sum_{i=1}^{N}t_i
		=
		x,
		$
		but
		$
		\sum_{i=1}^{N}a^{(i)}_{t_i}
		=
		b_x+\Delta
		<
		b_x,
		$
		contradicting $b=\bigotimes_{i=1}^{N}a^{(i)}$.
		
		\textbf{Case 2}. If $\Delta\geq 0$, define $t$ by replacing the coordinates of $q$ on $J$
		by those of $r$:
		$$
		t_i=
		\begin{cases}
		r_i & \text{if } i\in J,\\
		q_i & \text{if } i\notin J.
		\end{cases}
		$$
		Then
		$
		\sum_{i=1}^{N}t_i=y
		$
		and
		$
		\sum_{i=1}^{N}a_{t_i}^{(i)}
		=
		b_y-\Delta
		\leq b_y.
		$
		Since $b=\bigotimes_{i=1}^{N}a^{(i)}$, the inequality forces equality. Hence $\sum_{i=1}^{N}a_{t_i}^{(i)}=b_y$.
		This implies $t$ is a solution of $(a^{(1)},\ldots,a^{(N)},y)$. Moreover,
		$
		D(r,t)<D(r,q),
		$
		again contradicting the minimality of $q$.
		
		Therefore $D(r,q)\leq 2k$, and taking $r'=q$ proves the theorem.
	\end{proof}
	We now turn Lemma \ref{lem:adj} into our first algorithm for $(\min,+)$
	convolution. The theorem guarantees that optimal solutions at nearby output
	indices differ by only a bounded adjustment, which allows the algorithm to
	recover each new non-$\infty$ entry by searching a local window around the
	previous solution. This proves Theorem \ref{thm:fir_alg}.
	\begin{algorithm}[H]
		\caption{The first $(\min,+)$ convolution algorithm}
		\label{alg2}
		\begin{algorithmic}[1]
			\REQUIRE Two tropical sequences $a,b$, together with
			$\tdw(a)$ and $\tdw(b)$
			\ENSURE The $(\min,+)$ convolution $c=a\otimes b$
			\IF{$a\in \mathbb T$ or $b\in \mathbb T$}
			\STATE Compute $c=a\otimes b$ by brute-force
			\RETURN $c$
			\ENDIF
			\STATE Set
			$
			K\leftarrow \max(\tdw(a),\tdw(b)),
			\quad
			U\leftarrow 2K^2
			$
			\STATE Define $c$ to be a tropical sequence with $|c|=|a|+|b|$, and initialize $c_i\leftarrow \infty$ for every $0\le i\le |c|$
			\STATE Set $L\leftarrow L(a)+L(b),\quad ps\leftarrow L(a)$
			\FOR{$i=L$ \TO $|c|$}
			\STATE
			$
			l\leftarrow \max(0,i-|b|,ps-U),
			\quad
			r\leftarrow \min(|a|,i,ps+U)
			$
			\STATE $ps'\leftarrow ps$
			\FOR{$j=l$ \TO $r$}
			\IF{$a_j+b_{i-j}<c_i$}
			\STATE $c_i\leftarrow a_j+b_{i-j},\quad ps'\leftarrow j$
			\ENDIF
			\ENDFOR
			\STATE Set $ps\leftarrow ps'$
			\ENDFOR
			\RETURN $c$
		\end{algorithmic}
	\end{algorithm}
	\begin{theorem}
		\label{thm:fir_alg_sec}
		Algorithm \ref{alg2} correctly computes the $(\min,+)$ convolution
		$c=a\otimes b$, and runs in
		$$
		O\left((|a|+|b|)\max(\tdw(a),\tdw(b))^2\right)
		$$
		time.
	\end{theorem}
	\begin{proof}
		Let
		$
		K=\max(\tdw(a),\tdw(b))
		$
		and
		$
		U=2K^2.
		$
		Let $c^*=a\otimes b$ be the true convolution.
		
		We first prove correctness. If $a\in \mathbb T$ or $b\in \mathbb T$, then the brute-force algorithm computes $a\otimes b$ in time $O(|a|+|b|)$. Otherwise, let
		$
		d_1<d_2<\cdots<d_u
		$
		be all indices $d$ such that $c^*_d\ne\infty$. Since
		$\operatorname{csupp}(c^*)$ is a subsequence of these indices with the
		same first and last indices, and since
		$$
		\operatorname{cgap}(c^*)\le \tdw(c^*)
		\le \max(\tdw(a),\tdw(b))=K,
		$$
		by Theorem \ref{thm:fir_con_sec} and Corollary \ref{cor:cgap}, we have
		$
		d_t-d_{t-1}\le K
		$
		for every $2\le t\le u$.
		
		We prove by induction on $t$ that after the algorithm finishes the
		iteration for $i=d_t$, the stored value $c_{d_t},ps$ satisfies
		$$
		c^*_{d_t}=c_{d_t}=a_{ps}+b_{d_t-ps}.
		$$
		In other words, $(ps,d_t-ps)$ is a solution of $(a,b,d_t)$.
		
		For the base case, $d_1=L(c^*)=L(a)+L(b)$. In the iteration for $i=d_1$,
		the only index that can give a finite update is $j=L(a)$. Hence the algorithm updates $c_{d_1}$ using
		$
		a_{L(a)}+b_{L(b)}=c^*_{d_1},
		$
		and stores an index $ps$ satisfying
		$
		c^*_{d_1}=c_{d_1}=a_{ps}+b_{d_1-ps}.
		$
		
		Now assume the claim holds for $d_{t-1}$, where $2\le t\le u$. Thus,
		before reaching the iteration for $i=d_t$, the stored value $ps$ satisfies
		$
		c^*_{d_{t-1}}=a_{ps}+b_{d_{t-1}-ps}.
		$
		Choose tropical sequence decompositions
		$$
		a=\bigotimes_{i=1}^{N}\alpha^{(i)}
		\quad\text{and}\quad
		b=\bigotimes_{i=1}^{M}\beta^{(i)}
		$$
		such that
		$
		|\alpha^{(i)}|,|\beta^{(i)}|\le K.
		$
		Let $(u_1,\ldots,u_N)$ be a solution of
		$(\alpha^{(1)},\ldots,\alpha^{(N)},ps)$, and let
		$(v_1,\ldots,v_M)$ be a solution of
		$(\beta^{(1)},\ldots,\beta^{(M)},d_{t-1}-ps)$. Then we have
		$
		w=(u_1,\ldots,u_N,v_1,\ldots,v_M)
		$
		is a solution of
		$
		(\alpha^{(1)},\ldots,\alpha^{(N)},
		\beta^{(1)},\ldots,\beta^{(M)},d_{t-1}).
		$
		
		Since
		$
		|d_t-d_{t-1}|\le K\le 2K
		$
		and $c^*_{d_{t-1}}\ne\infty$, Lemma \ref{lem:adj} implies that there
		exists a solution
		$
		r=(r_1,\ldots,r_{N+M})
		$
		of
		$
		(\alpha^{(1)},\ldots,\alpha^{(N)},
		\beta^{(1)},\ldots,\beta^{(M)},d_t)
		$
		such that
		$
		D(w,r)\le 2K.
		$
		Set
		$
		qs=\sum_{i=1}^{N}r_i.
		$
		Then $a_{qs}+b_{d_t-qs}=c^*_{d_t}$. Moreover,
		$$
		|qs-ps|
		=
		\left|\sum_{i=1}^{N}r_i-\sum_{i=1}^{N}u_i\right|
		\le
		\sum_{i=1}^{N}|r_i-u_i|
		\le
		K\sum_{i=1}^{N}[r_i\ne u_i]
		\le
		K\cdot D(w,r)
		\le
		2K^2
		=
		U.
		$$
		Hence the iteration for $i=d_t$ scans the index $qs$, and therefore it finds a value at most $a_{qs}+b_{d_t-qs}=c^*_{d_t}$. Moreover, all terms $a_j+b_{d_t-j}\geq c^*_{d_t}$, so the value computed by the algorithm is exactly $c^*_{d_t}$. Thus the final stored index $ps$ in this iteration also attains the minimum,
		$$c^*_{d_t}=c_{d_t}=a_{ps}+b_{d_t-ps}.$$
		This completes the induction.
		
		It remains to consider indices $i$ with $c^*_i=\infty$. For such an
		index, every term $a_j+b_{i-j}$ is equal to $\infty$, so the algorithm
		never updates $c_i$ from its initial value $\infty$. Thus every entry of
		$c$ is equal to the corresponding entry of $c^*$.
		
		Finally, we bound the running time. For each output index $i$, the inner
		loop scans at most $2U+1=O(K^2)$ indices. Since there are
		$O(|a|+|b|)$ output indices, the total running time is
		$$
		O((|a|+|b|)K^2)
		=
		O\left((|a|+|b|)\max(\tdw(a),\tdw(b))^2\right).
		$$
	\end{proof}
	Next, we present another algorithm for $(\min,+)$ convolution. One of its key ingredients is a linear-time subroutine for the case where one input tropical sequence is convex. This subroutine has been employed in several recent knapsack algorithms \cite{prw21,clmz24a,jin24,bri24} via the SMAWK algorithm \cite{akmsw87}. However, few of these works carefully address the issue that SMAWK may fail when the Monge inequality involves infinite entries ($\infty$ or $-\infty$). Such infinite entries arise naturally from the convolution definition $M_{i,j}=a_j+b_{i-j}$ with $i<j$, where $b_{i-j}$ is undefined and thus conventionally set to $\infty$ or $-\infty$. For completeness, we provide a proof that rigorously handles such infinite entries.
	\begin{lemma}
		\label{lem:case_conv}
		Let $a,b$ be two tropical sequences. If one of the sequences is convex,
		then there is a deterministic algorithm that computes $a\otimes b$ in
		time $O(|a|+|b|)$.
	\end{lemma}
	\begin{proof}
		Assume that $b$ is convex; otherwise, we swap $a$ and $b$. Let
		$l=|a|+|b|$ and $s=|a|$. If either $a$ or $b$ is equal to $\infty$, then $a\otimes b$ is equal to $\infty$, and there is nothing to prove.
		
		Let $U=\max\left(\max_{0\leq i\leq |a|,a_i\neq\infty}(|a_i|),\max_{0\leq j\leq |b|,b_j\neq\infty}(|b_j|)\right)$, and choose $H>100U+100$.
		Let $p=L(b)$ and $q=|b|$. By Lemma \ref{lem:conv_equiv}, the indices $u$ with $b_u\neq \infty$ form the interval $[p,q]$. Define a new sequence
		$\hat b$ on the index range $-l,\ldots,l$ by
		$$
		\hat b_i=
		\begin{cases}
		b_p+H(p-i) & \text{if }i<p,\\
		b_i & \text{if }p\le i\le q,\\
		b_q+H(i-q) & \text{if }i>q.
		\end{cases}
		$$
		Then $\hat b$ has no entries equal to $\infty$. Moreover, the consecutive differences of $\hat b$ are nondecreasing: they are equal to
		$-H$ on the left tail and $H$ on the right tail, while on the original part
		of $b$ this follows from Lemma \ref{lem:conv_equiv}. The two junctions are valid by
		the choice of $H$.
		
		Consider the implicit $(l+1)\times(s+1)$ matrix $M$, with rows indexed by
		$i=0,\ldots,l$ and columns indexed by $j=0,\ldots,s$, defined by
		$$
		M_{i,j}=a_j+\hat b_{i-j}.
		$$
		Since $\hat b$ is finite on the whole range $-l,\ldots,l$, we have
		$M_{i,j}=\infty$ if and only if $a_j=\infty$.
		
		Let $m_d$ be the minimum entry in row $d$ of $M$. We claim that
		$m_d$ determines $(a\otimes b)_d$ as follows:
		$$
		(a\otimes b)_d=
		\begin{cases}
		m_d & \text{if }m_d<H/2,\\
		\infty & \text{if }m_d\ge H/2.
		\end{cases}
		$$
		By the choice of $H$, every sum $a_x+b_y$ with
		$a_x,b_y\ne\infty$ has absolute value less than $H/2$, while every sum $a_x+\hat{b}_y$ using
		a padded entry of $\hat b$ is at least $H/2$.
		Hence $m_d<H/2$ exactly when $(a\otimes b)_d\neq \infty$, and in this case
		$m_d=(a\otimes b)_d$; otherwise $(a\otimes b)_d=\infty$. Thus it suffices to
		compute all row minima of $M$. 
		
		We claim that $M$ is totally monotone for row minima, where ties are
		broken in favor of the leftmost minimum. Equivalently, for any
		$0\le i\le l-1$ and $j<k$, if $M_{i,j}>M_{i,k}$, then
		$M_{i+1,j}>M_{i+1,k}$.
		
		Suppose, for contradiction, that
		$M_{i,j}>M_{i,k}$ and $M_{i+1,j}\leq M_{i+1,k}$ for some $j<k$. Since
		$M_{i,j}=\infty$ if and only if $a_j=\infty$, these inequalities can hold
		only when
		$
		M_{i,j},M_{i,k},M_{i+1,j},M_{i+1,k}\ne\infty.
		$
		Hence all terms above are finite. By the definition of $M$, we have
		$$
		a_j+\hat b_{i-j}>a_k+\hat b_{i-k}
		\quad\text{and}\quad
		a_k+\hat b_{i-k+1}\ge a_j+\hat b_{i-j+1}.
		$$
		Adding the two inequalities gives
		$$
		\hat b_{i-j}+\hat b_{i-k+1}
		>
		\hat b_{i-k}+\hat b_{i-j+1}.
		$$
		Equivalently,
		$
		\hat b_{i-k+1}-\hat b_{i-k}
		>
		\hat b_{i-j+1}-\hat b_{i-j}.
		$
		This contradicts the fact that the consecutive differences of $\hat b$ are
		nondecreasing. Hence $M$ is
		totally monotone.
		
		By the SMAWK algorithm, all row minima of a totally
		monotone matrix with constant-time entry access can be computed in time
		linear in the number of rows and columns. Thus all entries of
		$a\otimes b$ can be computed in $O(l+s)=O(|a|+|b|)$ time.
	\end{proof}
	We now use Lemma \ref{lem:case_conv} as a subroutine to give our second algorithm for $(\min,+)$ convolution. The idea is to choose a
	modulus $k$ such that all modular subsequences of one input become convex,
	compute the convolution on each pair of residue classes, and merge the
	results into the final sequence.
	\begin{algorithm}[H]
		\caption{The second $(\min,+)$ convolution algorithm}
		\label{alg3}
		\begin{algorithmic}[1]
			\REQUIRE Two tropical sequences $a,b$
			\ENSURE The $(\min,+)$ convolution $c=a\otimes b$
			\IF{$a\in \mathbb T$ or $b\in \mathbb T$}
			\STATE Compute $c=a\otimes b$ by brute-force
			\RETURN $c$
			\ENDIF
			\STATE $k=1$
			\WHILE{true}
			\IF{all sequences in $\operatorname{MS}_k(a)$ are convex \textbf{or} all sequences in $\operatorname{MS}_k(b)$ are convex}
			\STATE \textbf{break}
			\ENDIF
			\STATE $k \leftarrow k+1$
			\ENDWHILE
			\STATE Define $c$ to be a tropical sequence with $|c|=|a|+|b|$, and initialize $c_i\leftarrow \infty$ for every $0\le i\le |c|$
			\IF{all sequences in $\operatorname{MS}_k(a)$ are convex}
			\STATE Swap $a$ and $b$
			\ENDIF
			\FOR{$r_0=0$ \TO $\min(k-1,|a|)$}
			\FOR{$r_1=0$ \TO $\min(k-1,|b|)$}
			\STATE $u\leftarrow \operatorname{ms}_{k,r_0}(a)$, \quad $v\leftarrow \operatorname{ms}_{k,r_1}(b)$
			\IF{$u\neq \infty$ and $v\neq \infty$}
			\STATE Compute $w=u\otimes v$ using the algorithm in Lemma \ref{lem:case_conv}
			\FOR{$i=0$ \TO $|w|$}
			\STATE $c_{ik+r_0+r_1}\leftarrow \min(c_{ik+r_0+r_1},w_i)$
			\ENDFOR
			\ENDIF
			\ENDFOR
			\ENDFOR
			\RETURN $c$
		\end{algorithmic}
	\end{algorithm}
	To analyze the running time of Algorithm \ref{alg3}, we use the following
	number-theoretic estimate.
	\begin{lemma}
		\label{lem:numb2}
		For a positive integer $k$, let $L_k=\operatorname{lcm}_{i=1}^{k} i$. Then
		$$
		\ln L_k = k+o(k).
		$$
	\end{lemma}
	\begin{proof}
		For each prime $p$, the exponent of $p$ in $L_k$ is
		$
		\max\left(t\mid p^t\le k\right)=\lfloor \log_p k\rfloor.
		$
		Hence
		$$
		L_k
		=
		\prod_{p\in \mathrm{prime},p\le k} p^{\lfloor \log_p k\rfloor}
		=
		\prod_{p\in \mathrm{prime},t\ge 1,p^t\le k} p.
		$$
		Taking logarithms, we get
		$$
		\ln L_k
		=
		\sum_{p\in \mathrm{prime},t\ge 1,p^t\le k} \ln p
		=
		\sum_{n=1}^{k}\Lambda(n)
		=
		\psi(k),
		$$
		where $\Lambda$ is the von Mangoldt function and $\psi$ is the second
		Chebyshev function. The classical estimate
		$
		\psi(k)=k+o(k)
		$
		is standard in analytic number theory; see, e.g., \cite{apo76}. Hence
		$$
		\ln L_k=k+o(k).
		$$
	\end{proof}
	\begin{theorem}
		\label{thm:sec_alg_sec}
		Algorithm \ref{alg3} correctly computes the $(\min,+)$ convolution
		$c=a\otimes b$, and runs in
		$$
		O\left(
		(|a|+|b|)e^{\min(\tdw(a),\tdw(b))(1+o(1))}
		\right)
		$$
		time.
	\end{theorem}
	\begin{proof}
		Let $q=\min(\tdw(a),\tdw(b))$, and let $c^*=a\otimes b$ denote the
		true convolution.
		
		We first prove correctness. If $a\in \mathbb T$ or $b\in \mathbb T$, then the brute-force algorithm computes $a\otimes b$ in time $O(|a|+|b|)$. Otherwise, let
		$L=\operatorname{lcm}_{i=1}^{q} i$. By Theorem \ref{thm:sec_con_sec}, for
		modulus $L$, at least one of $\operatorname{MS}_{L}(a)$ and
		$\operatorname{MS}_{L}(b)$ consists entirely of convex sequences.
		Since Algorithm \ref{alg3} chooses the smallest modulus with this
		property, the chosen value of $k$ satisfies $k\le L$.
		
		After the possible swap in the algorithm, every sequence in
		$\operatorname{MS}_k(b)$ is convex. Therefore, for every pair of
		residues $(r_0,r_1)$, the sequence
		$v=\operatorname{ms}_{k,r_1}(b)$ is convex. Let $u=\operatorname{ms}_{k,r_0}(a)$, If $u,v\neq \infty$, then Lemma \ref{lem:case_conv}
		correctly computes $w=u\otimes v$.
		
		We now show that the computed sequence $c$ is equal to $c^*$. First, consider
		any update made by the algorithm. Suppose $
		w_i
		=
		u_p+v_{i-p}.
		$
		Then, by the definitions of $u$ and $v$,
		$$
		w_i
		=
		u_p+v_{i-p}
		=
		a_{pk+r_0}+b_{(i-p)k+r_1}.
		$$
		From $c^*=a\otimes b$, we obtain $(c^*)_{ik+r_0+r_1}\leq a_{pk+r_0}+b_{(i-p)k+r_1}$, which implies $w_i\ge (c^*)_{ik+r_0+r_1}$. Since $c$ is updated
		only by taking minima with such values, we have $c_j\ge c^*_j$ for every
		$0\le j\le |c|$.
		
		Conversely, fix any $0\le j\le |c|$ and set $c^*_j=a_p+b_{j-p}$. If $c^*_j=\infty$, then $c^*_j=\infty\ge c_j$; otherwise, $a_p,b_{j-p}\neq \infty$. Consider
		the iteration with
		$r_0\equiv p\pmod{k}$ and $r_1\equiv j-p\pmod{k}$, where
		$0\le r_0,r_1<k$. Then
		$
		u_{(p-r_0)/k}=a_p,v_{(j-p-r_1)/k}=b_{j-p}.
		$
		Therefore,
		$$
		w_{(j-r_0-r_1)/k}
		\le
		u_{(p-r_0)/k}+v_{(j-p-r_1)/k}
		=
		a_p+b_{j-p}
		=
		c^*_j.
		$$
		By this inequality, Algorithm \ref{alg3} updates $c_j$ using a value at
		most $c^*_j$. Thus $c_j\le c^*_j$. Combining the two inequalities
		gives $c_j=c^*_j$ for every $j$, so the algorithm is correct.
		
		It remains to bound the running time. Finding the smallest $k$ satisfying the convexity condition takes $O((|a|+|b|)k)=O((|a|+|b|)L)$ time. And for each pair of residues
		$(r_0,r_1)$, Lemma \ref{lem:case_conv} computes
		$u\otimes v$ in time $O(|u|+|v|)$. Since there are $O(k^2)$ residue pairs
		and each subsequence has length $O((|a|+|b|)/k)$, the running time is
		$O((|a|+|b|)k)=O((|a|+|b|)L)$. By Lemma \ref{lem:numb2}, we have
		$
		L=e^{\ln L}=e^{q(1+o(1))}.
		$
		Hence the total running time is
		$$
		O\left((|a|+|b|)e^{q(1+o(1))}\right)
		=
		O\left((|a|+|b|)
		e^{\min(\tdw(a),\tdw(b))(1+o(1))}\right).
		$$
	\end{proof}
	\subsection{Multiple-Sequence $(\min,+)$ Convolution Problem}
	We first give a formal definition for the Multiple-Sequence $(\min,+)$ Convolution problem.
	\begin{problem}[Multiple-Sequence $(\min,+)$ Convolution]
		Given $k$ tropical sequences $a^{(1)},\ldots,a^{(k)}$ with $|a^{(i)}|\le n$, compute their $(\min,+)$ convolution $\bigotimes_{i=1}^{k}a^{(i)}$.
	\end{problem}
	Note that this gives the definition for the all-entry version. The single-entry version is to compute $\bigl(\bigotimes_{i=1}^{k}a^{(i)}\bigr)_t$ for a given target index $0\le t\le \left|\bigotimes_{i=1}^{k}a^{(i)}\right|$.
	
	By using Theorem \ref{thm:fir_alg_sec} in the previous section, we can obtain a simple deterministic divide-and-conquer algorithm for Multiple-Sequence $(\min,+)$ Convolution problem that runs in $O(kn^3\log k)$ time:
	\begin{theorem}
		\label{thm:try_alg_sec}
		The all-entry version of Multiple-Sequence $(\min,+)$ Convolution can be solved by a deterministic algorithm in time
		$$
		O\left(kn^3\log k\right).
		$$
	\end{theorem}
	\begin{proof}
		First, if some $a^{(i)}$ is equal to $\infty$, then we simply return $\infty$. Otherwise, we compute $\bigotimes_{i=1}^{k} a^{(i)}$ by a balanced
		divide-and-conquer procedure. For $1\le l\le r\le k$, let
		$\textsc{Solve}(l,r)$ be the recursive procedure that returns
		$
		\bigotimes_{i=l}^{r} a^{(i)}.
		$
		The procedure is defined as follows:
		\begin{itemize}
			\item If $l=r$, return $a^{(l)}$.
			\item Otherwise, let $m=\lfloor (l+r)/2\rfloor$. Recursively compute
			$u=\textsc{Solve}(l,m)$ and $v=\textsc{Solve}(m+1,r)$. Then compute
			$b=u\otimes v$ using Algorithm \ref{alg2}, except that we do not
			provide the exact values of $\operatorname{tdw}(u)$ and
			$\operatorname{tdw}(v)$; instead, we set the parameter $K$ in
			Algorithm \ref{alg2} to be $n$. Finally, return $b$.
		\end{itemize}
		We first prove correctness. For every recursive call, we have
		$$
		\operatorname{tdw}\left(\bigotimes_{i=l}^{r} a^{(i)}\right)
		\le
		\max_{i=l}^{r}|a^{(i)}|
		\le n.
		$$
		Therefore, in every recursive step,
		$
		\operatorname{tdw}(u),\operatorname{tdw}(v)\le n.
		$
		Since Algorithm \ref{alg2} sets the parameter $K$ to $n$, which satisfies $n\geq \max(\tdw(u),\tdw(v))$, it may scan more valid entries than necessary. This does not affect correctness. Hence each recursive step correctly computes $u\otimes v$, and by induction on $r-l+1$,
		$\textsc{Solve}(l,r)$ returns $\bigotimes_{i=l}^{r} a^{(i)}$. In
		particular, $\textsc{Solve}(1,k)$ returns
		$\bigotimes_{i=1}^{k} a^{(i)}$.
		
		It remains to bound the running time. Let $T(m)$ denote the running time
		of a call $\textsc{Solve}(l,r)$ with $r-l+1=m$. In such a call, the
		lengths of the two intermediate sequences $u$ and $v$ are both bounded by
		$O(mn)$. Since Algorithm \ref{alg2} is run with $K=n$, computing
		$u\otimes v$ takes
		$
		O(mnK^2)=O(mn^3)
		$
		time. Therefore,
		$
		T(m)=2T(m/2)+O(mn^3).
		$
		Solving this recurrence gives
		$
		T(m)=O(mn^3\log m).
		$
		Hence the total running time of the algorithm is
		$$
		O(kn^3\log k).
		$$
	\end{proof}
	To improve the time complexity of Theorem \ref{thm:try_alg_sec}, we use the random
	shuffle technique, which was used in \cite{hx24} for $0$-$1$ Knapsack. The basic idea
	is to randomly permute the input sequences and then argue that, with high
	probability, the adjustment needed between nearby states is well balanced along
	the random order. This concentration step is based on Hoeffding's inequality
	for sampling without replacement.
	\begin{lemma}[Hoeffding’s inequality without replacement, \cite{hoe63}, Section 4]
		\label{lem:hoeff}
		Given an integer multiset $S=\{s_1,\ldots,s_N\}$ with $l\leq s_i\leq r$ for all $1\leq i\leq N$. Let $\mu=(\sum_{i=1}^{N}s_i)/N$ be the population mean. Let $X_1,X_2,\ldots,X_n$ denote a random sample without replacement from $S$, set $\widetilde X=(\sum_{i=1}^{n}X_i)/n$. Then for any $\delta>0$, we have
		$$
		\Pr\left[\left|\widetilde X-\mu\right|\geq \delta\right]\leq 2\exp\left(-\frac{2n\delta^2}{(r-l)^2}\right).
		$$
	\end{lemma}
	However, in the Multiple-Sequence $(\min,+)$ Convolution problem, the solution
	for a fixed target index $x$ is not necessarily unique. Hence the solution selected
	after a random shuffle may depend on the shuffle order, and the adjustment
	cannot be viewed as being sampled from a fixed underlying solution. To eliminate
	this ambiguity, we apply the isolating lemma \cite{mvv87} by assigning independent random tie-breaking weights to
	the entries. With high probability, for each target index $x$, the perturbation
	selects a unique solution among all solutions with minimum original tropical
	value. The subsequent random shuffle then only exposes the coordinates of this
	fixed isolated solution in a random order, allowing us to apply Hoeffding's
	inequality without replacement.
	\begin{lemma}[isolation lemma, \cite{mvv87}, Lemma 1]
		\label{lem:iso}
		Let $S=\{e_1,\ldots,e_N\}$ be a finite set, and let $\mathcal F\subseteq 2^S$ be a nonempty family of subsets of $S$. For each element $e\in S$, choose its weight $w(e)$ independently and uniformly at random from $\{1,\ldots,D\}$. For every $F\in\mathcal F$, define its total weight by
		$$
		w(F)=\sum_{e\in F} w(e).
		$$
		Then, with probability at least $1-N/D$, there exists a unique set $F^\ast\in\mathcal F$ such that
		$$
		w(F^\ast)=\min_{F\in\mathcal F} w(F).
		$$
	\end{lemma}
	Since the random tie-breaking weight should be combined with the original
	tropical value, we work over a lexicographic lifting of the tropical semiring.
	\begin{definition2}[lexicographic tropical semiring]
		Let $\widetilde{\mathbb T}=\mathbb T\times \mathbb R$ be a commutative idempotent semiring. For $(a,r),(b,s)\in \widetilde{\mathbb T}$, define
		$$
		\min\left((a,r),(b,s)\right)=
		\begin{cases}
			(a,r) &\text{if }a<b,\\
			(b,s) &\text{if }a>b,\\
			(a,\min(r,s)) &\text{if }a=b\\
		\end{cases}
		$$
		and
		$$
		(a,r)+(b,s)=(a+b,r+s).
		$$
		Define $q(a,r)=(qa,qr)$ for every $q\in \mathbb R_{\geq 0},(a,r)\in \widetilde{\mathbb T}$.
		
		In particular, we identify $(\infty,x)$ with $(\infty,0)$ for all $x\in\mathbb R$, 
		so that $\widetilde{\mathbb T}$ is the quotient of $\mathbb T\times\mathbb R$ 
		by the congruence generated by $\{(\infty,x)\sim (\infty,0)\mid x\in \mathbb R\}$. 
		The additive identity is $(\infty,0)$, and the multiplicative identity is $(0,0)$.
	\end{definition2}
	\begin{remark}
		\label{rmk:lex_ext}
		The statements and proofs of Theorem \ref{thm:try_alg_sec}, together with its preceding definitions, lemmas, and theorems, can be extended from $\mathbb T$ to $\widetilde{\mathbb T}$. This is because the preceding results depend only on the following three properties of the base semiring $\mathcal R$:
		\begin{enumerate}
			\item[(1)] $\mathcal R$ is a commutative idempotent semiring whose natural order $\le$ is total,
			\item[(2)] every non-$\infty$ element of $\mathcal R$ has a multiplicative inverse,
			\item[(3)] $\mathcal R$ is an $\mathbb R_{\geq 0}$-semimodule.
		\end{enumerate}
		Whenever these three requirements are satisfied, the definitions and proofs of the preceding lemmas and theorems remain unchanged. Since $\widetilde{\mathbb T}$, equipped with the lexicographic order, satisfies these properties, Algorithm \ref{alg2} extends to sequences over $\widetilde{\mathbb T}$.
	\end{remark}
	\begin{theorem}
		\label{thm:mult_all_sec}
		The all-entry version of Multiple-Sequence $(\min,+)$ Convolution can be solved by a randomized algorithm in time
		$$
		O\left(kn^2\sqrt{\min(k,n)}\log^{1.5} (kn)\right)
		$$
		with high probability.
	\end{theorem}
	\begin{proof}
		We first describe the algorithm. If some $a^{(i)}$ is equal to $\infty$, then we simply return $\infty$. Otherwise, fix an arbitrary integer constant $c>0$, and set $W=(k(n+1))^{c+5}$. For each entry
		$a^{(i)}_j$, where $0\leq j\leq |a^{(i)}|$, independently choose
		$r^{(i)}_j$ uniformly from $\{1,\ldots,W\}$, and lift it to
		$\tilde a^{(i)}_j=(a^{(i)}_j,r^{(i)}_j)\in \widetilde{\mathbb T}$.
		We randomly permute $\tilde a^{(1)},\ldots,\tilde a^{(k)}$, and run
		the algorithm of Theorem \ref{thm:try_alg_sec} over $\widetilde{\mathbb T}$ with
		parameter
		$$
		U=10n\sqrt{(c+5)\min(k,2n)\log(kn)}+n.
		$$
		By Remark \ref{rmk:lex_ext}, the same adjustment arguments apply over the
		lexicographic tropical semiring. After computing the output sequence over
		$\widetilde{\mathbb T}$, we project each entry to its first coordinate.
		
		We prove correctness and bound the running time in four parts.
		
		\textbf{Part 1.} We prove correctness by bounding implication events along the recursion. For
		events $A$ and $B$, write $A\Rightarrow B$ for the event $\neg A\cup B$.
		Equivalently, its failure probability is $\Pr[A\wedge \neg B]$. Thus the
		bounds below are not conditional probabilities such as $\Pr[B\mid A]$, they bound the probability that
		the premise of an implication holds but the conclusion fails.
		
		Consider a non-leaf call $\textsc{Solve}(l,r)$, and let
		$m=\lfloor(l+r)/2\rfloor$. Let $C_L$ and $C_R$ be the events that the two
		recursive calls correctly compute
		$\bigotimes_{i=l}^{m}\tilde a^{(i)}$ and
		$\bigotimes_{i=m+1}^{r}\tilde a^{(i)}$, respectively. We bound the failure
		probability of the implication
		$$
		C_L\wedge C_R
		\Rightarrow
		\textsc{Solve}(l,r)\text{ correctly computes }
		\bigotimes_{i=l}^{r}\tilde a^{(i)}.
		$$
		
		Assume the premise $C_L\wedge C_R$ holds, and write
		$\tilde u=\bigotimes_{i=l}^{m}\tilde a^{(i)}$,
		$\tilde v=\bigotimes_{i=m+1}^{r}\tilde a^{(i)}$, and
		$\tilde w=\tilde u\otimes \tilde v$. Let
		$d_1<d_2<\cdots<d_s$ be all indices $d$ such that the first coordinate
		of $\tilde w_d$ is finite. As in Theorem \ref{thm:fir_alg_sec}, the first entry $\tilde w_{d_1}$ and the entries $\tilde w_d$ for which the first coordinate of $\tilde w_d$ is infinite will be computed correctly. Moreover, we have $d_t-d_{t-1}\leq n$ for every
		$2\leq t\leq s$. Hence it suffices to propagate correctness from $d_{t-1}$ to $d_t$.
		
		Fix $t$ with $2\leq t\leq s$. For every $1\leq u\leq s$, let $P_u$ be the event that $\tilde w_{d_u}$
		is computed correctly.
		
		\textbf{Part 2.} We now bound the failure probability of the implication
		$P_{t-1}\Rightarrow P_t$. We first consider the failure probability of isolating the optimum. For every $1\leq u\leq s$, let $I_u$ be the event that the solution of $(\tilde a^{(l)},\ldots,\tilde a^{(r)},d_u)$ is unique. Let
		$E=\{(i,j)\mid l\leq i\leq r,\ 0\leq j\leq |a^{(i)}|\}$, and define
		$\mathcal F_{d_u}$ to be the family
		$$
		\mathcal F_{d_u}
		=
		\left\{
		\{(i,q_i)\mid l\leq i\leq r\}\bigm|
		(q_l,\ldots,q_r)\text{ is a solution of }
		(a^{(l)},\ldots,a^{(r)},d_u)
		\right\}.
		$$
		Define the weight of each element $(i,j)\in E$ as $w((i,j))=r_j^{(i)}$, so that the weight of a set $F=\{(i,q_i)\mid l\le i\le r\}$ is $w(F)=\sum_{i=l}^{r} r_{q_i}^{(i)}$.
		Since $|E|\le k(n+1)$ and each $w((i,j))=r_j^{(i)}$ is chosen independently and uniformly from $\{1,\ldots,W\}$, Lemma \ref{lem:iso} implies that the probability that $\mathcal F_{d_u}$ is not isolated is at most $k(n+1)/W\le (kn)^{-(c+4)}$. 
		When $\mathcal F_{d_u}$ is isolated, its isolated solution is the unique solution of $(\tilde a^{(l)},\ldots,\tilde a^{(r)},d_u)$. Hence $\Pr[\neg I_u]\le (kn)^{-(c+4)}$.
		
		Now suppose $P_{t-1}$, $I_{t-1}$ and $I_t$ hold. Let
		$p=(p_l,\ldots,p_r)$ be the isolated solution of $(\tilde{a}^{(l)},\ldots,\tilde{a}^{(r)},d_{t-1})$, and let $q=(q_l,\ldots,q_r)$ be the isolated
		solution of $(\tilde{a}^{(l)},\ldots,\tilde{a}^{(r)},d_t)$. Since $d_t-d_{t-1}\leq n$, Lemma \ref{lem:adj} gives
		$D(p,q)\leq 2n$. Define a multiset $T=\{q_i-p_i\mid l\leq i\leq r,p_i\neq q_i\}$. Then
		$z=|T|\leq \min(k,2n)$, $\sum_{x\in T}x=d_t-d_{t-1}\in[0,n]$, and every
		element of $T$ lies in $[-n,n]$.
		
		Under the random shuffle, the nonzero coordinates of the difference between
		$p$ and $q$ appear in a uniformly random order. Write this order as
		$\Delta_1,\ldots,\Delta_z$. For every $1\leq h\leq z$, Lemma \ref{lem:hoeff} gives
		$$
		\Pr\left[\left|\frac{\sum_{i=1}^{h}\Delta_i}{h}-\frac{\sum_{x\in T}x}{z}\right|\geq \delta\right]\leq 2\exp\left(-\frac{2h\delta^2}{(2n)^2}\right).
		$$
		Taking $\delta=10n\sqrt{(c+5)\log(kn)/h}$, then 
		$$
		\Pr\left[
		\left|
		\sum_{i=1}^{h}\Delta_i-\frac{h}{z}\sum_{x\in T}x
		\right|
		\geq
		10n\sqrt{(c+5)h\log(kn)}
		\right]
		\leq 2\exp\left(-5(c+5)\log(kn)\right)
		\leq (kn)^{-(c+4)}.
		$$
		Since $h\leq z\leq\min(k,2n)$ and $\sum_{x\in T}x\leq n$, the choice
		$U=10n\sqrt{(c+5)\min(k,2n)\log(kn)}+n$ implies
		$$
		\begin{aligned}
		\Pr\left[\left|\sum_{i=1}^{h}\Delta_i\right|\geq U\right]&\leq \Pr\left[\left|\sum_{i=1}^{h}\Delta_i-\frac{h}{z}\sum_{x\in T}x\right|\geq U-\frac{h}{z}\sum_{x\in T}x\right]\\
		&\leq \Pr\left[\left|\sum_{i=1}^{h}\Delta_i-\frac{h}{z}\sum_{x\in T}x\right|\geq 10n\sqrt{(c+5)h\log(kn)}\right]\\
		&\leq (kn)^{-(c+4)}.
		\end{aligned}
		$$ Hence, by a union bound
		for all $1\leq h\leq z$, with probability at least $1-(kn)^{-(c+4)}\cdot z\geq 1-(kn)^{-(c+3)}$, every prefix
		sum of the shuffled difference sequence has absolute value at most $U$.
		
		On this event, as in the proof of Theorem \ref{thm:fir_alg_sec}, let $qs$ and $ps$ be the split indices of the isolated solutions of $(\tilde a^{(l)},\ldots,\tilde a^{(r)},d_t)$ and $(\tilde a^{(l)},\ldots,\tilde a^{(r)},d_{t-1})$, respectively. These two indices differ by a prefix sum of the shuffled difference sequence, so $|qs-ps|\le U$. Thus the modified Algorithm \ref{alg2} scans
		$qs$ in the interval $[ps-U,ps+U]$ and therefore computes
		$\tilde w_{d_t}$ correctly.
		
		\textbf{Part 3.} We finally bound the failure probability of the algorithm. 
		With $t$ fixed as above, let $H_t$ be the event that all prefix sums of the shuffled difference sequence are bounded by $U$ in the above process, provided that $P_{t-1},I_{t-1},I_t$ all hold. In particular, if $P_{t-1},I_{t-1},I_t$ fail, we simply define $H_t$ to be always true. From the discussion above, we have
		$$
		C_L\wedge C_R\wedge P_{t-1}\wedge I_{t-1}\wedge I_t\wedge H_t
		\Rightarrow P_t.
		$$
		Therefore, the implication $C_L\wedge C_R\wedge P_{t-1}\Rightarrow P_t$ can
		fail only through the three auxiliary bad events $\neg I_{t-1}$, $\neg I_t$ and $\neg H_t$. By the isolation bound,
		$\Pr[\neg I_{t-1}],\Pr[\neg I_t]\leq (kn)^{-(c+4)}$. Moreover, the Hoeffding bound and the union
		bound give $\Pr[\neg H_t]\leq (kn)^{-(c+3)}$. Hence
		$$
		\begin{aligned}
			\Pr\left[
			\neg(C_L\wedge C_R\wedge P_{t-1}\Rightarrow P_t)
			\right]
			&\leq \Pr[\neg I_{t-1}]+\Pr[\neg I_t]+\Pr[\neg H_t] \\
			&\leq 2(kn)^{-(c+4)}+(kn)^{-(c+3)} \\
			&\leq 2(kn)^{-(c+3)}
		\end{aligned}
		$$
		for sufficiently large $k,n$.
		
		By taking a union bound over no more than $k(n+1)$ implications in one merge step and then
		over no more than $2k$ recursive calls, the total failure probability is bounded by $2k\cdot k(n+1)\cdot 2(kn)^{-(c+3)}\leq (kn)^{-c}$ for sufficiently large $k,n$. Hence
		$\bigotimes_{i=1}^{k}\tilde a^{(i)}$ is computed correctly with probability
		at least $1-(kn)^{-c}$. Projecting every output pair to its first coordinate
		then gives the correct value of $\bigotimes_{i=1}^{k}a^{(i)}$. Therefore, the algorithm computes $\bigotimes_{i=1}^{k}a^{(i)}$ correctly with probability at least $1-(kn)^{-c}$.
		
		\textbf{Part 4.} We now analyze the running time. Let $T(m)$ be the time for a call on $m$ sequences. The
		two intermediate sequences $u$ and $v$ have length $O(mn)$, and
		$U=O(n\sqrt{\min(k,n)\log(kn)})$. Hence the modified Algorithm \ref{alg2} takes
		$O(mnU)=O(mn^2\sqrt{\min(k,n)\log(kn)})$ time, so
		$$
		T(m)=2T(m/2)+O\left(mn^2\sqrt{\min(k,n)\log(kn)}\right).
		$$
		Solving the recurrence gives
		$$
		T(k)=O\left(kn^2\sqrt{\min(k,n)\log(kn)}\log k\right)\subseteq O\left(kn^2\sqrt{\min(k,n)}\log^{1.5}(kn)\right).
		$$ Hence the total running time of the algorithm is
		$$
		O\left(kn^2\sqrt{\min(k,n)}\log^{1.5}(kn)\right).
		$$
	\end{proof}
	We now turn from the all-entry version to the single-entry version. When $k\gg n$, computing all entries of the convolution may be wasteful if only one target entry is required. In this regime, we can improve over the algorithm of Theorem \ref{thm:mult_all_sec} by using a more localized strategy. The algorithm uses Algorithm \ref{alg1} as a key subroutine: we first use it to compute a solution satisfying an appropriate proximity guarantee, and then apply the same random-shuffle technique as in Theorem \ref{thm:mult_all_sec} to control the remaining search space.
	\begin{theorem}
		\label{thm:mult_sing_sec}
		The single-entry version of Multiple-Sequence $(\min,+)$ Convolution can be solved by a randomized algorithm in time
		$$
		O\left(\left(kn+n^3\sqrt{\min(k,n)}\right)\log(kn)\right)
		$$
		with high probability.
	\end{theorem}
	\begin{proof}
		Let $b=\bigotimes_{i=1}^{k}a^{(i)}$. If some $a^{(i)}$ is equal to
		$\infty$, or if $t<L(b)=\sum_{i=1}^{k}L(a^{(i)})$, then we simply return
		$\infty$. Otherwise, fix an arbitrary integer constant $c>0$. We first apply Algorithm \ref{alg1} to compute a weak convex
		support sequence $q$ of $b$, together with the corresponding values $w$.
		By Theorem \ref{thm:cgap}, we have
		$$
		\max_{i=1}^{|q|-1}(q_{i+1}-q_i)
		\leq
		\max_{i=1}^{k}\left(\operatorname{cgap}(a^{(i)})\right)
		\leq n.
		$$
		If $t=q_d$ for some $d$, then we simply return $w_d$. Otherwise, choose
		$d$ such that $q_d<t<q_{d+1}$. Then $|t-q_d|\leq n$. By Remark
		\ref{rmk:sol_rec}, Algorithm \ref{alg1} can also recover a solution
		$p$ of $(a^{(1)},\ldots,a^{(k)},q_d)$. And by Lemma \ref{lem:adj}, there exists a solution $r$ of $(a^{(1)},\ldots,a^{(k)},t)$ such that
		$D(r,p)\leq 2n$.
		
		We now construct a small set of possible adjustments from $p$. For each
		$-n\leq \Delta \leq n$, initialize an empty multiset $S_\Delta$. For every
		$1\leq i\leq k$ and every $0\leq j\leq |a^{(i)}|$ such that $j\neq p_i,a^{(i)}_j\neq \infty$,
		insert $(i,a^{(i)}_j-a^{(i)}_{p_i})$ into $S_{j-p_i}$. For each
		$-n\leq \Delta\leq n$, sort $S_\Delta$ by increasing second coordinate. If the second
		coordinates are equal, break ties by increasing first coordinate when
		$\Delta>0$, and by decreasing first coordinate when $\Delta<0$. Then keep
		only the first $\min(2n,|S_\Delta|)$ elements in $S_\Delta$.
		
		Let $I$ be the set of all indices $i$ such that $(i,\gamma)$ is retained in
		some $S_\Delta$. For each $i\in I$, define $W_i$ to be the multiset of all pairs
		$(\Delta,\gamma)$ such that the retained element $(i,\gamma)$ belongs to
		$S_\Delta$. We also add the dummy adjustment $(0,0)$ to every $W_i$. This implies
		$$
		\sum_{i\in I}|W_i|=|I|+\sum_{\substack{-n\le \Delta\le n,\Delta\ne 0}}|S_\Delta|
		\leq 2\cdot\left(\sum_{\substack{-n\le \Delta\le n,\Delta\ne 0}}|S_\Delta|\right)
		\le 2\cdot (2n)^2
		=O(n^2).
		$$
		We then randomly permute the indices in $I$, obtaining an ordering sequence $v=(v_1,\ldots,v_z)$.
		
		Set $U=10n\sqrt{(c+5)\min(k,2n)\log(kn)}+n$. We run a dynamic program over
		$W_{v_1},\ldots,W_{v_z}$, truncating the total index adjustment to the range
		$[-U,U]$. Let $dp[h,s]$ be the minimum total value increase after processing
		the first $h$ multisets and obtaining total index adjustment $s$.
		
		We initialize the zeroth layer by setting $dp[0,0]=0$ and
		$dp[0,s]=\infty$ for every $s\neq 0$. For each $1\leq h\leq z$, we compute
		the $h$-th layer from the $(h-1)$-st layer by
		$$
		dp[h,s]=
		\min_{(\Delta,\gamma)\in W_{v_h}}
		\left(dp[h-1,s-\Delta]+\gamma\right),
		$$
		where only states with $s\in[-U,U]$ are kept, and states outside this range
		are treated as $\infty$. The algorithm finally returns
		$\sum_{i=1}^{k}a^{(i)}_{p_i}+dp[z,t-q_d]$.
		
		We prove correctness. If $b_t=\infty$, the dynamic program considers no feasible adjustment and therefore correctly returns $\infty$. Otherwise, among all solutions of $(a^{(1)},\ldots,a^{(k)},t)$
		whose distance from $p$ is at most $2n$, choose one
		$r=(r_1,\ldots,r_k)$ that is lexicographically largest among all such
		solutions. Let $J=\{i\mid 1\leq i\leq k, r_i\neq p_i\}$. Then
		$|J|=D(r,p)\leq \min(k,2n)$.
		
		We claim that for every $i\in J$, the adjustment
		$(r_i-p_i,a^{(i)}_{r_i}-a^{(i)}_{p_i})$ belongs to $W_i$. Suppose not.
		Put $\Delta=r_i-p_i$ and $\gamma=a^{(i)}_{r_i}-a^{(i)}_{p_i}$. Then
		$(i,\gamma)$ was discarded from $S_\Delta$. Hence before the discarding step, we have $|S_\Delta|>2n\geq |J|$. Then
		among the retained $2n$ elements of $S_\Delta$ there is an element
		$(i',\gamma')$ with $i'\notin J$. Since $i'\notin J$, we have $r_i'=p_i'$. Moreover, by the sorting rule, either
		$\gamma'<\gamma$, or $\gamma'=\gamma$ and the tie-breaking order places
		$i'$ before $i$.
		
		Define a new sequence $u=(u_1,\ldots,u_k)$ by
		$$
		u_j=
		\begin{cases}
			p_i & \text{if } j=i,\\
			p_{i'}+\Delta & \text{if } j=i',\\
			r_j & \text{if } j\neq i\text{ and } j\neq i'
		\end{cases}
		$$
		for all $1\leq j\leq k$. Then we have
		$$\sum_{j=1}^{k}u_j=\sum_{j=1}^{k}r_j\quad \text{ and }\quad \sum_{j=1}^{k}a_{u_j}^{(j)}-\sum_{j=1}^{k}a_{r_j}^{(j)}=\gamma'-\gamma.$$
		
		If $\gamma'<\gamma$, this contradicts the optimality of
		$r$. If $\gamma'=\gamma$, the sign-dependent tie-breaking rule makes $u$
		lexicographically larger than $r$, again a contradiction. Therefore all
		adjustments from $p$ to $r$ are present in the dynamic program.
		
		It remains to show that the truncation to $[-U,U]$ preserves the dynamic programming path corresponding to $r$ with high probability. We keep only those elements of $v$ that appear in $J$, obtaining a sequence $\bar v=(\bar v_1,\ldots,\bar v_m)$, where $m\le |J|\le \min(k,2n)$. For each $1\le h\le m$, let $\bar \Delta_h$ denote the total index adjustment from $p$ to $r$ in $W_{\bar v_h}$. Since $(v_1,\ldots,v_z)$ is a random shuffle, so is $(\bar v_1,\ldots,\bar v_m)$. Hence $(\bar \Delta_1,\ldots,\bar \Delta_m)$ is also a random shuffle. Applying the same Hoeffding inequality for sampling without replacement as in the proof of Theorem \ref{thm:mult_all_sec}, we obtain that, with probability at least $1-(kn)^{-c}$,
		$$
		\left|\sum_{h=1}^{o}\bar \Delta_h\right|\leq U
		$$
		for every $1\leq o\leq m$. Conditional on this event, the dynamic program does not truncate the path corresponding to $r$, and hence returns the correct value $b_t$.
		
		Finally, Algorithm \ref{alg1} takes $O(kn\log k)$ time. Constructing and
		sorting the multisets $S_\Delta$ takes $O(kn\log(kn))$ time. Since
		$\sum_{i\in I}|W_i|=O(n^2)$ and the DP range has size $O(U)$, where $U=O\left(n\sqrt{\min(k,n)\log(kn)}\right)$, the dynamic
		program takes $O(n^2U)=O\left(n^3\sqrt{\min(k,n)\log(kn)}\right)$ time. Hence the total running
		time is $$O\left(\left(kn+n^3\sqrt{\min(k,n)}\right)\log(kn)\right).$$
	\end{proof}
	We also give conditional lower bounds for the all-entry version of Multiple-Sequence $(\min,+)$ Convolution problem. These lower bounds are based on the $(\min,+)$ convolution hypothesis, which was explicitly introduced in \cite{cmww19}.
	\begin{hypothesis}[$(\min,+)$ convolution hypothesis]
		\label{hyp:minplus}
		For every fixed $\epsilon>0$, there is no randomized algorithm that, given two tropical sequences $a,b$ with $|a|=|b|=n$ and entries in $\{0,1,\ldots,n^{O(1)}\}$, correctly computes $a\otimes b$ in time $O(n^{2-\epsilon})$ with high probability.
	\end{hypothesis}
	\begin{theorem}
		Assuming the $(\min,+)$ convolution hypothesis, for any fixed $\epsilon>0$, the all-entry version of Multiple-Sequence $(\min,+)$ Convolution
		cannot be solved by a randomized algorithm in time $$O(kn^{2-\epsilon}+k^{1-\epsilon}n^2)$$
		with high probability, even when the sequence entries are restricted to $\{0,1,\ldots,n^{O(1)}\}$.
	\end{theorem}
	\begin{proof}
		Suppose, for contradiction, that such an algorithm exists. We show that this would give a truly subquadratic algorithm for $(\min,+)$ convolution in Hypothesis \ref{hyp:minplus}.
		
		Let $a,b$ be two tropical sequences such that $|a|=|b|=n$ and entries in $\{0,1,\ldots,n^{O(1)}\}$. Let $c^*=a\otimes b$ be the true convolution. Set $B=\lfloor\sqrt n\rfloor$. Partition $[0,n]$ into consecutive intervals
		$[l_1,r_1],[l_2,r_2],\ldots,[l_s,r_s]$
		such that $r_t+1=l_{t+1}$, $r_t-l_t+1\leq B$, and $s=O(n/B)=O(\sqrt n)$.
		
		For each $1\leq x,y\leq s$, define
		$\tilde a^{(x)}=(a_{l_x},\ldots,a_{r_x})$ and
		$\tilde b^{(y)}=(b_{l_y},\ldots,b_{r_y})$.
		List the $s^2$ block pairs
		$(\tilde a^{(x)},\tilde b^{(y)})$
		as
		$(u^{(1)},v^{(1)}),\ldots,(u^{(s^2)},v^{(s^2)})$.
		For $1\leq i\leq s^2$, let $r^{(i)}=u^{(i)}\otimes v^{(i)}$.
		
		Let $M=\max\left(\max_{i=0}^{n}(|a_i|),\max_{j=0}^{n}(|b_j|)\right)$, and choose $H>100M+100$. We construct $2s^2$ tropical sequences
		$w^{(1)},\ldots,w^{(2s^2)}$ as follows. Define $\bar w^{(2i-1)}=u^{(i)}$ and $\bar w^{(2i)}=v^{(i)}$ for $1\leq i\leq s^2$. Then for each $1\leq j\leq 2s^2$, set $|w^{(j)}|=|\bar w^{(j)}|$ and define
		$
		w^{(j)}_d=\bar w^{(j)}_d+\lceil j/2\rceil dH
		$
		for all $0\leq d\leq |w^{(j)}|$.
		
		Now compute
		$z=\bigotimes_{j=1}^{2s^2}w^{(j)}$
		using the assumed algorithm. Since $|w^{(j)}|\le B$ and $s=O(\sqrt n)$, the assumed algorithm runs in time
		$$
		O\left(s^2B^{2-\epsilon}+s^{2(1-\epsilon)}B^2\right)
		=
		O\left(n^{2-\epsilon/2}\right).
		$$
		We claim that all convolutions of block pairs $r^{(i)}$ can be recovered from $z$. Fix $1\leq i\leq s^2$, and let
		$P_i=\sum_{j=1}^{2i-2}|w^{(j)}|$ and
		$C_i=\sum_{j=1}^{2i-2}w^{(j)}_{|w^{(j)}|}+\sum_{j=2i+1}^{2s^2}w^{(j)}_0$.
		For every $0\leq d\leq |r^{(i)}|$, we show that
		$$
		z_{P_i+d}=C_i+idH+r^{(i)}_d.
		$$
		
		We now prove the claim. Let $h=(h_1,\ldots,h_{2s^2})$ be any solution of
		$(w^{(1)},\ldots,w^{(2s^2)},P_i+d)$. We first observe the following
		exchange property. Suppose that there exist an exchanging pair $(x,y)$ such that
		$\lceil x/2\rceil<\lceil y/2\rceil$, $h_x<|w^{(x)}|$, and $h_y>0$.
		Define a new sequence $h'=(h'_1,\ldots,h'_{2s^2})$ by
		$$
		h'_j=
		\begin{cases}
			h_j+1 & \text{if } j=x,\\
			h_j-1 & \text{if } j=y,\\
			h_j & \text{if } j\neq x\text{ and }j\neq y
		\end{cases}
		$$
		for all $1\leq j\leq 2s^2$. Then $\sum_{j=1}^{2s^2}h'_j=\sum_{j=1}^{2s^2}h_j=P_i+d$. The total value adjustment from $h$ to $h'$ is
		$$
		\begin{aligned}
			&\sum_{j=1}^{2s^2}w^{(j)}_{h'_j}-\sum_{j=1}^{2s^2}w^{(j)}_{h_j}  \\
			&=
			\left(w^{(x)}_{h_x+1}-w^{(x)}_{h_x}\right)
			+
			\left(w^{(y)}_{h_y-1}-w^{(y)}_{h_y}\right) \\
			&=
			\left(\bar w^{(x)}_{h_x+1}-\bar w^{(x)}_{h_x}\right)
			+
			\left(\bar w^{(y)}_{h_y-1}-\bar w^{(y)}_{h_y}\right)
			+
			\left(\left\lceil\frac{x}{2}\right\rceil-\left\lceil\frac{y}{2}\right\rceil\right)H.
		\end{aligned}
		$$
		Since all coefficients in $\bar w^{(x)},\bar w^{(y)}$ have absolute value at most $M$, the first two terms together are at most $4M$, and $(\lceil x/2\rceil-\lceil y/2\rceil)H\leq -H$. Therefore
		$
		\sum_{j=1}^{2s^2}w^{(j)}_{h'_j}-\sum_{j=1}^{2s^2}w^{(j)}_{h_j}
		\leq 4M-H<0.
		$
		This contradicts that $h$ is a solution.
		
		Therefore, no such exchanging pair $(x,y)$ can exist. Since
		$
		\sum_{i=1}^{2s^2}h_i=P_i+d\le P_i+|w^{(2i-1)}|+|w^{(2i)}|,
		$
		this exchange property implies that $h_j=|w^{(j)}|$ for all $1\le j\le 2i-2$ and $h_j=0$ for all $2i+1\le j\le 2s^2$. Hence the only nontrivial choices occur in the two active sequences $w^{(2i-1)}$ and $w^{(2i)}$, and they satisfy $h_{2i-1}+h_{2i}=d$. Thus
		$$
		z_{P_i+d}
		=
		C_i+idH+\min_{p+q=d}(u^{(i)}_p+v^{(i)}_q)
		=
		C_i+idH+r^{(i)}_d.
		$$
		Consequently, we recover $r^{(i)}_d$ by
		$r_d^{(i)}=z_{P_i+d}-C_i-idH$.
		
		Finally, we assemble $a\otimes b$. Initialize a tropical sequence $c$ with $|c|=|a|+|b|$ and $c_t=\infty$ for all $0\le t\le |c|$. For each block pair $(\tilde a^{(x)},\tilde b^{(y)})$, suppose it is listed as $(u^{(i)},v^{(i)})$. Then for every $0\le d\le |r^{(i)}|$, update $c_{l_x+l_y+d}$
		by taking the minimum with $r_d^{(i)}$. By definition, $r^{(i)}_d=(\tilde a^{(x)}\otimes \tilde b^{(y)})_d$. Every contribution $a_p+b_q$ to $(c^*)_t=\min_{p+q=t}(a_p+b_q)$ is covered by exactly one block pair $(\tilde a^{(x)},\tilde b^{(y)})$ with $l_x\le p\le r_x$ and $l_y\le q\le r_y$. Hence after all updates, $c=c^*=a\otimes b$.
		
		Therefore, the assumed algorithm gives a randomized algorithm for $(\min,+)$ convolution running in $O(n^{2-\epsilon/2})$ time, contradicting Hypothesis \ref{hyp:minplus}. This proves the theorem.
	\end{proof}
	\subsection{Connection with $(\min,+)$ Convolution and Knapsack-Type Problems}
	\label{sec:con}
	We first give the definitions of the Knapsack-type problems outlined in the introduction.
	\begin{problem}[Subset Sum]
		Given a capacity $W$ and $n$ items, where the $i$-th item has a positive integer weight $w_i$, decide whether there exists a subset $S\subseteq \{1,\ldots,n\}$ such that
		$\sum_{i\in S}w_i=W$.
	\end{problem}
	\begin{problem}[$0$-$1$ Knapsack]
		Given a capacity $W$ and $n$ items, where the $i$-th item has a positive integer weight $w_i$ and profit $p_i$, compute the maximum total profit $\sum_{i\in S} p_i$ over all subsets $S\subseteq \{1,\ldots,n\}$ such that $\sum_{i\in S} w_i \le W$.
	\end{problem}
	\begin{problem}[Multiple-Choice Knapsack]
		Given a capacity $W$ and $m$ nonempty groups $G_1,\ldots,G_m$, where each group $G_i$ consists of $k_i$ items with positive integer weights $w_1^{(i)},\ldots,w_{k_i}^{(i)}$ and profits $p_1^{(i)},\ldots,p_{k_i}^{(i)}$. Let $n=\sum_{i=1}^{m} k_i$ be the total number of items. Exactly one item must be chosen from each group. The goal is to maximize
		$$\sum_{i=1}^{m} p_{r_i}^{(i)}$$
		subject to
		$\sum_{i=1}^{m} w_{r_i}^{(i)} \le W$, where $r_i\in\{1,\ldots,k_i\}$ indicates the chosen item from group $i$. If no feasible selection exists, the answer is defined to be $-\infty$.
	\end{problem}
	For all problems considered here, let $w_{\operatorname{tot}}$ denote the total weight of all items, and let $w_{\max}$ denote the maximum item weight. Although these problems are stated for the single-entry version, the all-entry version is obtained by solving the problem for all capacities $0\le W\le w_{\operatorname{tot}}$.
	
	Subset Sum reduces to $0$-$1$ Knapsack by taking each item $w_i$ as an item with weight $w_i$ and profit $w_i$. The Subset Sum instance has a solution if and only if the optimal profit in the Knapsack instance equals $W$. In turn, $0$-$1$ Knapsack reduces to Multiple-Choice Knapsack by replacing each item $(w_i,p_i)$ with a group $\{(1,0),(w_i+1,p_i)\}$ and increasing the capacity from $W$ to $W+n$. Since exactly one item is chosen from each group, the two instances have the same optimal profit. Both reductions run in linear time and justify the hardness hierarchy described in the Introduction.
	
	We next show that Multiple-Choice Knapsack reduces to Multiple-Sequence $(\min,+)$ Convolution. Together with the reductions above, this implies that all the Knapsack-type problems considered here can be transformed to the Multiple-Sequence $(\min,+)$ Convolution problem.
	\begin{theorem}
		\label{thm:mkp_red}
		A $T(k,n)$-time algorithm for the single-entry version of Multiple-Sequence $(\min,+)$ Convolution implies an $O(n+T(m+\lceil \log w_{\max}\rceil,w_{\max}))$-time algorithm for the single-entry version of Multiple-Choice Knapsack. In particular, a $T(k,n)$-time algorithm for the all-entry version of Multiple-Sequence $(\min,+)$ Convolution implies an $O(w_{\operatorname{tot}}+T(m+\lceil \log w_{\max}\rceil,w_{\max}))$-time algorithm for the all-entry version of Multiple-Choice Knapsack.
	\end{theorem}
	\begin{proof}
		We construct an algorithm for the single-entry version of Multiple-Choice Knapsack as follows.
		
		Consider an instance of the single-entry version of Multiple-Choice Knapsack with capacity $W$. Since each of the $m$ groups contributes at most one item, and every item has weight at most $w_{\max}$, the total weight of any feasible selection is at most $mw_{\max}$. If $W>mw_{\max}$, replacing $W$ by $mw_{\max}$ does not change the answer. Thus we may assume $0\le W\le mw_{\max}$.
		
		For each group $G_i$, let $q_i$ be the maximum weight among items in $G_i$. Initialize a tropical sequence $a^{(i)}$ with $|a^{(i)}|=q_i\leq w_{\max}$ and $a^{(i)}_t=\infty$ for all $0\le t\le q_i$. For every item $(w,v)\in G_i$, update $a^{(i)}_w$ by taking the minimum with $-v$.
		
		Next, we add $t=m+\lceil \log w_{\max}\rceil$ additional tropical sequences $a^{(m+1)},\ldots,a^{(m+t)}$. For each $1\le i\le t$, define $L=\min(2^{i-1},w_{\max})$, and set $|a^{(m+i)}|= L$. Define
		$$
		a^{(m+i)}_j =
		\begin{cases}
			0 & j=0 \text{ or } j=L,\\
			\infty & 1\leq j\leq L-1.
		\end{cases}
		$$ We claim that
		$$\text{OPT}(W)=-\left(\bigotimes_{i=1}^{m+t} a^{(i)}\right)_W,$$
		where $\text{OPT}(W)$ denotes the optimal profit for capacity $W$.
		
		To prove the claim, we first establish the inequality $\text{OPT}(W)\le -\left(\bigotimes_{i=1}^{m+t} a^{(i)}\right)_W$.
		Consider an optimal selection for capacity $W$. If $\text{OPT}(W)=-\infty$, the inequality is immediate, so assume a feasible solution exists. Suppose group $G_i$ chooses item $(w^{(i)}_{s_i},p^{(i)}_{s_i})$. Let $b=\bigotimes_{i=m+1}^{m+t} a^{(i)}$.
		Then $$|b|=\sum_{i=1}^{t}\min(2^{i-1},w_{\max})\geq \sum_{i=\lceil \log w_{\max}\rceil+1}^{m+\lceil \log w_{\max}\rceil}w_{\max}=mw_{\max}\ge W,$$ and $b_j=0$ for all $0\leq j\leq |b|$. Set $R=W-\sum_{i=1}^{m} w^{(i)}_{s_i}$. Since the selection is feasible, $0\le R\le W\le |b|$. Consequently,
		$$
		-\left(\bigotimes_{i=1}^{m+t}a^{(i)}\right)_W
		=
		-\left(\bigotimes_{i=1}^{m}a^{(i)}\otimes b\right)_W
		\ge
		-\sum_{i=1}^{m} a^{(i)}_{w^{(i)}_{s_i}} - b_R
		=
		\sum_{i=1}^{m} p^{(i)}_{s_i}
		=
		\text{OPT}(W).
		$$
		
		For the reverse direction, we prove $\text{OPT}(W)\ge -\left(\bigotimes_{i=1}^{m+t} a^{(i)}\right)_W$.
		Suppose $(v_1,\ldots,v_{m+t})$ is a solution of $(a^{(1)},\ldots,a^{(m+t)},W)$. If $(\bigotimes_{i=1}^{m+t}a^{(i)})_W=\infty$, the inequality is immediate, so assume that $(\bigotimes_{i=1}^{m+t}a^{(i)})_W\neq \infty$. For each $1\le i\le m$, the finiteness of $a^{(i)}_{v_i}$ guarantees an item in group $G_i$ with weight $v_i$ and profit $-a^{(i)}_{v_i}$; denote its profit by $\tilde p_i$. The total weight of these selected items is $\sum_{i=1}^{m} v_i \le \sum_{i=1}^{m+t} v_i = W$, so the selection is feasible. Hence
		$$\operatorname{OPT}(W)
		\ge
		\sum_{i=1}^{m} \tilde p_i
		=
		-\sum_{i=1}^{m} a^{(i)}_{v_i}
		=
		-\sum_{i=1}^{m+t} a^{(i)}_{v_i}
		=
		-\left(\bigotimes_{i=1}^{m+t} a^{(i)}\right)_W.$$
		Combining the two inequalities proves the claim. This algorithm runs in $O(n+mw_{\max}+T(m+t,w_{\max}))=O(n+T(m+\lceil \log w_{\max}\rceil,w_{\max}))$ time.
		
		For the all-entry version, the same reduction applies, with the additional requirement of outputting $O(w_{\operatorname{tot}})$ numbers. This gives an algorithm in time $O\left(w_{\operatorname{tot}}+T(m+\lceil \log w_{\max}\rceil,w_{\max})\right)$.
	\end{proof}
	\begin{corollary}
		\label{cor:mkp_alg}
		The all-entry version of Multiple-Choice Knapsack can be solved by a randomized algorithm in time $\widetilde O\left(w_{\operatorname{tot}}+mw_{\max}^2\sqrt{\min(m,w_{\max})}\right)$ with high probability. The single-entry version of Multiple-Choice Knapsack can be solved by a randomized algorithm in time $\widetilde O\left(n+w_{\max}^3\sqrt{\min(m,w_{\max})}\right)$ with high probability. 
	\end{corollary}
	\begin{proof}
		The first statement follows directly from Theorem \ref{thm:mkp_red} and Theorem \ref{thm:mult_all_sec}.
		
		For the second statement, we apply the reduction from Theorem \ref{thm:mkp_red} and use the algorithm from Theorem \ref{thm:mult_sing_sec}. However, during the process we only store the finite entries of the tropical sequences. In the resulting instance, the tropical sequences contain only $O(n+\log w_{\max})$ finite entries in total, so the reduction itself requires no more than $O(n+\log w_{\max})$ time.
		
		For the algorithm from Theorem \ref{thm:mult_sing_sec}, Algorithm \ref{alg1} scans only the finite entries to construct the convex support sequences, which takes $O\left((n+\log w_{\max})\log(nw_{\max})\right)$ time. The construction and sorting of the multisets $S_{\Delta}$ also involve only the finite entries and require $O\left((n+\log w_{\max})\log(nw_{\max})\right)$ time. Hence the processing time of the algorithm in Theorem \ref{thm:mult_sing_sec} is $\widetilde O(n)$. Consequently, the total running time of the reduction is bounded by
		$$\widetilde O\left(n+w_{\max}^3\sqrt{\min(m, w_{\max})}\right).$$
	\end{proof}
	\begin{corollary}
		The all-entry version of Multiple-Choice Knapsack can be solved by a randomized algorithm in time $\widetilde O\left(nw_{\max}^2\sqrt{\min(n,w_{\max})}\right)$ with high probability, and the single-entry version of Multiple-Choice Knapsack can be solved by a randomized algorithm in time $\widetilde O\left(nw_{\max}\sqrt{\min(n,w_{\max})}\right)$ with high probability.
	\end{corollary}
	\begin{proof}
		Since $m\leq n$, the first statement is immediate from Corollary \ref{cor:mkp_alg}.
		
		For the second statement, we apply the reduction in Theorem \ref{thm:mkp_red}. In the resulting instance for the single-entry version, the tropical sequences contain only $O(n+\log w_{\max})$ finite entries in total. Consequently, in the proof of Theorem \ref{thm:mult_sing_sec}, $\sum_{i\in I}|W_i|$ is bounded by twice the total number of finite entries, hence $\sum_{i\in I}|W_i|=O(n+\log w_{\max})$. Therefore, the running time is
		$$\widetilde O\left(n+(n+\log w_{\max})w_{\max}\sqrt{\min(m,w_{\max})}\right)=\widetilde O\left(n+nw_{\max}\sqrt{\min(n,w_{\max})}\right).$$
	\end{proof}
	By the reduction in Theorem \ref{thm:mkp_red}, every instance of these Knapsack-type problems can be transformed into an instance of Multiple-Sequence $(\min,+)$ Convolution. In the resulting instance, every tropical sequence has maximum index at most $w_{\max}$. Hence, for any such Knapsack-type problem, the answer sequence
	$$b=\bigotimes_{i=1}^{N}a^{(i)}$$
	is the $(\min,+)$ convolution of tropical sequences each of which has maximum index at most $w_{\max}$. Consequently, $\operatorname{tdw}(b)\le w_{\max}$.
	
	Hence the structural theorems for bounded tropical decomposition width polynomials apply to these answer sequences. Let $b$ be the answer sequence for any of the problems above. Then for every positive integer $d$ and every arithmetic progression subsequence $c\in \operatorname{MS}_d(b)$, Theorem \ref{thm:fir_con} gives $\operatorname{cgap}(c)\le w_{\max}$. Furthermore, by Theorem \ref{thm:sec_con}, if $\operatorname{lcm}_{i=1}^{w_{\max}} i \mid L$, then every $c\in \operatorname{MS}_L(b)$ is convex.
	
	Therefore, when $w_{\max}$ is small, the answer sequences arising from these Knapsack-type problems exhibit strong regularity. They have bounded tropical decomposition width and satisfy modular convexity properties. This provides a possible structural explanation for why small-weight regimes often admit more efficient algorithms for Knapsack-type problems: the parameter $w_{\max}$ arises naturally from the perspective of tropical decomposition width.
	\section{Power of Extension Algebras}
	\label{sec:ext}
	In this section, we examine the role of extension algebras over $\mathbb T$ in interpolation and decomposition. To align the algebra with the structure of $\mathbb T$, we require that extension algebras be $\mathbb T$-algebras. The following lemma shows that this constraint allows us to restrict our attention to commutative idempotent semirings, which considerably simplifies the structural analysis.
	\begin{lemma}
		\label{lem:ide}
		Let $\mathcal R$ be a commutative semiring that is a $\mathbb T$-algebra. Then $\mathcal R$ is a commutative idempotent semiring.
	\end{lemma}
	\begin{proof}
		Since $\mathcal R$ is a $\mathbb T$-algebra, by definition there exists a unital semiring homomorphism 
		$\varphi: \mathbb T \to \mathcal R$. 
		In particular, $\varphi$ preserves the multiplicative identity, so 
		$1_{\mathcal R}=\varphi(1_{\mathbb T})$.
		
		Note that in the tropical semiring $\mathbb T$, $\oplus$ is idempotent. Hence 
		$1_{\mathbb T}\oplus 1_{\mathbb T}=1_{\mathbb T}$. Applying $\varphi$, we get
		$$
		1_{\mathcal R}\oplus 1_{\mathcal R}
		=\varphi(1_{\mathbb T})\oplus \varphi(1_{\mathbb T})
		=\varphi(1_{\mathbb T}\oplus 1_{\mathbb T})
		=\varphi(1_{\mathbb T})
		=1_{\mathcal R}.
		$$
		
		Now for any $a\in\mathcal R$, using the distributivity of $\otimes$ over $\oplus$,
		$$
		a\oplus a
		=(1_{\mathcal R}\otimes a)\oplus(1_{\mathcal R}\otimes a)
		=(1_{\mathcal R}\oplus 1_{\mathcal R})\otimes a
		=1_{\mathcal R}\otimes a
		=a.
		$$
		Thus $\oplus$ is idempotent, so $\mathcal R$ is a commutative idempotent semiring.
	\end{proof}
	Henceforth, whenever $\mathcal R$ is a commutative semiring that is a $\mathbb T$-algebra, we identify it with a commutative idempotent semiring and write it in the $(\min,+)$ notation.
	\subsection{Extension Algebras for Interpolation}
	We first introduce some definitions needed for interpolation algebra on $\mathcal T_k$.
	\begin{definition2}[congruence and quotient semirings on a commutative semiring]
		Let $\mathcal R$ be a commutative semiring. A \textit{congruence} on $\mathcal R$ is a set $\mathcal S$ of relations on $\mathcal R$ (i.e., $\mathcal S \subseteq \mathcal R \times \mathcal R$) satisfying the following properties, where we write $a \sim b$ to denote the relation $(a,b)$:
		\begin{itemize}
			\item[(E1)] $a \sim a \in \mathcal S$ for all $a \in \mathcal R$.
			\item[(E2)] If $a \sim b \in \mathcal S$, then $b \sim a \in \mathcal S$.
			\item[(E3)] If $a \sim b \in \mathcal S$ and $b \sim c \in \mathcal S$, then $a \sim c \in \mathcal S$.
			\item[(I1)] If $a \sim b \in \mathcal S$ and $c \sim d \in \mathcal S$, then $a\oplus c \sim b\oplus d \in \mathcal S$.
			\item[(I2)] If $a \sim b \in \mathcal S$ and $c \sim d \in \mathcal S$, then $a\otimes c \sim b\otimes d \in \mathcal S$.
		\end{itemize}
		
		For any set $\mathcal S$ of relations on $\mathcal R$, the \textit{congruence generated by} $\mathcal S$ is the intersection of all congruences containing $\mathcal S$, denoted by $\langle \mathcal S \rangle$. It is straightforward to verify that $\langle \mathcal S \rangle$ is indeed a congruence on $\mathcal R$. A congruence $\mathcal E$ is called \textit{finitely generated} if there exists a finite set $\mathcal S$ of relations such that $\mathcal E = \langle \mathcal S \rangle$.
		
		We denote the trivial congruence by $\Delta_{\mathcal R}=\{a\sim a\mid a\in \mathcal R\}$.
		
		Let $\mathcal E$ be a congruence on $\mathcal R$. The \textit{quotient semiring} $\mathcal R/\mathcal E$ is the set of equivalence classes induced by $\mathcal E$. Its addition and multiplication are defined by
		$$\min([a],[b]) = [\min(a,b)], \qquad [a] + [b] = [a+b],$$
		where $[x]$ denotes the equivalence class of $x \in \mathcal R$.
	\end{definition2}
	\begin{definition2}[tropical cyclic polynomial semiring]
		For every positive integer $k$, the \textit{$k$-tropical cyclic polynomial semiring} is the quotient semiring
		$$
		\mathbb T[y]/\langle ky\sim 0\rangle.
		$$
		Each element admits a unique representative of the form
		$$
		\min_{i=0}^{k-1}(a_i+iy),\quad a_i\in\mathbb T,
		$$
		which we call its canonical representation. For an element $\tau\in\mathbb T[y]/\langle ky\sim 0\rangle$ and $0\le r<k$, we denote by $[ry]\tau$ the coefficient of $ry$ in the canonical representative of $\tau$.
		
		Under this representation, the sum is given coefficientwise:
		$$
		\min\left(
		\min_{i=0}^{k-1}(a_i+iy),
		\min_{i=0}^{k-1}(b_i+iy)
		\right)
		=
		\min_{i=0}^{k-1}\left(\min(a_i,b_i)+iy\right),
		$$
		and the product is given by cyclic convolution:
		$$
		\min_{i=0}^{k-1}(a_i+iy)+\min_{i=0}^{k-1}(b_i+iy)
		=
		\min_{i=0}^{k-1}
		\left(
		\min_{j=0}^{k-1}\left(a_j+b_{(i-j)\bmod k}\right)+iy
		\right).
		$$
	\end{definition2}
	We apply the structural Theorem \ref{thm:sec_con_sec} to show that $\mathbb T[y]/\left\langle\left(\operatorname{lcm}_{i=1}^{k}i\right)y\sim 0\right\rangle$ is an interpolation algebra of $\mathcal T_k$, which proves Theorem \ref{thm:int_alg_intro}.
	\begin{theorem}
		\label{thm:int_alg}
		For every positive integer $k$, 
		$$
		\mathbb T[y]/\left\langle\left(\operatorname{lcm}_{i=1}^{k}i\right)y\sim 0\right\rangle
		$$
		is an interpolation algebra of $\mathcal T_k$.
	\end{theorem}
	\begin{proof}
		Let $\mathcal{R}=\mathbb T[y]/\left\langle\left(\operatorname{lcm}_{i=1}^{k}i\right)y\sim 0\right\rangle$, and let $A,B\in\mathcal T_k$ be distinct tropical polynomials. If $A=\infty$ or $B=\infty$, then clearly $A\not\sim_{\mathcal R}B$. Thus, it suffices to consider $A,B\neq\infty$. Write
		$A=\min_{i=0}^{m}(a_i+ix)$ and $B=\min_{i=0}^{t}(b_i+ix)$. If $m\neq t$, then evaluating $A$ and $B$ at a sufficiently negative real number separates them.
		
		Now suppose that $m=t$, and choose an index $p$ such that $a_p\neq b_p$. Set $L=\operatorname{lcm}_{i=1}^{k}i$ and $q=p\bmod L$. Consider subpolynomials
		$C=\operatorname{ms}_{L,q}(A)$ and $D=\operatorname{ms}_{L,q}(B)$. Since the coefficients of $C$ and $D$ corresponding to the index $p$ are different, we have $C\neq D$.
		
		Since $A,B\in\mathcal T_k$, Theorem \ref{thm:sec_con_sec} implies that $C$ and $D$ are convex. And since $C\neq D$, by Lemma \ref{lem:poly_con_ip}, there exists $x_0\in\mathbb R$ such that $C(x_0)\neq D(x_0)$.
		
		Set $\sigma=x_0/L+y\in\mathcal R$. It follows that
		$$
		\begin{aligned}
			[qy]A(\sigma)
			&=
			[qy]\min_{i=0}^{m}\left(a_i+i\left(\frac{x_0}{L}+y\right)\right)\\
			&=
			[qy]\min_{i=0}^{m}\left(a_i+\frac{ix_0}{L}+iy\right)\\
			&=
			\min_{0\leq i\leq m,i\equiv q\pmod L}
			\left(a_i+\frac{i x_0}{L}\right)\\
			&=
			\min_{j=0}^{\lfloor(m-q)/L\rfloor}
			\left(a_{q+jL}+\frac{(q+jL)x_0}{L}\right)\\
			&=
			\min_{j=0}^{\lfloor(m-q)/L\rfloor}
			\left(c_j+jx_0\right)+\frac{qx_0}{L}\\
			&=
			C(x_0)+\frac{qx_0}{L}.
		\end{aligned}
		$$
		Similarly, $[qy]B(\sigma)=D(x_0)+q x_0/L$. Therefore $[qy]A(\sigma)\neq[qy]B(\sigma)$, and consequently $A(\sigma)\neq B(\sigma)$. Thus $A\not\sim_{\mathcal R}B$.
		
		Hence every two distinct polynomials in $\mathcal T_k$ can be separated by evaluation over $\mathcal R$, so $\mathcal R$ is an interpolation algebra of $\mathcal T_k$.
	\end{proof}
	\begin{corollary}
		\label{cor:poly_bf_ei}
		If a multiplicatively closed set $\mathcal S\subseteq\mathbb T[x]$ satisfies the bounded factorization property, then it satisfies the extended identity property. Moreover, there even exists a multiplicatively closed set satisfying the polynomial identity property but not satisfying the bounded factorization property.
	\end{corollary}
	\begin{proof}
		If $\mathcal S$ satisfies the bounded factorization property, then
		$\mathcal S\subseteq\mathcal T_k$ for some positive integer $k$. By Theorem \ref{thm:int_alg}, $\mathcal T_k$ admits an interpolation algebra $\mathcal R$ such that $\mu_{\mathbb T}(\mathcal R)\leq \operatorname{lcm}_{i=1}^{k}i<\infty$.
		Since $\mathcal S\subseteq\mathcal T_k$, $\mathcal R$ is also an interpolation algebra for $\mathcal S$. Therefore, $\mathcal S$ satisfies the extended identity property.
		
		For the second statement, consider the set
		$$
		\mathcal S=\left\{\sum_{i=1}^{m}\min(0,c_i+1+c_i x)\;\middle|\; m\ge 1,\ (c_1,\ldots,c_m)\in \mathbb Z_{>0}^{m}\right\}.
		$$
		As shown in the proof of Lemma \ref{lem:poly_con_uf}, $\mathcal S$ does not satisfy the bounded factorization property. On the other hand, take any $A=\sum_{i=1}^{m}\min(0,c_i+1+c_ix)\in\mathcal S$, and assume $c_1\le \cdots \le c_m$. For each $d\in\mathbb Z_{>0}$, set $\lambda=(d+1)/d$. Evaluating $A$ at $x=-\lambda$ and $x=-\lambda-1/(d(d+1))$ yields
		$$
		A(-\lambda)=\sum_{1\le i\le m,c_i\ge d}\left(c_i+1-c_i \lambda\right),
		$$
		$$
		A\left(-\lambda-\frac{1}{d(d+1)}\right)=\sum_{1\le i\le m,c_i\ge d}\left(\left(1-\frac{1}{d(d+1)}\right)c_i+1-c_i \lambda\right).
		$$
		Thus $$d(d+1)\cdot \left(A(-\lambda)-A\left(-\lambda-\frac{1}{d(d+1)}\right)\right)=\sum_{1\le i\le m,c_i\ge d} c_i.$$
		Denote this quantity by $\tau_d$. The number of occurrences of $d$ in the sequence $c$ can be recovered by
		$\sum_{i=1}^{m}[c_i=d]=(\tau_d-\tau_{d+1})/d.$
		Since $c$ is non-decreasing, these counts uniquely determine the entire sequence $c$, and hence the polynomial $A$. Moreover, this reconstruction uses only the evaluation values of $A$. Therefore $\mathcal S$ satisfies the polynomial identity property.
	\end{proof}
	Next, we establish a lower bound on the generating rank of any interpolation algebra for the set of tropical polynomials of bounded degree. The key idea is to view the action of all elements of an interpolation algebra on a generating set as left multiplication by tropical matrices, as in representation theory. This then allows us to interpret the problem as a shortest path problem and analyze it using techniques from graph theory.
	
	We show that this lower bound is tight by again appealing to the tropical cyclic polynomial semiring. It is noteworthy that the tropical cyclic polynomial semiring continues to play a central role in interpolation.
	\begin{theorem}
		\label{thm:int_alg_neg}
		For a positive integer $n$, let
		$\mathcal S=\{A\in\mathbb T[x]\mid\deg A\leq n\}$.
		Then every interpolation algebra $\mathcal R$ of $\mathcal S$ satisfies
		$\mu_{\mathbb T}(\mathcal R)\geq\lfloor n/2\rfloor+1$.
		Moreover, there exists an interpolation algebra $\mathcal R$ of
		$\mathcal S$ satisfying
		$\mu_{\mathbb T}(\mathcal R)=\lfloor n/2\rfloor+1$.
	\end{theorem}
	\begin{proof}
		Suppose, for contradiction, that $\mathcal R$ is an interpolation algebra
		of $\mathcal S$ with
		$m=\mu_{\mathbb T}(\mathcal R)\leq\lfloor n/2\rfloor$.
		Then $2m\leq n$. Since $\mathcal R$ has generating rank $m$, we can pick a generating set $\mathcal G=\{g_1,\ldots,g_m\}$ of $\mathcal R$ over $\mathbb T$. There exist $d_1,\ldots,d_m\in\mathbb T$ such that
		$0=\min_{j=1}^{m}(d_j+g_j)$.
		
		Fix $\sigma\in\mathcal R$. For every $1\leq j\leq m$, choose
		$c_{j,1},\ldots,c_{j,m}\in\mathbb T$ such that
		$g_j+\sigma=\min_{k=1}^{m}(c_{j,k}+g_k)$.
		For every nonnegative integer $t$, let
		$\mathcal P_t=\{(p_0,\ldots,p_t)\mid p_i\in\{1,\ldots,m\}\}$ be the set of paths, and define
		the weight of path $p=(p_0,\ldots,p_t)\in\mathcal P_t$ by
		$w(p)=d_{p_0}+\sum_{i=0}^{t-1}c_{p_i,p_{i+1}}+g_{p_t}$.
		An induction on $t$ gives
		$$
		t\sigma=\min_{j=1}^{m}(d_j+g_j)+t\sigma=\min_{p\in\mathcal P_t}\left(w(p)\right).
		$$
		
		Consider the two distinct tropical polynomials
		$A=\min_{i=0}^{2m}(ix)$ and
		$B=\min_{0\leq i\leq2m,i\neq m}(ix)$.
		Since $2m\leq n$, both polynomials belong to $\mathcal S$.
		We show that $A(\sigma)=B(\sigma)$ for every $\sigma\in \mathcal R$.
		
		Fix a path $p=(p_0,\ldots,p_m)\in\mathcal P_m$ with $w(p)\neq\infty$.
		By the pigeonhole principle, there exist $0\leq a<b\leq m$ such that
		$p_a=p_b$. Put $q=b-a$ and
		$\Delta=\sum_{i=a}^{b-1}c_{p_i,p_{i+1}}$.
		Define
		$$
		\begin{aligned}
			p^-&=(p_0,\ldots,p_a,p_{b+1},\ldots,p_m),\\
			p^+&=(p_0,\ldots,p_a,p_{a+1},\ldots,p_b,p_{a+1},\ldots,p_b,p_{b+1},\ldots,p_m).
		\end{aligned}
		$$
		Since $p_a=p_b$, both are valid paths,
		with $p^-\in\mathcal P_{m-q}$ and
		$p^+\in\mathcal P_{m+q}$. Since $p^-$ and $p^+$ are obtained from $p$ by deleting and by duplicating a cycle of weight $\Delta$, their weights satisfy
		$w(p^-)=w(p)-\Delta$ and $w(p^+)=w(p)+\Delta$. Therefore,
		$$
		\min\left(w(p^-),w(p^+)\right)=\min\left(w(p)-\Delta,w(p)+\Delta\right)=w(p)-|\Delta|\leq w(p).
		$$
		Since $0\le m-q<m<m+q\le 2m$, both $w(p^+)$ and $w(p^-)$ appear in $B(\sigma)$, which implies $B(\sigma)\le \min(w(p^+),w(p^-))\le w(p)$. Hence $B(\sigma)\le \min_{p\in\mathcal P_m,w(p)\neq \infty} w(p)=\min_{p\in\mathcal P_m} w(p)=m\sigma$. Consequently, $A(\sigma)=\min(B(\sigma),m\sigma)=B(\sigma)$.
		
		Since $A(\sigma)=B(\sigma)$ for every $\sigma\in \mathcal R$, we have
		$A\sim_{\mathcal R}B$, contradicting that $\mathcal R$ is an
		interpolation algebra of $\mathcal S$. Therefore,
		$\mu_{\mathbb T}(\mathcal R)\geq\lfloor n/2\rfloor+1$.
		
		For the second statement, let $m=\lfloor n/2\rfloor+1$ and set
		$\mathcal R=\mathbb T[y]/\langle my\sim 0\rangle$. For every $A\in \mathcal S$, since $n<2m$, every polynomial
		in $\operatorname{MS}_m(A)$ has degree at most $1$, and is therefore
		convex. The same argument as in the proof of
		Theorem \ref{thm:int_alg} shows that $\mathcal R$ is an interpolation algebra
		of $\mathcal S$. Moreover, $\{0,y,\ldots,(m-1)y\}$ is a generating set of $\mathcal R$ over $\mathbb T$, so
		$\mu_{\mathbb T}(\mathcal R)\leq m$. Together with the lower bound
		proved above, this gives $\mu_{\mathbb T}(\mathcal R)=m=\lfloor n/2 \rfloor+1$.
	\end{proof}
	\begin{corollary}
		\label{cor:int_alg_neg}
		For every positive integer $k$, every interpolation algebra $\mathcal R$ of $\mathcal T_k$ satisfies $\mu_{\mathbb T}(\mathcal R)\ge \lfloor k/2\rfloor+1$. Moreover, if $k\ge 2$, then every interpolation algebra $\mathcal R$ of $\mathcal C_k$ satisfies $\mu_{\mathbb T}(\mathcal R)=\infty$.
	\end{corollary}
	\begin{proof}
		Since $\{A\in\mathbb T[x]\mid \deg A\le k\}\subseteq \mathcal T_k$, the first statement follows directly from Theorem \ref{thm:int_alg_neg}.
		
		For the second statement, suppose $k\ge 2$ and let $\mathcal R$ be an interpolation algebra of $\mathcal C_k$ with $m=\mu_{\mathbb T}(\mathcal R)<\infty$. Define $A=\min_{i=0}^{2m}(ix),B=\min_{0\le i\le 2m,i\ne m}(ix)$.
		Then $\operatorname{cgap}(A),\operatorname{cgap}(B)\le 2\le k$, so $A,B\in\mathcal C_k$. In the proof of Theorem \ref{thm:int_alg_neg}, it is shown that $A\sim_{\mathcal R} B$ whenever $\mu_{\mathbb T}(\mathcal R)=m$, contradicting the assumption that $\mathcal R$ is an interpolation algebra of $\mathcal C_k$. This proves the second statement.
	\end{proof}
	\begin{corollary}
		\label{cor:poly_ei_uf}
		$\mathbb T[x]$ satisfies neither the extended identity property nor the unique factorization property.
	\end{corollary}
	\begin{proof}
	    Since $\mathcal C_2\subseteq \mathbb T[x]$, Corollary \ref{cor:int_alg_neg} implies that every interpolation algebra $\mathcal R$ of $\mathbb T[x]$ satisfies $\mu_{\mathbb T}(\mathcal R)=\infty$. Therefore, $\mathbb T[x]$ does not have the extended identity property.
		
		Lemma \ref{lem:poly_con_uf} directly implies that $\mathbb T[x]$ fails the unique factorization property.
	\end{proof}
	\begin{lemma}
		\label{lem:eva_int}
		Let $\mathcal S\subseteq\mathbb T[x]$ be a multiplicatively closed set, and let $\mathcal R$ be an interpolation algebra for $\mathcal S$. Denote by $\mathcal F$ the set of all functions from $\mathcal R$ to itself. Then the evaluation map $\operatorname{ev}_{\mathcal R}:\mathcal S\to\mathcal F$, defined by $\left(\operatorname{ev}_{\mathcal R}(A)\right)(x_0)=A(x_0)$ for every $x_0\in\mathcal R$, is injective and preserves multiplication. More precisely, for all $A,B\in\mathcal S$, we have
		$$
		\operatorname{ev}_{\mathcal R}(A+B)
		=
		\operatorname{ev}_{\mathcal R}(A)+\operatorname{ev}_{\mathcal R}(B),
		$$
		where $\operatorname{ev}_{\mathcal R}(A)+\operatorname{ev}_{\mathcal R}(B)$ is defined pointwise by
		$$
		\left(\operatorname{ev}_{\mathcal R}(A)+\operatorname{ev}_{\mathcal R}(B)\right)(x_0)
		=
		\operatorname{ev}_{\mathcal R}(A)(x_0)+\operatorname{ev}_{\mathcal R}(B)(x_0),\quad x_0\in \mathcal R.
		$$
	\end{lemma}
	\begin{proof}
		Let $A,B\in\mathcal S$ with $A\neq B$. Since $\mathcal R$ is an interpolation algebra for $\mathcal S$, there exists $x_0\in\mathcal R$ such that $A(x_0)\neq B(x_0)$. Hence $\operatorname{ev}_{\mathcal R}(A)\neq \operatorname{ev}_{\mathcal R}(B)$, proving that $\operatorname{ev}_{\mathcal R}$ is injective.
		
		Now write $A=\min_{j=0}^{m}(a_j+jx)$ and $B=\min_{k=0}^{t}(b_k+kx)$. For every $x_0\in\mathcal R$, we have
		$$
		\begin{aligned}
			(A+B)(x_0)
			&=\min_{i=0}^{m+t}
			\left(
			\min_{j+k=i}
			(a_j+b_k)+ix_0
			\right)\\
			&=\min_{0\leq j\leq m,0\leq k\leq t}
			\left(a_j+jx_0+b_k+kx_0\right)\\
			&=\min_{j=0}^{m}(a_j+jx_0)
			+\min_{k=0}^{t}(b_k+kx_0)\\
			&=A(x_0)+B(x_0).
		\end{aligned}
		$$
		Therefore, $\operatorname{ev}_{\mathcal R}(A+B)=\operatorname{ev}_{\mathcal R}(A)+\operatorname{ev}_{\mathcal R}(B)$.
	\end{proof}
	\begin{lemma}
		For a positive integer $k$, let $L=\operatorname{lcm}_{i=1}^{k}i$, set $\mathcal S=\mathcal T_k$, $\mathcal R=\mathbb T[y]/\langle Ly\sim 0\rangle$, and let $\mathcal F$ be the set of all functions from $\mathcal R$ to itself. Suppose that $\operatorname{ev}_{\mathcal R}:\mathcal S\to\mathcal F$ is the evaluation map defined in Lemma \ref{lem:eva_int}. Then the interpolation map $\operatorname{in}_{\mathcal R}:\operatorname{Im}(\operatorname{ev}_{\mathcal R})\to\mathcal S$, defined by
		$$
		[dx]\left(\operatorname{in}_{\mathcal R}(f)\right)
		=
		\sup_{r\in\mathbb R}\left( [(d\bmod L)y]f(r+y) - dr \right),\quad d\in\mathbb N,
		$$
		is a left inverse of $\operatorname{ev}_{\mathcal R}$.
	\end{lemma}
	\begin{proof}
		Let $A\in\mathcal T_k$. If $A=\infty$, the statement is immediate. Otherwise, write $A=\min_{i=0}^{n}(a_i+ix)$ and set $f=\operatorname{ev}_{\mathcal R}(A)$. Write $d=q+sL$ with $0\leq q<L$ and $s\in\mathbb N$. Consider the subpolynomial $C=\operatorname{ms}_{L,q}(A)$. By Theorem \ref{thm:sec_con_sec}, $C$ is convex. As in the proof of Theorem \ref{thm:int_alg}, we have $[qy]f(r+y)=[qy]A(r+y)=C(Lr)+qr$. For each $s\in \mathbb N$, let $c_s=[sx]C$ denote the coefficient of $x^{\otimes s}$ in $C$.
		
		Define $\tau=\sup_{r\in\mathbb R}\left( [qy]f(r+y) - dr \right)$. We claim that $\tau=c_s$. If $C=\infty$, then $a_e=\infty$ for all $0\le e\le n$ with $e\equiv q\pmod L$; hence $[qy]f(r+y)=\infty$ for all $r\in\mathbb R$, and therefore $\tau=\infty$, proving the claim. Now suppose $C\ne\infty$, and let $p=L(C)$. Since $C$ is convex, Lemma \ref{lem:conv_equiv} implies that $\operatorname{csupp}(C)=(p,p+1,\ldots,m)$, so all finite coefficients occur exactly for indices $i\in[p,m]$.
		
		Let $g(r)=\min_{i=p}^{m}\left(c_i+(i-s)Lr\right).$ Observe that
		$$
		\begin{aligned}
			\tau
			&=\sup_{r\in\mathbb R}\bigl(C(Lr)+qr-dr\bigr) \\
			&=\sup_{r\in\mathbb R}\left(\min_{i=0}^{m}(c_i+iLr)+qr-dr\right) \\
			&=\sup_{r\in\mathbb R}\left(\min_{i=p}^{m}\bigl(c_i+(i-s)Lr\bigr)\right) \\
			&=\sup_{r\in\mathbb R} g(r).
		\end{aligned}
		$$
		
		If $s<p$, then letting $r\to\infty$ makes $g(r)\to\infty$, so $\tau=\infty=c_s$. If $s>m$, then letting $r\to-\infty$ gives $g(r)\to\infty$, and again $\tau=\infty=c_s$.
		
		It remains to handle the case $p\le s\le m$. By choosing $i=s$ in the inner minimum, we obtain
		$$\tau=\sup_{r\in\mathbb R}\left(\min_{i=p}^{m}\bigl(c_i+(i-s)Lr\bigr)\right)
		\le \sup_{r\in\mathbb R}\bigl(c_s+(s-s)Lr\bigr)=c_s.$$
		Thus $\tau\le c_s$.
		
		For the reverse inequality, since $p\le s\le m$, we have $s\in\operatorname{csupp}(C)$. Lemma \ref{lem:env} allows us to choose $\lambda\in\partial \hat f_C(s)$. Then for every $p\le i\le m$,
		$$c_s=\hat f_C(s)\le \hat f_C(i)+\lambda(s-i)\le c_i+\lambda(s-i).$$
		Rearranging gives
		$c_i+(i-s)L\cdot(-\lambda/L)\ge c_s.$
		Therefore $g(-\lambda/L)\ge c_s$, and consequently $\tau\ge c_s$. Combining the two inequalities yields $\tau=c_s$.
		
		Therefore, $$[dx](\operatorname{in}_{\mathcal R}(\operatorname{ev}_{\mathcal R}(A)))=\tau=c_s=[dx]A,$$ which completes the proof.
	\end{proof}
	\begin{remark}
		The preceding lemma provides an explicit left inverse for the evaluation map in the special case where $\mathcal R$ is a tropical cyclic polynomial semiring. For a general interpolation algebra $\mathcal R$, however, even if $\operatorname{ev}_{\mathcal R}$ is injective, which guarantees the existence of a left inverse, it is not clear how to construct one explicitly. It remains an interesting problem to construct an explicit interpolation map for $\operatorname{ev}_{\mathcal R}$ in the general setting.
	\end{remark}
	\subsection{Extension Algebras for Decomposition}
	In this subsection, we show that tropical decomposition width remains invariant under flat commutative semiring extension. The proof proceeds by translating the problem into a system of tropical congruence equations, and then applying the weak tropical Nullstellensatz from \cite{ar17}. We first present the definitions needed for this theorem.
	\begin{definition2}[tropical decomposition width over a commutative semiring]
		Let $\mathcal R$ be a commutative semiring that is a $\mathbb T$-algebra, where the natural embedding $\varphi$ is injective. Let $A\in \mathcal R[x]$ be a non-$\infty$ polynomial. A \textit{tropical polynomial decomposition} of $A$ over $\mathcal R$ is an equation
		$$
		A=\sum_{i=1}^{N} B^{(i)},\quad B^{(i)}\in \mathcal R[x].
		$$
		The \textit{tropical decomposition width} of $A$ over $\mathcal R$, denoted by $\operatorname{tdw}_{\mathcal R}(A)$, is the smallest nonnegative integer $k$ such that $A$ admits such a decomposition with $\deg B^{(i)}\le k$ for all $1\leq i\leq N$. In particular, we set $\operatorname{tdw}_{\mathcal R}(\infty)=0$.
	\end{definition2}
	\begin{definition2}[flatness over tropical semiring]
		Let $\mathcal R$ be a commutative semiring that is a $\mathbb T$-algebra. A congruence $\mathcal E$ on $\mathcal R$ is called \textit{flat over} $\mathbb T$ if for all $a,b\in \mathbb T$ and $c\in \mathcal R$,
		$$a+c \sim b+c \in \mathcal E \quad\Longrightarrow\quad a=b\text{ or } c\sim\infty \in \mathcal E.$$
		$\mathcal R$ itself is called \textit{flat over} $\mathbb T$ if the trivial congruence $\Delta_{\mathcal R}$ is flat.
	\end{definition2}
	\begin{definition2}[congruence and variety]
		Let $n$ be a positive integer and set $\mathcal R = \mathbb T[x_1,\ldots,x_n]$. For any subset $S\subseteq \mathbb T^n$, we define the \textit{congruence} of $S$ as the set of relations
		$$E(S)=\{f\sim g\mid f(a)=g(a)\text{ for all }a\in S\}.$$
		Similarly, for any set $\mathcal S$ of relations on $\mathcal R$, we define the \textit{variety} of $\mathcal S$ by
		$$V(\mathcal S)=\{a\in \mathbb T^n \mid f(a)=g(a)\text{ for all }f\sim g\in \mathcal S\}.$$
	\end{definition2}
	\begin{theorem}[weak tropical Nullstellensatz, \cite{ar17}, Theorem 2]
		\label{thm:weak_trop}
		Let $\mathcal E$ be a finitely generated congruence on $\mathbb T[x_1,\ldots,x_n]$. Then $V(\mathcal E)=\varnothing$ if and only if there exist $f\in \mathbb T[x_1,\ldots,x_n]$ and $r\in \mathbb T$ such that
		\begin{itemize}
			\item the constant term of $f$ is not $\infty$,
			\item $r\neq 0$,
			\item $f \sim f+r \in \mathcal E$.
		\end{itemize}
	\end{theorem}
	We then prove Theorem \ref{thm:tdw_inv_intro}, which establishes the rigidity of tropical polynomials under decomposition via extension algebras.
	\begin{theorem}
		Let $\mathcal R$ be a commutative semiring that is a flat $\mathbb T$-algebra. Then for every tropical polynomial $A$,
		$$
		\operatorname{tdw}_{\mathcal R}(A)=\operatorname{tdw}(A).
		$$
	\end{theorem}
	\begin{proof}
		Let $\varphi:\mathbb T\to\mathcal R$ be the natural embedding. Suppose $a,b\in\mathbb T$ satisfy $\varphi(a)=\varphi(b)$. Then $\varphi(a)+0=\varphi(a)=\varphi(b)=\varphi(b)+0$. By flatness of $\mathcal R$, either $a=b$ or $0=\infty$. As $0\ne\infty$, we conclude $a=b$. Thus $\varphi$ is injective, which implies that $\operatorname{tdw}_{\mathcal R}(A)$ is well-defined.
		
		If $A\in \mathbb T$, then $\operatorname{tdw}_{\mathcal R}(A)=\operatorname{tdw}(A)=0$. Otherwise set $q=\operatorname{tdw}_{\mathcal R}(A)\geq 1$ and choose a tropical polynomial decomposition
		$A=\sum_{i=1}^{N} \widetilde B^{(i)}$
		over $\mathcal R$ with $\deg \widetilde B^{(i)}\le q$ for all $1\leq i\leq N$. Write
		$A=\min_{j=0}^{n}(a_j+jx)$ and $\widetilde B^{(i)}=\min_{j=0}^{s_i}(\tilde b_j^{(i)}+jx)$. In particular, we set $a_j=\infty$ for all $j>n$.
		
		For each $1\leq i\leq N$, introduce variables $c^{(i)}_0,\ldots,c^{(i)}_{s_i}$, and let
		$
		\mathcal R_2=\mathbb T\left[ c^{(1)}_0,\ldots,c^{(1)}_{s_1},\ldots,c^{(N)}_0,\ldots,c^{(N)}_{s_N}\right].
		$
		For every $0\le d\le \sum_{i=1}^{N}s_i$, define the polynomial
		$$
		h_d=\min_{(p_1,\ldots,p_N)\in \mathbb Z^N,0\le p_i\le s_i,\sum_{i=1}^{N}p_i=d}
		\left(\sum_{i=1}^{N} c^{(i)}_{p_i}\right)
		$$
		in $\mathcal R_2$. Let $\mathcal S=\{h_d \sim a_d \mid 0\le d\le \sum_{i=1}^{N}s_i\}$ be a set of relations on $\mathcal R_2$, and let $\mathcal E=\langle \mathcal S\rangle$.
		
		We claim that $V(\mathcal E)\neq\varnothing$. Suppose otherwise. Since $\mathcal E$ is finitely generated, Theorem \ref{thm:weak_trop} provides $f\in \mathcal R_2$ and $r\in \mathbb T$ with constant term of $f$ not $\infty$, $r\neq 0$, and $f\sim f+r\in\mathcal E$.
		
		Set $\tau=\left(\tilde b^{(1)}_0,\ldots,\tilde b^{(1)}_{s_1},\ldots,\tilde b^{(N)}_0,\ldots,\tilde b^{(N)}_{s_N}\right)$, and define a congruence on $\mathcal R_2$ by $\mathcal E_2=\{\, f\sim g \mid f(\tau)=g(\tau)\,\}$. Since $A=\sum_{i=1}^{N}\widetilde B^{(i)}$, for every $0\le d\le \sum_{i=1}^{N}s_i$ we have
		$$
		h_d(\tau)
		=
		\min_{(p_1,\ldots,p_N)\in \mathbb Z^N,0\le p_i\le s_i,\sum_{i=1}^{N}p_i=d}
		\left(\sum_{i=1}^{N}\tilde b^{(i)}_{p_i}\right)
		=
		a_d.
		$$
		Hence $h_d\sim a_d\in \mathcal E_2$ for all $0\leq d\leq \sum_{i=1}^{N}s_i$, so $\mathcal S\subseteq \mathcal E_2$. By definition of generated congruence, $\mathcal E\subseteq \mathcal E_2$.
		
		Thus $f\sim f+r\in \mathcal E_2$, which gives $f(\tau)=f(\tau)+r$. Let $u\in \mathbb R$ be the constant term of $f$. Then since $u\neq\infty$, we have $f(\tau)\le u<\infty$. Since $\mathcal R$ is flat, we conclude that $0=r$, contradicting $r\neq 0$. Therefore $V(\mathcal E)\neq\varnothing$.
		
		Now pick 
		$\sigma=\left(b^{(1)}_0,\ldots,b^{(1)}_{s_1},\ldots,b^{(N)}_0,\ldots,b^{(N)}_{s_N}\right)\in V(\mathcal E).$
		For each $1\leq i\leq N$, let $v_i$ be the largest index $p$ such that $b^{(i)}_p\neq\infty$, and define
		$B^{(i)}=\min_{j=0}^{v_i}(b^{(i)}_j+jx)$. Since $\sigma\in V(\mathcal E)$, for every $0\le d\le \sum_{i=1}^{N}s_i$ we have
		$$
		\min_{(p_1,\ldots,p_N)\in \mathbb Z^N,0\le p_i\le s_i,\sum_{i=1}^{N}p_i=d}
		\left(\sum_{i=1}^{N}b^{(i)}_{p_i}\right)
		=h_d(\sigma)=a_d.
		$$
		Thus $A=\sum_{i=1}^{N} B^{(i)}$. Moreover, $\deg B^{(i)}=v_i\le s_i=\deg \widetilde B^{(i)}\le q$ for all $1\leq i\leq N$, so $\operatorname{tdw}(A)\le q=\operatorname{tdw}_{\mathcal R}(A)$.
		
		Conversely, since $\mathcal R$ is a $\mathbb T$-algebra, every tropical polynomial decomposition of $A$ over $\mathbb T$ is also a decomposition over $\mathcal R$. Hence $\operatorname{tdw}_{\mathcal R}(A)\le \operatorname{tdw}(A)$. Combining the two inequalities gives $\operatorname{tdw}_{\mathcal R}(A)=\operatorname{tdw}(A)$.
	\end{proof}
	\begin{remark}
		The preceding argument does not rely on any special properties of the tropical congruence equation associated with decomposition width. It applies equally to other parameters based on the decomposition shape, which determine a tropical equation system solely from the decomposition structure, for example, the Barvinok rank (\cite{ms15}, Definition 5.3.1) in tropical matrix theory. This shows that such parameters also admit no refined decompositions under extension algebras.
	\end{remark}
	\section*{Open Problems}
	We conclude this paper with a selection of open problems motivated by our findings.
	\begin{itemize}
		\item What is the exact complexity of deciding whether $\operatorname{tdw}(A)\leq k$ for a tropical polynomial $A$ of degree $n$? Can we show that this problem is solvable in time $f(k)n^{O(1)}$, or alternatively, that it is W[1]-hard?
		\item Can the time complexities of Multiple-Sequence $(\min,+)$ Convolution and Multiple-Choice Knapsack be improved? Our algorithm does not use the powerful additive-combinatorial techniques that have led to recent improvements for $0$-$1$ Knapsack \cite{jin24,bri24}. Since these problems are harder than $0$-$1$ Knapsack, incorporating these techniques would require a more sophisticated algorithm design.
		\item Can our results be extended from the univariate setting to multivariate tropical polynomials? Multivariate tropical polynomials appear to have connections to lattice theory and discrete convex analysis.
		\item Can the gap between the lower and upper bounds on the generating rank of interpolation algebras for $\mathcal T_k$ be closed? We currently have a linear lower bound and an upper bound of $e^{(1+o(1))k}$. The converse construction in Theorem \ref{thm:sec_con_sec} shows that the latter is tight for tropical cyclic polynomial semirings, so any polynomial upper bound would require a different algebraic construction.
	\end{itemize}
	\section*{AI Disclosure}
	All original concepts, theorem statements, proof ideas, and figures presented as contributions of this paper were developed by the authors. We used GPT-5.6 to assist in verifying proof details, including checking for typographical errors in mathematical notation and identifying corner cases we had initially overlooked. All proofs were written and subsequently revised by the authors themselves. The authors take full responsibility for the correctness and integrity of this work.
	\section*{Acknowledgments}
	We thank Ce Jin for helpful discussions.
	
\end{document}